\documentclass{lmcs}
\pdfoutput=1

\usepackage{lastpage}
\renewcommand{\dOi}{14(3:24)2018}
\lmcsheading{}{1--\pageref{LastPage}}{}{}%
{Aug.~07,~2017}{Sep.~28,~2018}{}


\usepackage[UKenglish]{babel}
\usepackage{comment}
\usepackage{enumitem}
\usepackage{amsmath,amsthm}
\usepackage{hyperref}
\usepackage[olditem,oldenum]{paralist}  %% for inparaenum
\usepackage{amsmath}
\usepackage{amsthm}
\usepackage{amssymb,stmaryrd}

\usepackage{basic}
\usepackage{ambients}

\begin{document}

\title{Equational Reasonings in Wireless Network Gossip Protocols\rsuper*}
\titlecomment{{\lsuper*}An extended abstract appeared in Proc.\ 37th FORTE, LNCS 10321:139-155, Springer, 2017~\cite{LMT17}.}

\author[R.~Lanotte]{Ruggero Lanotte\rsuper{a}}
\address{\lsuper{a}Dipartimento di Scienza e Alta Tecnologia, Universit\`a degli Studi delll'Insubria, Como, Italy}
\email{ruggero.lanotte@uninsubria.it}
\email{simone.tini@uninsubria.it}

\author[M.~Merro]{Massimo Merro\rsuper{b}}
\address{\lsuper{b}Dipartimento di Informatica, Universit\`a degli Studi di Verona, Verona, Italy}
\email{massimo.merro@univr.it}

\author[S.~Tini]{Simone Tini\rsuper{a}} 
% \address{\lsuper{a}Dipartimento di Scienza e Alta Tecnologia,  Universit\`a degli Studi dell'Insubria, Como, Italy}
% \email{simone.tini@uninsubria.it}
 
\def\titlerunning{Weak simulation quasimetric in a gossip scenario}
\def\authorrunning{R. Lanotte,  M. Merro,  S. Tini}

\begin{abstract}
\emph{Gossip protocols} have been proposed as a robust and efficient method
for disseminating information throughout large-scale networks.
In this paper, we propose a compositional analysis technique to study formal probabilistic models of gossip protocols expressed 
in  a simple probabilistic timed process calculus for  \emph{wireless sensor networks}. We equip the calculus  with a  \emph{simulation theory}  to compare probabilistic protocols that have similar behaviour up to a certain
tolerance. The theory is used to prove a number of  algebraic laws which revealed to be very effective to estimate the performances of \emph{gossip networks\/}, with and without \emph{communication collisions\/}, and \emph{randomised gossip networks\/}.
Our simulation theory is an asymmetric variant of the \emph{weak bisimulation metric} that maintains most of the properties of the original definition. However, our asymmetric version is particularly suitable to reason on protocols in which the systems under consideration are not approximately equivalent, as in the case of gossip protocols. 
\end{abstract}

\maketitle

%%%%%%%%%%%%%%%%%%%%%% INTRO %%%%%%%%%%%%%%%%%%%%%%%%%%%%%%%%%%%%%%%
\section*{Introduction}
\label{sec:intro}

\emph{Wireless sensor networks} (WSNs) are (possibly large-scale) networks of sensor nodes deployed in strategic areas to gather data.
 Sensor nodes collaborate using wireless communications with an asymmetric many-to-one data transfer model. Typically, they send  sensed events or data, by a specific communication protocol, to a specific node called the sink node or base station, which collects the requested information.
WSNs are primarily designed for monitoring environments that humans cannot easily reach (\emph{e.g.}, motion, target tracking, fire detection, chemicals, temperature); they are also used as embedded systems (\emph{e.g.}, environmental monitoring, biomedical sensor engineering, smart homes) or mobile applications (\emph{e.g.}, when attached to robots or vehicles).
In wireless sensor networks, sensor nodes are usually battery-powered, and
the energy expenditure of sensors has to be wisely managed by their architectures and protocols to prolong the overall network lifetime. 
\emph{Energy conservation} is thus one of the major issues in sensor network communications.

\emph{Flooding} is a traditional robust algorithm that delivers data packets in a network from a source to a destination. In the context of  WSNs, flooding means that all nodes receiving a message propagate it to all  neighbours by broadcast. This causes unnecessary retransmissions, which increase the number of collisions and deprive sensors of valuable battery power. Therefore, flooding algorithms may not be suitable in the context of dense networks like wireless sensor networks.

\emph{Gossiping\/} is an algorithm that addresses some critical problems of flooding overhead. Basically, it is based on the repeated probabilistic exchange of information between two members~\cite{gossip}. The goal of gossip protocols is to reduce the number of retransmissions by making some of the nodes discard the message instead of forwarding it.
Gossip protocols exhibit both  \emph{nondeterministic} and \emph{probabilistic} behaviour.
Nondeterminism arises as they deal with distributed networks in which the activities of individual nodes occur nondeterministically. As to the probabilistic
behaviour, nodes  are required to forward packets with a pre-specified gossip probability $p_{\mathrm{gsp}}$. When a node receives a message, rather than immediately retransmitting it as in flooding, it relies on the probability $p_{\mathrm{gsp}}$ to determine whether or not to retransmit.  The main benefit is that when  $p_{\mathrm{gsp}}$ is sufficiently large, the entire network receives the broadcast message with very high probability, even though only a subset of nodes has forwarded the message.

 Most of the analyses of protocols for large-scale WSNs are usually based on discrete-event simulators (\emph{e.g.}, ns-2, Opnet and Glomosim).  However, different simulators often support different models of the MAC physical-layer yielding different results, even for simple systems (see, \emph{e.g.},
\cite{flooding-unreliable}).
 In principle, as noticed by Bakhashi et al.~\cite{Fokkink2007},
 due to their relatively simple structure, gossip protocols lend themselves very well to formal analysis, in order to predict their behaviour with high confidence. Formal analysis techniques
are supported by (semi-)automated tools.
For instance, \emph{probabilistic model checking} 
has been used  in the context of gossip protocols by Fehnker and Gao~\cite{ansgar2006} and Kwiatkowska et al.~\cite{Marta2008}
 to provide both an exhaustive search of all possible behaviours of the system, and exact, rather than approximate, quantitative results. 
 Of course, a trade-off inevitably exists.
Model checking suffers from the so-called state explosion problem,
meaning that the state space of a specified system grows exponentially with respect to its number of components. 
 On the other hand, simulation-based approaches are scalable to much larger and more complex systems, at the expense of exhaustiveness and numerical accuracy.

Among the formal techniques for the analysis of complex systems, 
\emph{behavioural semantics\/}, such as preorders and equivalences, 
provide formal instruments to compare the behaviour of 
\emph{probabilistic systems}~\cite{JYL01,Den15}. 
Preorders allow us to determine whether a system can mimic the stepwise behaviour of another system; whereas equivalences require a sort of mutual simulation between two systems. The most prominent examples are the  \emph{simulation preorder} and the \emph{bisimulation equivalence}~\cite{LS91,SL95}.  Since probability values usually originate from observations (statistical sampling) or from requirements (probabilistic specification), both preorders and equivalences are only partially satisfactory as they can only say whether a system can mimic another one. Any tiny variation of the probabilistic behaviour of a system 
 will break the preorder (or the equivalence) without any further information. 
In practice, many system implementations can only approximate the system specification; thus, the verification of such implementations requires  appropriate instruments to measure the quality of the approximation.
To this end, bisimulation-based \emph{metric semantics}~\cite{GJS90,DGJP04,BW05,DCPP06,AFS09} have been successfully employed to formalise the \emph{behavioural distance} between two systems.

\subsection*{Contribution.}
The \emph{goal} of this paper is to provide a compositional analysis technique to study probabilistic models of gossip protocols in the context of WSNs. We adopt a simple process calculus for wireless sensor networks proposed by Lanotte and Merro~\cite{LaMe11} and called \emph{Probabilistic Timed Calculus of Wireless Systems} (\cname{}). We then introduce the notion of \emph{weak simulation quasimetric} 
as the asymmetric counterpart of the \emph{weak bisimulation metric} \cite{DJGP02}, and the quantitative analogous of the \emph{weak simulation preorder}~\cite{BKHH02,BHK04}. We use the definition of weak simulation quasimetric to derive  a definition of \emph{weak simulation with tolerance} $p \in [0,1]$, written $\simtol{p}$, to compare probabilistic protocols; being $0$ and $1$  the minimum and the maximum distance,  respectively. 

The compositionality of a behavioural semantics with respect to the parallel operator is fundamental when reasoning on large-scale systems.
We prove that our weak simulation with tolerance matches one of the most restrictive compositional criteria for (bi)simulation metrics, 
namely \emph{non-expansiveness}~\cite{DJGP02,DGJP04}  (also known as 1-non-extensiveness~\cite{BBLM13b}).

We then rely on our simulation with tolerance  to develop an  \emph{algebraic theory} for \cname{} that allows us to compose the tolerances of different sub-networks with different behaviours. In particular,  for a gossip network $\mathrm{GSP}_{p_{\mathrm{gsp}}}$, which transmits with gossip probability
$p_{\mathrm{gsp}}$, we are able to estimate an \emph{upper bound} 
of the tolerance $p_{\mathrm{ok}}$ needed to simulate a 
non-probabilistic network $\mathrm{DONE}$
whose target nodes successfully receive the message:
\[
\mathrm{DONE}
 \q \sqsubseteq_{p_{\mathrm{ok}}} \;
\mathrm{GSP}_{\!p_{\mathrm{gsp}}}
    .
\]
To this end, we prove and apply a significant number of  algebraic laws, 
whose application can be \emph{mechanised}, to evaluate the performances
of gossip networks in terms of message delivery at destination. 
Our analysis takes into consideration \emph{communication collisions}, 
to determine the impact of collisions on the performance results, 
and investigate \emph{randomised gossip protocols}, in which messages may be broadcast in different instants of time according to some uniform distribution probability to mitigate the impact of communication collisions.

\subsection*{Outline.} In Section~\ref{sec:syntax}, we recall syntax and probabilistic operational semantics of \cname{}. 
 In Section~\ref{sec:simulation}, we provide the notions of weak simulation quasimetric and weak simulation with tolerance. In Section~\ref{sec:gossip-nocollisions},  we start our analysis on collision-free gossip protocols. In Section~\ref{sec:gossip-collisions}, we  extend our analysis to collision-prone gossip protocols.  In Section~\ref{sec:random}, we study randomised gossip protocols. 
In Section~\ref{sec:ConRelWor}, we draw conclusions and discuss related 
 and  future  work.

%%%%%%%%%%%%%%%%%%%%%%%%%%%%%%%%%%%%%%
%%%%                                                                                               %%%%
%%%%          A PROBABILISTIC TIMED PROCESS CALCULUS    %%%%
%%%%                                                                                                %%%%    
%%%%%%%%%%%%%%%%%%%%%%%%%%%%%%%%%%%%%%

\section{A probabilistic timed process calculus}
\label{sec:syntax}
 \begin{table}[t]
    \begin{math}
      \begin{array}{@{\hspace*{8mm}}rcl@{\hspace*{40mm}}l@{\hspace*{6mm}}rcl@{\hspace*{7mm}}l}
\multicolumn{4}{l}{\textit{Networks:}}\\
          M, N
          & \Bdf & \zero
          & \mbox{empty network} \\[1pt]
          & \Bor &
          {M_1} | {M_2}
          & \mbox{parallel composition} \\[1pt]
          & \Bor &
         \node \nu  n P
          & \mbox{node}\\[1pt]
       & \Bor &
         \dummyN
          & \mbox{stuck network}
         \\[3pt]
\multicolumn{4}{l}{\textit{Processes:}}
          \\
          P,Q
          & \Bdf &   \nil & \mbox{termination}\\[1pt]
& \Bor &   \bcast  {u}  C
          & \mbox{broadcast} \\[1pt]
          & \Bor &
            \rcvtime x {C} {D} & \mbox{receiver with timeout}\\[1pt]
           &\Bor &
           \tau.{C} & \mbox{internal  }\\[1pt]
          & \Bor &
          \delay \sigma C
          & \mbox{sleep} \\[1pt]
          & \Bor &
          X & \mbox{process variable}\\[1pt]
                 & \Bor &
           \fix X P
          & \mbox{recursion}\\[3pt]
\multicolumn{4}{l}{\textit{Probabilistic Choice:}}
          \\
          C,D
          & \Bdf & \bigoplus_{i\in I}p_i {:} P_i
         \end{array}
    \end{math}
\caption{Syntax}
\label{tab:syntax}
\end{table}

In Table~\ref{tab:syntax}, we provide the syntax of the  \emph{Probabilistic Timed Calculus of Wireless Systems}, \cname{}, in a two-level structure, a lower one for \emph{processes}, ranged over by letters $P$, $Q$, and $R$, and an upper one for \emph{networks\/}, ranged over by letters $M$, 
$N$,  and $O$.  We use letters $m,n, \ldots$ for logical names, Greek symbols $\mu,\nu,\nu_1,\ldots$ for \emph{sets of names}, $x,y,z$ for \emph{variables}, $u$ for \emph{values}, and $v$ and $w$ for \emph{closed values}, \emph{i.e.}, values that do not contain variables. 
Then, we use $p_i$ for probability weights, hence $p_i \in [0,1]$.

A network in \cname\ is a (possibly empty) collection of nodes (which represent devices) running in parallel and using a unique common radio channel to communicate with each other. Nodes are also unique; \emph{i.e.}, a node $n$ can occur in a network only once.  
All nodes are assumed to have the same transmission range (this is a quite common assumption in models for ad hoc networks~\cite{Guide-adhoc}).
The communication paradigm is \emph{local broadcast\/}; only nodes located in the range of the transmitter may receive data.
We write $\node \nu  n P$ for a node named $n$ (the device network address) executing the sequential process $P$. 
The set $\nu$ contains (the names of) the  neighbours of  $n$.
Said in other words, $\nu$ contains all nodes laying in the transmission cell of $n$  (except $n$).  
In this manner, we model the network topology.
Notice that the network topology could have been represented using some kind of restriction operator \`a la CCS over node names; we preferred our notation to keep at hand the neighbours of a node.
Our wireless networks have a fixed topology.  Moreover, nodes cannot be created or destroyed. Finally, we write $\dummyN$ to denote a deadlocked network which  prevents the execution of parallel components. This is a fictitious network which is introduced for technical convenience in the definition of our metrics 
(Definition~\ref{def:simulation_quasimetric}). 

Processes are sequential and live inside the nodes. The symbol $\nil$ denotes terminated processes. The sender process $\bcast v C$ broadcasts the value $v$, the continuation being $C$. 
The process $\rcvtime x {C} {D}$ denotes a receiver with timeout.
Intuitively, this process either receives a value $v$ from a neighbour node in the current time interval, and then continues as $C$ in which the variable $x$ is instantiated with $v$, or it idles for one time unit and then continues as $D$, if there are no senders in the neighbourhood. 
The process $ \tau.C$ performs an internal action and then continues as $C$.
The process $\sigma.C$ models sleeping for one time unit.
In processes of the form $\sigma.D$  and $\rcvtime x C D$ the occurrence of $D$ is said to be \emph{time-guarded}. The process $\fix X P$ denotes \emph{time-guarded recursion\/}, as all occurrences of the process variable $X$ may only occur time-guarded in $P$.

The construct $\bigoplus_{i\in I}p_i {:} P_i$ denotes \emph{probabilistic choice}, where $I$ is a  \emph{finite\/}, \emph{non-empty} set of indexes, and $p_i \in (0, 1]$ denotes the probability to execute the process $P_i$, with  $\sum_{i \in I}p_i=1$. Notice that, as  in~\cite{Dengetal2008}, in order to simplify the operational semantics,  probabilistic choices occur always underneath prefixing.

In processes of the form $\rcvtime x C {D}$ the variable $x$ is bound in $C$.
Similarly, in process $\fix X P$ the process variable $X$ is bound in $P$.
This gives rise to the standard notions of \emph{free (process) variables} and \emph{bound (process) variables} and \emph{$\alpha$-conversion}.
We identify processes and networks up to $\alpha$-conversion. 
A process 
is said to be \emph{closed} if it does not contain free (process) variables.
We always work with closed processes: 
 the absence of free variables is trivially maintained at run-time.
We write ${\subst v x}P$ (\emph{resp.}, ${\subst v x}C$) for the substitution of the variable $x$ with the value $v$ in the process $P$ (\emph{resp.}, probabilistic choice $C$).
Similarly, we write ${\subst P X}Q$ (\emph{resp.}, ${\subst P X}C$)   for the substitution of the process variable $X$ with the process $P$ in $Q$ (\emph{resp.}, in $C$).

We report some \emph{notational conventions\/}. With an abuse of notation, we will write \emph{$\rcv x C$}  as an abbreviation for  $\fix X {\rcvtime x C {(1{\colon}X)}}$, where the process
variable $X$ does  not occur in $C$. $\prod_{i \in I}M_i$ denotes the parallel composition of all $M_i$, for $i \in I$. We identify $\prod_{i \in I}M_i$ with $ \zero$, if $I =\emptyset$.
We  write $\pchoice p {P_1} {P_2}$   for the probabilistic  process $p{\colon}P_1 \oplus (1{-}p){\colon}P_2$. We identify  $1 {\colon} P$ with $P$. We write $\bcastzero v$ as an abbreviation for $\bcast v 1{\colon}\nil$. For $k>0$ we write $\sigma^{k}.P$ as an abbreviation for $\sigma \ldots \sigma.P$, where prefix $\sigma$ appears $k$ times.  Given a network $M$, $\nds M$ returns the names of $M$. 
If $m \in \nds{M}$, the function $\cell {m,M}$ returns the set of the neighbours of $m$ in $M$. 
Thus, for $M = M_1 | \node \nu  m {P} | M_2$ it holds that $\cell {m,M}=\nu$.
We write  $\cell{M}$ for $\bigcup_{m \in \nds{M}}\cell{m, M}$.

\begin{defi}[Structural congruence]\label{def:struc-cong}
\emph{Structural congruence\/}, written $\equiv$, is defined as the smallest equivalence relation over networks, preserved by parallel composition, which is a commutative monoid with respect to parallel composition with neutral element $\zero$, and for which $\node \nu  n {\fix X P} \equiv \node \nu  n {{\subst {\fix X P} X}P}$.
\end{defi}

The syntax presented in Table~\ref{tab:syntax} allows us to derive networks which are somehow ill-formed. 
With the following definition we rule out networks: 
\begin{inparaenum}[(i)]
\item
where nodes can be neighbours of themselves;  
\item 
 with different nodes with the same name;  
\item
 with non-symmetric neighbouring relations. 
\end{inparaenum} 
Finally, in order to guarantee clock synchronisation among nodes, we
require network connectivity.
\begin{defi}[Well-formedness]
\label{def:well-formedness} 
A network $M$ is said to be \emph{well-formed} if
\begin{itemize}
\item whenever $M \equiv M_1 \mid \node {\nu}  {m} {P_1}$ it holds that $m \not\in \nu$;
\item whenever $M \equiv M_1 \mid \node {\nu_1}  {m_1} {P_1} \mid \node {\nu_2} {m_2} {P_2}$ it holds that $m_1 \neq m_2$;
\item whenever $M \equiv N \mid \node {\nu_1}  {m_1} {P_1} \mid \node {\nu_2}  {m_2} {P_2}$ we have $m_1 \in \nu_2$ iff $m_2 \in \nu_1$;
\item for all $m,n \in \nds{M}$ there are $m_1, \ldots , m_k \in \nds{M}$, such that  $m = m_1$, $n = m_k$, and $m_i \in \cell{{m_{i{+}1}},M}$ for $1 \leq i \leq k{-}1$.
\end{itemize}
\end{defi}
Henceforth, we will always work with well-formed networks.

\subsection{Probabilistic labelled transition semantics}
\label{sec:lts}
\label{sec:operational_semantics}
Along the lines of \cite{Dengetal2008},
 we propose an \emph{operational semantics\/} for  \cname{} associating with each network a graph-like structure representing its possible evolutions:
we use a generalisation of labelled transition systems that includes probabilities. Below, we report the necessary mathematical machinery.
\begin{defi} 
A (discrete) \emph{probability sub-distribution} over a set $S$ is a function $\Delta \colon S \to [0,1]$ with  $\sum_{s \in S}\Delta(s) \in (0 , 1]$.
We denote $\sum_{s \in S}\Delta(s)$ by $\size{\Delta}$.
The \emph{support} of a probability sub-distribution $\Delta$ is given by $\lceil \Delta \rceil = \{ s \in S : \Delta(s) > 0 \}$.
We write ${\mathcal D}_{\mathrm{sub}}(S)$, ranged over by $\Delta$, $\Theta$, $\Phi$, for the set of all finite-support probability sub-distributions over $S$.
A probability sub-distribution $\Delta \in {\mathcal D}_{\mathrm{sub}}(S)$ is said to be a \emph{probability distribution} if $\size{\Delta} =1$.
With ${\mathcal D}(S)$ we denote the set of all probability distributions over $S$ with finite support.
For any $s \in S$, the \emph{point (Dirac) distribution at $s$\/}, denoted $\overline{s}$, assigns probability $1$ to $s$ and $0$ to all others elements of $S$, so that $\lceil \overline{s} \rceil = \{ s \}$. 
\end{defi}

Let $I$ be a finite index such that (i) $\Delta_i$ is a sub-distribution in ${\mathcal D}_{\mathrm{sub}}(S)$ for each $i \in I$, and (ii) 
 $p_i \geq 0$ are probabilities such that $\sum_{i\in I}p_i \in (0,1]$. 
Then, the probability sub-distribution $\sum_{i \in I}p_i \cdot \Delta_i  \in {\mathcal D}_{\mathrm{sub}}(S)$ is defined as
$(\sum_{i \in I}p_i  \cdot \Delta_i)(s) \, \deff \, \sum_{i\in I}p_i \cdot \Delta_i(s)$, 
for all $s \in S$. We  write a sub-distribution as $p_1 \cdot \Delta_1 + \ldots + p_n \cdot 
\Delta_n$ when the index set $I$ is $\{ 1, \ldots , n \}$.  In the following,  we will often  write  $\sum_{i \in I}p_i  \Delta_i$ instead of $\sum_{i \in I}p_i  \cdot \Delta_i$.
Definition~\ref{def:struc-cong} and Definition~\ref{def:well-formedness}  generalise to sub-distributions in ${\mathcal D}_{\mathrm{sub}}(\cname)$. 
Given two probability sub-distributions $\Delta$ and $\Theta$, we write $\Delta \equiv \Theta$ if $\Delta([M]_{\equiv})=\Theta([M]_{\equiv})$ for all equivalence classes $[M]_{\equiv} \subseteq \cname$ of $\equiv$. 
Moreover, a probability sub-distribution $\Delta \in {\mathcal D}_{\mathrm{sub}}(\cname)$ is said to be well-formed if its support contains only well-formed networks.

We follow the standard probabilistic generalisation of labelled transition systems~\cite{SegalaPhD95}.
\begin{defi}[Probabilistic LTS]
A \emph{probabilistic labelled transition system} 
(pLTS) is a triple $\langle S, {\mathcal {L}}, \rightarrow \rangle$ where 
\begin{inparaenum}[(i)]
\item 
 $S$ is a set of states; 
\item
 ${\mathcal {L}}$ is a set of transition labels; 
\item
 $\rightarrow$ is a labelled transition relation contained in $S \times {\mathcal L} \times \mathcal{D}(S)$.
\end{inparaenum}
\end{defi}
The operational semantics of \cname\ is given by a particular pLTS
$\langle \cname, {\mathcal L}, \rightarrow \rangle$, 
where
${\mathcal L} = \{ \sndto m v \mu,$ $\rcva m v,$ $\tau, \sigma \}$
contains the labels to denote broadcasting, reception, internal actions and time passing, respectively. The definition of the relations $\transSim{\lambda}$, for $\lambda \in {\mathcal L}$, is given by the SOS rules in Table~\ref{tab:net}. 
Some of these rules use an obvious notation for distributing parallel composition over a  sub-distribution: $({\Delta} \mid {\Theta})(M) = \Delta(M_1)\cdot\Theta(M_2)$ if  $M = {M_1} | {M_2}$; $({\Delta} \mid {\Theta})(M) = \zero$ 
otherwise.

Furthermore, the definition of the labelled transition relation
 relies on a semantic 
interpretation of (nodes containing) probabilistic processes 
in terms of probability distributions.

\begin{defi}
For any probabilistic choice $\bigoplus_{i \in I} p_i {\colon} P_i$ 
over a finite index set $I$, 
we write $\sem {\node \mu n {\bigoplus_{i \in I} p_i {\colon} P_i}}$ 
to denote the probability distribution
$\sum_{i \in I}p_i\cdot \overline{\node \mu n {P_i}}$. 
\end{defi}

\begin{table}[t] 
       \begin{displaymath}
\begin{array}{l@{\hspace*{4mm}}l}

\Txiom{Snd}{-}{\node{\nu}{m}{\bcast v C} \transSim{m!v \vartriangleright \nu}
                                         \sem{\node{\nu}{m}{C}}} & 
\Txiom{Rcv}{ m \in \nu}
                                 {\node{\nu}{n}{\rcvtime x C {D}} \transSim{m?v}
                                  \sem{\node{\nu}{n}{{\subst v x}C}}}
\\[15pt]
\Txiom{Rcv-$\zero$}{-}{\zero \transSim{\rcva m v} \overline{\zero}}
& 
\Txiom{RcvEnb}{\neg (m \in \nu  \: \wedge \: \rcvrs{P})\quad \wedge \quad m\neq n}
{\node \nu n P \transSim{\rcva m v} \overline{\node \nu n P}}
\\[15pt]
\Txiom{RcvPar}{M \transSim{m?v} \Delta  \quad  \quad N \transSim{m?v} \Theta}
                                        {M \mid N \transSim{m?v} \Delta | \Theta}
& 
\Txiom{Bcast}{M \transSim{m!v \vartriangleright \nu} \Delta  \quad N \transSim{m?v} \Theta
 \quad {\scriptstyle {\mu {:=} \nu \setminus \nds{N}}}}
                                 {M \mid N \transSim{m!v \vartriangleright \mu}
 \Delta | \Theta}
\\[15pt]
\Txiom{Tau}{-}{{\node \nu m  {\tau.C}  } \transSim{\tau} \sem{\node \nu m {C}}}
& 
\Txiom{TauPar}{M \transSim{\tau} \Delta \quad N \not\equiv \dummyN | N' }{M|N \transSim{\tau} \Delta | \overline{N}}
\\[15pt]
\Txiom{$\sigma$-$\zero$}{-}{\zero \transSim{\sigma} \overline{\zero}}
& 
\Txiom{Timeout}{-}
                                {\node{\nu}{n}{ {\rcvtime x C {D}}} \transSim{\sigma}
                                 \sem{\node{\nu}{n}{D}}}
\\[15pt]
\Txiom{$\sigma$-nil}{-}
{\node \nu  n \nil \transSim{\sigma}
\overline{\node \nu  n \nil}}
& 
\Txiom{Sleep}{-}{\node{\nu}{n}{\sigma.C} \transSim{\sigma}
                                          \sem{\node{\nu}{n}{C}}}
\\[15pt]
\Txiom{$\sigma$-Par}{M \transSim{\sigma} \Delta  \quad  \quad N \transSim{\sigma} \Theta}
                                     {M|N \transSim{\sigma} \Delta | \Theta}
& 
\Txiom{Rec}{\node \nu  n {{\subst {\fix X P} X}P} \transSim{\lambda}
                   \Delta }
                  {\node \nu  n {\fix X P} \transSim{\lambda}
                   \Delta }
\end{array}
     \end{displaymath}
    \caption{Probabilistic Labelled Transition System}
    \label{tab:net}
  \end{table}
Let us comment on the rules of Table~\ref{tab:net}. 
In rule~\rulename{Snd} a node $m$ broadcasts a message $v$ to its neighbours $\nu$, the continuation being the probability distribution associated to 
$\node \nu m C$.  In the label $\sndto m v \nu$ the set $\nu$ denotes the neighbours of $m$: all nodes in $\nu$ that are currently listening will receive the message $v$.
In rule~\rulename{Rcv} a node $n$ gets a message $v$ from a neighbour node $m$,  
the continuation being the probability distribution associated to $\node \nu n {{\subst v x}C}$. 
If no message is received in the current time interval then the node $n$ will 
continue according to  $D$, as specified in rule~\rulename{Timeout}. 
Rules \rulename{Rcv-\zero} and \rulename{RcvEnb} serve to model reception enabling for synchronisation purposes. 
For instance, rule \rulename{RcvEnb} regards nodes $n$ which are not involved in transmissions originating from $m$. 
This may happen either because the two nodes are out of range (\emph{i.e.}, $m \not\in \nu$) or because $n$ is not willing to receive ($\rcvrs P$ is a boolean predicate that returns true if $\node \nu n P \equiv \node \nu n {\rcvtime x C D}$, for some $x$, $C$, $D$).
In both cases,  node $n$ is not affected by the transmission.
In rule~\rulename{RcvPar} we model the composition of two networks receiving the same message from the same transmitter.
Rule~\rulename{Bcast} models the propagation of messages on the broadcast channel.
Note that  we lose track of those transmitter's neighbours that are in $N$. 
Rule \rulename{Tau} models internal computations in a single node. 
The rule \rulename{TauPar} propagates internal computations on parallel components different from the $\dummyN$ network.
Rules \rulename{$\sigma$-nil} and \rulename{$\sigma$-\zero} are straightforward as both terms $\zero$ and $\node \nu n \nil$ do not prevent time-passing.
Rule \rulename{Sleep} models sleeping for one time unit.
Rule \rulename{$\sigma$-Par} models time synchronisation between parallel components (we recall that  receiver processes can timeout only if there are no senders in their neighbourhood). 
Rule \rulename{Rec} is standard.
For rules \rulename{Bcast} and \rulename{TauPar}  we assume their symmetric counterparts. Finally, note that the 
semantics of the network $\dummyN$ is different from that of $\zero$:
$\dummyN$ does not perform any action and it prevents 
the evolution of any parallel component.

\subsubsection{Extensional labelled transition semantics}
Our focus is on weak similarities, which abstract away non-observable actions, 
\emph{i.e.}, those actions that cannot be detected by a parallel network. 
The adjective \emph{extensional} is used to stress that here we focus on  
 those activities that  require a contribution of the environment.
To this end,  we extend  Table~\ref{tab:net} by the following two rules:
\[
\Txiom
{ShhSnd}
{M \transSim{\sndto m {v} \emptyset} \Delta }
{M \transSim{\tau} \Delta}
\Q\Q\Q\Q\Q
\Txiom{ObsSnd}
{M \transSim{\sndto m {v} \nu } \Delta \Q \nu \neq \emptyset }
{{M} \transSim{\sndat v \nu} {\Delta}}
\]
Rule \rulename{ShhSnd} models transmissions that cannot be observed because there is no potential receiver outside the network $M$.  
Rule \rulename{ObsSnd} models transmissions that may be observed by those nodes of the environment contained in $\nu$.
Notice that the name of the transmitter is removed from the label.
This is motivated by the fact that receiver nodes do not have  a direct manner to observe the identity of the transmitter. 
On the other hand, a network $M$ performing the action $\rcva m  v$  can be observed by an external node $m$ which  transmits the value $v$ to an appropriate set of nodes in $M$. 
Notice that the action  $\sndat v \nu$ does not propagate over parallel components (there is no rule for that). 
As a consequence,  the rule~\rulename{ObsSnd} can only be applied to the whole network, never in a sub-network.

In the rest of the paper, the metavariable $\alpha$ will range over the following four kinds of actions: $\sndat v \nu$, $\rcva m v$, $\sigma$, $\tau$. 
They denote anonymous broadcast to specific nodes, message reception, time passing, and internal activities, respectively. 

Finally, having defined the labelled transitions that can be performed by
a network, we can easily concatenate these transitions to define
the possible computation traces of a network. 
A \emph{computation trace} for a network $M$
 is a sequence of steps of the form
$M_1  \transSim{\alpha_1} \dots \transSim{\alpha_{n-1}} M_n$ where for any $i$, with $1 \leq i \leq n-1$, we have $M_i \transSim{\alpha_i} \Delta_{i+1}$   such that $M_{i+1}  \in \lceil \Delta_{i+1} \rceil$.

\subsection{Time properties}
\label{sec:properties}
The calculus \cname{} enjoys a number of desirable time properties~\cite{LaMe11} that will be very useful for our purposes.

Proposition~\ref{prop:timedet} formalises the deterministic nature of time passing: a network can reach at most one new distribution by executing the action $\sigma$.
\begin{prop}[Time Determinism]\label{prop:timedet}
Let $M$ be a well-formed network.
If $M \transSim{\sigma} \Delta$ and $M \transSim{\sigma} \Theta$ then $\Delta$ and $\Theta$ are the same.
\end{prop}

The maximal progress property says that sender nodes transmit immediately. 
Said in other words,  the passage of time cannot block transmissions.
\begin{prop}[Maximal Progress]
\label{prop:MaxProg}
Let $M$ be a well-formed network.
If  either $M \equiv \node \nu  m {! \langle v \rangle .C} | N$
or $M \equiv \node \nu  m {\tau.C} | N$
then $M \transSim{\sigma} \Delta$ for no distribution $\Delta$.
\end{prop}

Patience guarantees that a process will wait indefinitely until it can communicate~\cite{HennessyR95}.
In our setting, this means that if neither transmissions nor 
internal actions can fire then it must be possible to execute a $\sigma$-action to let time pass.

\begin{prop}[Patience] \label{prop:pat}
Let $M = \prod_{i\in I} {\node {\nu_i}  {m_i} {P_i}}$ be a well-formed network such that for all $i \in I$ it holds that ${P_i} \neq {\bcast v {C}}$
and $P_i \neq \tau.C$,
then there is a
distribution $\Delta$ s.t.\ $M \transSim{\sigma} \Delta$.
\end{prop}

Finally, our networks satisfy the  \emph{well-timedness} (or \emph{finite variability}) property~\cite{Nicollin1991}. 
Intuitively, only a finite number of instantaneous actions can fire between two contiguous $\sigma$-actions.  In \cname{}, this property holds because recursion is defined to be time-guarded.
\begin{prop}[Well-Timedness]
\label{prop:well-timedness}
For any well-formed network $M$ there is an upper bound $k \in \mathbb{N}$ such that whenever $M \trans{\alpha_1} \! \cdots \! \trans{\alpha_h} N $, with $\alpha_j\not\in \{ \rcva m v, \sigma\}$ for $1 \leq j \leq h$, \nolinebreak then \nolinebreak  $h \leq k$.
\end{prop}

%%%%%%%%%%%%%%%%%%%%

\section{Weak simulation with tolerance}
\label{sec:simulation}
In this section, we introduce \emph{weak simulation quasimetrics\/} as an instrument  to derive a notion of approximate simulation between networks. 
Our goal is to  define a family of relations $\simtol{p}$ over networks, with $p \in [0,1]$, to formalise the concept of \emph{simulation with a tolerance
 $p$\/}.  Intuitively, we wish to write $M \simtol{p} N$ if $N$ can simulate $M$ with a tolerance $p$. Thus, $\simtol{0}$ should coincide with the standard
 weak probabilistic simulation~\cite{BKHH02,BHK04},
 whereas $\simtol{1}$ should be equal to $\cname{} \times \cname$.

In a probabilistic setting, the definition of weak transition is somewhat complicated by the fact that (strong) transitions take processes (in our case networks) to distributions;
consequently if we are to use weak transitions $\TransSim{\hat{\alpha}}$ which abstract away from non-observable actions then we need to generalise transitions, so that they take (sub-)distributions to (sub-)distributions.

For a network $M$ and a distribution $\Delta$, we write $M \urg{\hat{\tau}} \Delta$ if either $M \transSim{\tau} \Delta$ or $\Delta = \overline{M}$.
Then, for $\alpha \neq \tau$, we write $M \transSim{\hat{\alpha}} \Delta$ if $M \transSim{\alpha} \Delta$.
Relation $\transSim{\hat{\alpha}}$ is extended to model transitions from sub-distributions to sub-distributions.
For a sub-distribution $\Delta =\sum_{i \in I}p_i \overline{M_i}$, we write $\Delta \urg{\hat{\alpha}} \Theta$ if there is 
a non-empty set  $J\subseteq I$
 such that $M_j \transSim{\hat{\alpha}} \Theta_j$ for all $j \in J$,
$M_i \ntransSimone[\hat{\alpha}]$, for all $i \in I \setminus J$,  and $\Theta = \sum_{j \in J}p_j  \Theta_j$.
Note that if $\alpha \neq \tau$ then this definition admits that only some networks in the support of $\Delta$ make the $\transSim{\hat{\alpha}}$ transition.
Then, we define the weak transition relation $\TransSimone[\hat{\tau}]$ as the transitive and reflexive closure of $\transSimone[\hat{\tau}]$, \emph{i.e.}, 
$\TransSim{\hat{\tau}} = (\transSim{\hat{\tau}})^{\ast}$, while for $\alpha \neq \tau$ we let $\TransSim{\hat{\alpha}}$ denote $\TransSim{\hat{\tau}} \transSim{\hat{\alpha}} \TransSim{\hat{\tau}}$.

\begin{exa}
Assume a node $\node \nu n {\tau. ({\bcastzero {v_1}}\oplus_{\frac{1}{3}} {\bcastzero {v_2}})}$ that  broadcasts $v_1$ with probability $\frac{1}{3}$, and $v_2$ with probability $\frac{2}{3}$.
By an application of rule \rulename{Tau} we derive the transition
\[ \node \nu n {\tau. ({\bcastzero {v_1}}\oplus_{\frac{1}{3}} {\bcastzero {v_2}})} \; \trans{\tau} \q
\frac{1}{3} \cdot \overline{\node \nu n{ {\bcastzero {v_1}}}}+ \frac{2}{3} \cdot \overline{ \node \nu n{\bcastzero {v_2}}}  . \]
Then, by an application of rule \rulename{Snd} we can derive: 
\begin{itemize}
\item  either $\node \nu n {\tau. ({\bcastzero {v_1}}\oplus_{\frac{1}{3}} {\bcastzero {v_2}})}  \Urg{\widehat{n!{v_1} \vartriangleright \nu}} 
\frac{1}{3} \cdot \overline{ \node \nu n \nil}$, because 
$\node \nu n{{\bcastzero {v_1}}} \trans{n!{v_1} \vartriangleright \nu} \node \nu n \nil$  
\item or
 $\node \nu n {\tau. ({\bcastzero {v_1}}\oplus_{\frac{1}{3}} {\bcastzero {v_2}})}  \Urg{\widehat{n!{v_2} \vartriangleright \nu}} 
\frac{2}{3} \cdot \overline{ \node \nu n \nil}$, because 
$\node \nu n{{\bcastzero {v_2}}} \trans{n!{v_2} \vartriangleright \nu} \node \nu n \nil$. 
\end{itemize}
\end{exa}

In order to define our notion of simulation with tolerance, we adapt the concept of \emph{weak bisimulation metric} of Desharnais et al.'s \cite{DJGP02}. 
In that paper  the behavioural distance between systems is measured by means of suitable \emph{pseudometrics}, namely symmetric functions assigning a numeric value to any pair of systems.
Here, we consider asymmetric variants of pseudometrics, called \emph{pseudoquasimetrics\/}, measuring the tolerance of the \emph{simulation} between networks.

\begin{defi}[Pseudoquasimetric]\label{def:pseudoquasimetric}
A function $d \colon\cname{} \times \cname{}  \to [0,1]$ is a \emph{1-bounded pseudoquasimetric} over \cname{} if
\begin{itemize}	
\item   $d(M,M)= 0$ for all $M \in \cname{}$, and
\item  $d(M,N) \le d(M,O) + d(O,N)$ for all $M,N,O\in \cname{}$ (triangle inequality).
\end{itemize}
\end{defi}

Weak simulation quasimetrics provide the quantitative analogous of the weak simulation game: a networks $M$ is simulated by a network $N$ with tolerance $p$ if each transition $M \transSimone[\alpha] \Delta$ is mimicked by a transition $N \TransSimone[\hat{\alpha}] \Theta$ such that the distribution $\Delta$ is simulated by the distribution $\Theta$ with a tolerance $q$, with $q \leq p$.
As the derivative of a transition is a probability distribution we have to lift 
the notion of pseudometric from \cname{} to distributions over \cname{}.
In \cite{DJGP02}, this lifting is realised by means of linear programs, relying on the symmetry of pseudometrics.
Since pseudoquasimetrics are not symmetric, we need a different technique.
Thus, to this end,  we adopt the notions of matching~\cite{Vil08} (also known as \emph{coupling}) and Kantorovich lifting~\cite{Den09}.

\begin{defi}[Matching]
\label{def_matching}
A \emph{matching} for a pair of distributions 
$(\Delta,\Theta) \in {\mathcal D}(\cname{}) \times {\mathcal D}(\cname{})$ is a distribution $\omega$ in the network product space ${\mathcal D}(\cname{} \times \cname{})$ 
such that: 
\begin{itemize}
\item 
$\sum_{N \in \cname{}} \omega(M,N)=\Delta(M)$, for all $M \in \cname{}$, and 
\item 
 $\sum_{M\in \cname{}} \omega(M,N)=\Theta(N)$, for all $N \in \cname{}$. 
\end{itemize}
We write $\Omega(\Delta,\Theta)$ to denote the set of all matchings for $(\Delta,\Theta)$.
\end{defi}
A matching for $(\Delta,\Theta)$ may be understood as a transportation schedule for the shipment of probability mass from $\Delta$ to $\Theta$ \cite{Vil08}.

\begin{defi}[Kantorovich lifting] \label{def:KantorovichLifting}
\label{def:Kantorovich}
Let $d\colon \cname{} \times \cname{}  \to [0,1]$ be a pseudoquasimetric. The \emph{Kantorovich lifting} of $d$ is the function
$\Kantorovich(d)\colon {\mathcal D}(\cname{} ) \times {\mathcal D}(\cname{}) \to [0,1]$ defined 
as: 
\[
\Kantorovich(d)(\Delta,\Theta) \deff  \min_{\omega \in \Omega(\Delta,\Theta)} \sum_{M,N \in \cname}\omega(M,N) \cdot d(M,N).
\]
\end{defi}
Note that since we are considering only distributions with finite support, the minimum over the set of matchings $\Omega(\Delta,\Theta)$ is well defined. 

\begin{defi}[Weak simulation quasimetric]
\label{def:simulation_quasimetric}
We say that a pseudoquasimetric $d \colon \cname{} \times \cname{} \to [0,1]$ is a \emph{weak simulation quasimetric}
 if for all networks $M,N \in \cname$, with $d(M,N)<1$, whenever $M \transSim{\alpha} \Delta$ there is a distribution $\Theta$ such that
 $N \TransSim{\hat{\alpha}} \Theta$  and $\Kantorovich(d)(\Delta \: , \:  \Theta + (1{-}\size{\Theta}) \overline{\dummyN}) \le d(M,N)$. 
\end{defi}
In the previous definition, if 
$\size{\Theta} < 1$ then, 
 with 
probability $1-\size{\Theta}$,  there is no way to simulate 
the behaviour of any network in the support of $\Delta$ (the special network $\dummyN$ does not perform any action).

As expected, the kernel of a weak simulation quasimetric is a \emph{weak probabilistic simulation} defined along the lines of \cite{DD11} in terms of the  notion of matching.
\begin{defi}
A relation ${\mathcal R} \subseteq \cnamed{} \times \cnamed{}$ is a weak probabilistic simulation iff whenever $M \, {\mathcal R} \, N$ and $M \transSim{\alpha} \Delta$ there exist a transition $N \TransSim{\hat{\alpha}} \Theta$ and a matching $\omega \in \Omega(\Delta,\Theta)$ with $\omega(M',N')>0$ such that $M' \, {\mathcal R} \, N'$. 
\end{defi} 
\begin{prop}
\label{prop_simulazione}
Let $d$ be a weak simulation quasimetric.
The relation $\{ (M, N)  \colon  d(M,N) =0 \} \subseteq \cname{} \times \cname{}$ is a weak probabilistic simulation.
\end{prop} 
A crucial result in our construction process is the existence of the minimal weak simulation quasimetric, which
can be viewed as the asymmetric counterpart of the minimal weak bisimulation metric~\cite{DJGP02}.
\begin{thm}
\label{thm:exists_metric}
There exists a weak simulation quasimetric $\metric$ such that $\metric(M,N) \le d(M,N)$ for all weak simulation quasimetrics $d$ and all networks $M,N \in \cname{}$.
\end{thm}

Now, we have all ingredients to define our 
simulation with tolerance $p$. 
\begin{defi}[Weak simulation with tolerance]
\label{def:simtol}
Let $p \in [0,1]$, we say that \emph{$N$ simulates $M$ with tolerance $p$}, written $M \sqsubseteq_p N$, iff $\metric(M,N) \le p$.
We write $M \simeq_{p} N$ if both $M \simtol{p} N$ and $N \simtol{p} M$. 
\end{defi}

Since the minimum weak simulation quasimetric $\metric$ satisfies the triangle inequality, our simulation relation is trivially  transitive in an additive sense:
\begin{prop}[Transitivity]
\label{prop:transitivity}
$M \simtol{p} N$ and $N \simtol{q} O$ imply $M \simtol{r} O$, with  $r{=}\min(1,p{+}q)$. 
\end{prop}

As expected, if $M \TransSimone[\hat{\tau}] \Delta$ then $M$ can simulate all networks in $\lceil \Delta \rceil$.
\begin{prop}
\label{M_go_to_N}
If $M  \TransSimone[\hat{\tau}] (1{-}q) \overline{N} + q\Delta$,   
for some $\Delta \in {\mathcal D}(\cname{})$, then $N \simtol{q} M$.
\end{prop}

Clearly the transitivity property is quite useful when doing algebraic reasoning. 
However, 
we can derive a better tolerance when concatenating two simulations, if one of them is derived by an application of Proposition~\ref{M_go_to_N}. 
\begin{prop}
\label{lem:O_go_to_N_sim_M}
If $M \simtol{p} N$ and $O \TransSim{ \hat \tau} (1-q) \overline{N} + q \Delta$,
for some $\Delta \in {\mathcal D}(\cname{})$, then 
 $M \simtol{p  (1-q) + q} O$. 
\end{prop}
Intuitively, in the simulation between $M$ and $N$ the tolerance $p$ 
must be weighted by taking into consideration  that $O$ may 
evolve into $N$ with a probability $(1-q)$.

The previous results can be generalised to the case where Proposition~\ref{M_go_to_N} is applied to networks that are reached after one $\sigma$-transition.
\begin{prop}
\label{lem:O_go_to_N_sim_M_conSigma}
Let $M$, $N$ and $O$ be networks in \cname{} such that:  
(i)  
$M \simtol{p} N$;  
(ii)
 $N$ can only perform a transition of the form $N \transSim{\sigma} \overline{N'}$, for some $N'$  such that  $N' \ntransSimone[\tau]$;  
(iii) 
 $O \transSim{\sigma} \overline{O'} \TransSim{ \hat \tau} (1-q) \overline{N'} + q \Delta$,
for some $O'\in \cname{}$ and $\Delta \in {\mathcal D}(\cname{})$. 
Then, $M \; \simtol{p  (1-q) + q} \; O$. 
\end{prop}

In order to understand the intuition behind our weak simulation with tolerance, we report here a few simple \emph{algebraic laws} (recall that $1{\colon}P=P$).
\begin{prop}[Simple algebraic laws] \hfill 
\label{prop:first-laws}   
\begin{enumerate}
\item  \label{random}
$\node {\mu} n {P} 
\: \simtol{1-p} \: 
\node {\mu} n {\tau.(P \oplus_p Q)}$

\item 
$\node \mu n Q 
\, \simtol{r} \, 
\node \mu n {\tau.(\tau.(P \oplus_q Q) \oplus_p R)}$, with $r=(1-p) + pq$

\item  
\label{random3} 
$
\node \mu n {   \tau.(\bigoplus_{i \in I, \, j \in J}\,  (p_i\, q_j) {\colon}  P_j^i  ) }
\, \simtol{0} \, 
\node \mu n {   \tau.(\bigoplus_{i \in I} p_i {\colon}   \tau.(\bigoplus_{j \in J} q_j {\colon} P^i_j)  }
$ 
\item 
$\node \mu n {\bcast{v}{\big(  \tau.(P \oplus_q \tau.P)  \oplus_{p} Q\big)}}
\: \simeq_{0} \: 
\node \mu n {\bcast {v}{(P \,  \oplus_{p} \, \tau.Q)}}
$
\item\label{doppiobcst}
$\node \mu n {\bcast{v}{\bcastzero{w}}} 
\: \simtol{r} \: 
\node \mu n {\tau.\big( \bcast{v}{\tau.(\bcastzero{w} \oplus_q P)} \: \oplus_p \: Q\big)}
$, with $r=1-pq$.
\end{enumerate}
\end{prop}
The first law is straightforward.
The second law is a generalisation of  the first one where the right-hand side must resolve two probabilistic choices in order to simulate the left-hand side.
 The third law is an adaptation of the $\tau$-law $\tau.P =P$ of CCS~\cite{Mil89} in a distributed and probabilistic setting.  Similarly, 
the fourth law  reminds a probabilistic and distributed variant of the  $\tau$-law $a.(\tau.(P+\tau.Q))+a.Q=a.(P+\tau.Q)$. 
The last law gives an example of a probabilistic simulation involving sequences of actions. Here, $pq$ is the probability that the network on the right-hand side simulates the network on the left-hand side by transmitting both $u$ and $v$.
In fact, we have: 
\[
\node \mu n {\tau.( \bcast{v}{\tau.(\bcastzero{w} \oplus_q P)} \oplus_p Q)} \;  \transSim{\tau} \; p \cdot \overline{\node \mu n {\bcast{v}{\tau.(\bcastzero {w} \oplus_q P)}}} + (1-p) \cdot \overline{\node \mu n Q}
\]
with $\node \mu n {\bcast{v}{\bcastzero{w}}} \simtol{1-q}   \node \mu n {\bcast{v}{\tau.(\bcastzero {w} \oplus_q P)}}$;  the result   follows by an application of Proposition~\ref{lem:O_go_to_N_sim_M}.

A crucial property of our  simulation 
is the possibility to reason on parallel networks 
in a \emph{compositional} manner. 
Thus, if  $M_1 \simtol{p_1} N_1$ and $M_2 \simtol{p_2} N_2$ then $M_1 \mid M_2 \simtol{p} N_1 \mid N_2$ for some $p$ depending on $p_1$ and $p_2$; the intuition being that if one fixes the maximal  tolerance $p$ between $M_1 \mid M_2$ and $N_1 \mid N_2$, then there are tolerances $p_i$ between $M_i$ and $N_i$, $i \in \{ 1, 2 \}$, ensuring that the tolerance $p$ is respected.
Following this intuition, several compositional criteria for bisimulation metrics can be found in the literature.
Here, we show that our weak simulation with tolerance matches one of the most restrictive among those studied in \cite{GLT16,GT18}, namely \emph{non-expansiveness}~\cite{DJGP02,DGJP04}  (also known as 1-non-extensiveness~\cite{BBLM13b}) requiring that $p \le p_1+p_2$.

\begin{thm}[Non-expansiveness law]
\label{thm:non_expansiveness}
\label{cor:non_expansiveness}
$M_1 \simtol{p_1} N_1$ and $M_2 \simtol{p_2} N_2$ entails $M_1 \mid M_2 \simtol{r} N_1 \mid N_2$, with $r=\min(1,p+q)$. 
\end{thm}

Another useful property is that a network can be simulated by means 
of a probabilistic choice whenever it is simulated by all components.
\begin{prop}[Additive law]
\label{prop:propagation-general}
\label{comp_failed_in_sim}
Let $M \sqsubseteq_{s_i} \node \mu n {  P_i }  | N$, for 
 all  $i \in I$, with $I$ a finite index set. Then,   
$M \sqsubseteq_{r}  \node \mu n  {\tau.\bigoplus_{i \in I} p_i {\colon} P_i } \mid N$, 
for $r=\sum_{i \in I} p_i s_i$. 
\end{prop}

Finally, we report a number of algebraic laws that will be useful in the next section, when analysing gossip protocols. 
\begin{prop}[Further algebraic laws] 
\label{thm:timing} \
\begin{enumerate}
\item
\label{timing1}
$\node \mu n {\sigma^k.\nil}  \q \simeq_0  \q \node \mu n \nil$
\item
\label{timing2}
$\prod_{i \in I}\node {\mu_{i}} {m_i} {P_i}  
\simeq_{r} 
\prod_{j \in J}\node {\nu_{j}} {n_j} {Q_j}$ entails 
$\prod_{i \in I}\node {\mu_{i}} {m_i} {\sigma.P_i}  
\simeq_{r} 
\prod_{j \in J}\node {\nu_{j}} {n_j} {\sigma.Q_j}$  
\item
\label{timing34}
$\node \mu n {\rcvtime x C D}  \;   \simeq_0 \;  
  \node \mu n {\sigma.D}$, 
if nodes in $\mu$ do not send in the current time interval
\item
\label{timing5} 
$\node \mu n {\rcv x C} 
\; \simeq_0 \; 
\node \mu n {\sigma.\rcv x C}$, 
if nodes in $\mu$ do not send in the current time interval
\item
\label{law:no-trans}
$
\node {\mu} m {\nil} | \prod_{i \in I} \node {\mu_{i}} {n_i} {P_i} 
\; \simtol{0} \;
\node {\mu} m {\tau.(\bcastzero{v} \oplus_p \nil)} | \prod_{i \in I} \node {\mu_{i}} {n_i} {P_i}
$
if $\mu \subseteq  \bigcup_{i \in I}n_i$,
and for all $n_i \in \mu$  it holds that $P_i\neq \rcvtime x C D$.
\end{enumerate}
\end{prop}
Please, notice that all laws but last one hold with respect to weak bisimilarity; the last one holds with respect to weak similarity.
Intuitively:
\begin{inparaenum}[(1)]
\item 
$\nil$ does not prevent time passing;
\item 
equalities are preserved underneath $\sigma$ prefixes;
\item 
\label{spiega_laws_tre}
receivers timeout if there are not senders around; 
\item 
 this law is an instance of (\ref{spiega_laws_tre}) because $\rcv x C$ is an 
abbreviation for $\fix X \timeout{\rcv x C}{(1{:}X)} \equiv   \timeout{\rcv x C}(1{:}{\rcv x C})=  \timeout{\rcv x C}({\rcv x C})$ ($1{:}P$ is abbreviated with $P$);
\item
 broadcasting evolves silently if there are no receivers in the neighbourhood. 
\end{inparaenum}

%%%%%%%%%%%%%%%%%%%%%%%%%%%%%%%%%%%%%

\section{Gossiping without collisions}

\label{sec:gossip-nocollisions}
The baseline model for our case study is gossiping without communication 
collisions, where all nodes  are perfectly synchronised (relying on some
clock synchronisation algorithm for WSNs). %%~\cite{SBK05}).
For the sake of clarity, communication proceeds in synchronous \emph{rounds}: 
a node can transmit or receive only one message per round.
In our implementation, rounds are separated by $\sigma$-actions.

The processes involved in the protocol are the following:
\[
\mathsf{snd} \langle u \rangle_{p_{\mathrm{g}}}
\!\!\deff \! \!\tau.(\bcastzero u \oplus_{p_{\mathrm{g}}} \nil)
\quad \quad
\fwd_{p_{\mathrm{g}}} \!\!\deff \!\!{\rcv x { \mathsf{resnd} \langle x \rangle_{p_{\mathrm{g}}}}}
\quad \quad 
\mathsf{resnd} \langle u \rangle_{p_{\mathrm{g}}}\!\!\deff \!\!\sigma.\mathsf{snd} \langle u \rangle_{p_{\mathrm{g}}} \, . 
\]
A sender  broadcasts with a gossip probability $p_{\mathrm{g}}$,
whereas a forwarder  rebroadcasts the received value, in the subsequent round, with the same probability.

We apply our simulation theory to develop algebraic reasonings
 on \emph{message propagation\/}.
As an introductory example, let us consider a fragment of a network with two
sender nodes, $m_1$ and $m_2$, and two forwarder nodes, $n_1$ and $n_2$ which 
are both neighbours of $m_1$ and $m_2$. 
Then, assuming $\nu = \{n_1,n_2\}$ and $\{m_1,m_2\} \subseteq \nu_1 \cap \nu_2$, the following algebraic law holds: 

\begin{equation}
\label{eq:2s2rNoColl}
\begin{array}{c}
\node {\nu} {m_1} {\sndv_{p_1}} \q \big| \q \node {\nu} {m_2} {\sndv_{p_2}} \q \big| \q \node {\nu_1} {n_1} {\fwd_{q_1}} \q \big| \q \node {\nu_2} {n_2} {\fwd_{q_2}} \\[3pt] 
\simtoll{r}  \\[3pt]
\node {\nu} {m_1} {\nil} \q \big| \q \node {\nu} {m_2} {\nil} \q \big| \q \node {\nu_1} {n_1} {\wsndv_{q_1}} \q \big| \q \node {\nu_2} {n_2} {\wsndv_{q_2}} 
\end{array}
\end{equation}

with tolerance $r=(1-p_1)(1-p_2)$. Here, the network on the left-hand-side evolves to the Dirac distribution of the network on the right-hand-side by performing a sequence of $\tau$-actions. 
Thus, the algebraic law follows by an application of Proposition~\ref{M_go_to_N} %%
being $1-r$ the probability that the message $v$ is broadcast to both 
forwarder nodes $n_1$ and $n_2$.
To explain in detail how the sequence of $\tau$-actions are derived, we first observe that both senders $m_1$ and $m_2$ start with the $\tau$-action resolving the probabilistic choice on the transmission. Thus, by an application of rule \rulename{Tau} we can derive the transitions
${\textstyle \node {\nu} {m_i} {\sndv_{p_i}} \transSim{\tau} p_i  \cdot \overline{\node {\nu} {m_i} { \bcastzero v } } + (1-p_i) \cdot \overline{\node {\nu} {m_i} { \nil }}}$, for $i \in \{ 1,2 \}$, which are propagated to the whole network by means of the rule \rulename{TauPar}. At this point, when composing these two transmissions,  there are three possible cases.
\begin{itemize}
\item
With probability $p_1 \cdot p_2$ both senders $m_1$ and $m_2$ decide to broadcast (the order of transmission is not relevant). This means that with probability $p_1 \cdot p_2$ the execution of the two $\tau$-transitions above leads to the distribution  
${\textstyle  \overline{\node {\nu} {m_1} {\bcastzero v}   |   \node {\nu} {m_2} {\bcastzero v}   | \node {\nu_1} {n_1} {\fwd_{q_1}}  |  \node {\nu_2} {n_2} {\fwd_{q_2}}}}$. 
Now, if we suppose that $m_1$ transmits  first, then  we have the transition:  
\[
\begin{array}{c}
\node {\nu} {m_1} {\bcastzero v} \; | \;  \node {\nu} {m_2} {\bcastzero v} \; | \;  \node {\nu_1} {n_1} {\fwd_{q_1}} \; | \;  \node {\nu_2} {n_2} {\fwd_{q_2}}
\\[4pt] 
\transSim{\tau} \\[4pt]
\overline{\node {\nu} {m_1} {\nil} \; | \; \node {\nu} {m_2} {\bcastzero v} \; | \; \node {\nu_1} {n_1} {\wsndv_{q_1}} \; | \; \node {\nu_2} {n_2} {\wsndv_{q_2}} 
 }  . 
\end{array}
\]
\enlargethispage{.3\baselineskip}
When also node $m_2$ broadcasts its value  $v$ then we will have:
\[
\begin{array}{c}
\node {\nu} {m_1} {\nil} \, | \, \node {\nu} {m_2} {\bcastzero v} \, | \, \node {\nu_1} {n_1} {\wsndv_{q_1}} \, | \, \node {\nu_2} {n_2} {\wsndv_{q_2}} \\[4pt]
\transSim{\tau} \\[4pt]
\overline{ \node {\nu} {m_1} {\nil} \, | \,  \node {\nu} {m_2} {\nil} \,  | \,  \node {\nu_1} {n_1} {\wsndv_{q_1}} \,  | \,  \node {\nu_2} {n_2} {\wsndv_{q_2}} } 
 . 
\end{array}
\]
\item
With probability $p_1 \cdot (1-p_2) + (1-p_1) \cdot p_2$ exactly one of the two senders decides to broadcast. This means that the $\tau$-transitions due to the resolution of the probabilistic choice lead either to the  distribution 
${\textstyle  \overline{\node {\nu} {m_1} {\bcastzero v}  | \node {\nu} {m_2} {\nil} |  \node {\nu_1} {n_1} {\fwd_{q_1}}  |  \node {\nu_2} {n_2} {\fwd_{q_2}}}}$ 
or to the distribution 
$ {\textstyle \overline{\node {\nu} {m_1} {\nil}  |  \node {\nu} {m_2} {\bcastzero v} |  \node {\nu_1} {n_1} {\fwd_{q_1}}  |  \node {\nu_2} {n_2} {\fwd_{q_2}}}}$. 
In the former case, \emph{i.e.}, when $m_1$ transmits,  we have the transition
\[
\begin{array}{c}
\node {\nu} {m_1} {\bcastzero v} \, | \, \node {\nu} {m_2} {\nil} \, | \, \node {\nu_1} {n_1} {\wsndv_{q_1}} \, | \, \node {\nu_2} {n_2} {\wsndv_{q_2}} \\[4pt]
\transSim{\tau} \\[4pt]
\overline{ \node {\nu} {m_1} {\nil} \,  | \,  \node {\nu} {m_2} {\nil} \,  | \,  \node {\nu_1} {n_1} {\wsndv_{q_1}} \, | \, \node {\nu_2} {n_2} {\wsndv_{q_2}}}
\end{array}
\]
while in the latter case we have the transition 
\[
\begin{array}{c}
\node {\nu} {m_1} {\nil} \, | \,  \node {\nu} {m_2} {\bcastzero v} \,  | \,  \node {\nu_1} {n_1} {\wsndv_{q_1}} \,  | \,  \node {\nu_2} {n_2} {\wsndv_{q_2}} 
\\[4pt] 
\transSim{\tau} \\[4pt]
\overline{ \node {\nu} {m_1} {\nil} \,  | \,  \node {\nu} {m_2} {\nil} \,  | \,  \node {\nu_1} {n_1} {\wsndv_{q_1}} \, | \,  \node {\nu_2} {n_2} {\wsndv_{q_2}}}
 \, . 
\end{array}
\]
\item
With probability $(1-p_1)(1-p_2)$ both senders decide not to transmit.
In this case,  it is simply not possible to reach the desired network in which both  nodes $n_1$ and $n_2$ forward the message, as they will never receive anything to forward. 
\end{itemize}
Summarising, the network ${\textstyle \node {\nu} {m_1} {\sndv_{p_1}} |   \node {\nu} {m_2} {\sndv_{p_2}}   |  \node {\nu_1} {n_1} {\fwd_{q_1}}  |   \node {\nu_2} {n_2} {\fwd_{q_2}}}$ has probability $1-(1-p_1)(1-p_2)$ to reach (the Dirac distribution of) $ \node {\nu} {m_1} {\nil}  |  \node {\nu} {m_2} {\nil}   |  \node {\nu_1} {n_1} {\wsndv_{q_1}}   |   \node {\nu_2} {n_2} {\wsndv_{q_2}}$ by the application of  a sequence of $\tau$-transitions. This allows us to apply Proposition~\ref{M_go_to_N} to  derive Equation~\ref{eq:2s2rNoColl}.

However, Equation~\ref{eq:2s2rNoColl}  can be generalised to an arbitrary number of senders and forwarders,  provided that  parallel contexts are unable to receive messages in \nolinebreak the \nolinebreak  current \nolinebreak round.  

\begin{thm}[Message propagation]
\label{thm:propagation}
Let $I$ and $J$ be  disjoint subsets of $\mathbb{N}$. 
Let $M$ be a well-formed network defined as 
\[
M \: \equiv \: N \q \big| \q
\prod_{i\in I} \node {\nu_{m_i}} {m_i} {\sndv_{p_i}} \q \big| \q
\prod_{j \in J } \node {\nu_{n_j}} {n_j} {\fwd_{q_j}}
\]
such that,  for all $i \in I$:
\begin{inparaenum}[(1)]
\item
$\{n_j : j \in J\} \subseteq \nu_{m_i} \subseteq   \nds{M}$; 
\item
 the nodes in $\nu_{m_i} \cap \nds{N}$ cannot receive in the current 
round.
\end{inparaenum}
 Then, 
\[
M  
\; \simtoll{r} \; 
N \q \big| \q
\prod_{i \in I} \node {\nu_{m_i}} {m_i} {\nil} \q \big| \q
 \prod_{j \in J} \node {\nu_{n_j}} {n_j} {\wsndv_{q_j}} 
\]
with $r= \prod_{i \in I}(1-p_i)$. 
\end{thm}

As an example, consider a simple gossiping network $\mathrm{GSP}_1$, with gossip probability $p$, composed by two source nodes $s_1$ and $s_2$  and a destination node
 $d$. The network is the following:
\[ 
\mathrm{GSP}_1  \;  \deff \;  
\prod_{i=1}^{2}\node {\nu_{s_i}}  {s_i} {\sndv_{p}} \q \big| \q
\node {\nu_d} d {\fine}
\]
with $\nu_{s_1} = \nu_{s_2} = \{d\}$ and $\nu_d = \{s_1, s_2, \mathit{tester} \}$, 
where ${\mathit{tester}}$ is a fresh node of the observer to test successful gossiping. For simplicity, here and in the rest of the paper, the node $\mathit{tester}$  can receive messages but it cannot transmit.

We would like to estimate the distance between  $\mathrm{GSP}_1$ and a network  $\mathrm{DONE}_1$  in which the message $v$  has been delivered to the destination node $d$: 
\[
\mathrm{DONE}_1   \;  \deff \; 
\prod_{i=1}^{2}\node {\nu_{s_i}}  {s_i} {\nil} \q \big| \q 
 \node {\nu_d} d {\wsndv_1}   .
\]
By an application of  Theorem~\ref{thm:propagation} it follows that  
$\mathrm{GSP}_1  \simtoll{ (1-p)^2 } \,   \mathrm{DONE}_1$. 
Thus, the network $\mathrm{GSP}_1$ succeeds in  delivering the message to the destination $d$  with  probability (at least) $1-(1-p)^2= 2p -p^2$.

Actually, Theorem~\ref{thm:propagation} represents an effective tool to deal with message propagation in gossip networks. However, it  requires that  all forwarders $n_j$  should be in the neighbouring of all senders $m_i$ (constraint $\{n_j : j \in J\} \subseteq \nu_{m_i} $), which  may represent a limitation in many cases. Consider, for example, a  gossiping network $\mathrm{GSP}_2$, with gossip probability $p$, composed by two source nodes $s_1$ and $s_2$, a destination node $d$, and three intermediate nodes $n_1$, $n_2$, and $n_3$: 
\[
\mathrm{GSP}_2 \;  \deff \;  
\prod_{i=1}^{2}\node {\nu_{s_i}}  {s_i} {\sndv_{p}} \q  \big| \q
\prod_{i=1}^{3} \node {\nu_{n_i}}  {n_i} {\fwd_{p}} \q  \big| \q 
\node {\nu_d} d {\fine}
\]
where 
$\nu_{s_1} = \{ n_1\}$, $\nu_{s_2} = \{ n_1 , n_2 \}$,
$\nu_{n_1}{=} \{ s_1, s_2 , n_3\}$, $\nu_{n_2}{=}\{ s_2, n_3 \}$,
$\nu_{n_3} = \{ n_1 , n_2,d \}$ and $\nu_d = \{ n_3 , \mathit{tester} \}$.

Here, we would like to estimate the distance between  $\mathrm{GSP}_2$, and a network  $\mathrm{DONE}_2$,  in which  the message $v$ propagated somehow up to the  destination  node  $d$: 
\[ 
\mathrm{DONE}_2 \; \deff \; 
\prod_{i=1}^{2}\node {\nu_{s_i}}  {s_i} {\nil} \q \big| \q
\prod^{3}_{i=1} \node {\nu_{n_i}}  {n_i} {\nil} \q \big| \q
 \node {\nu_d} d {\sigma^2 \!.\, \wsndv_1}  .
\]
Unfortunately,  we cannot directly apply Theorem~\ref{thm:propagation} 
to capture this message propagation because node $s_2$, unlike $s_1$, 
can transmit to both $n_1$ and $n_2$.
In this case, before applying Theorem~\ref{thm:propagation}, 
we would need a result to \emph{compose estimates} of partial networks. 
More precisely, a result which would  allow us  to take into account, in the calculation of the tolerance, both the probability that a sender transmits 
and the probability that the same sender does not transmit.

The additive law of Proposition~\ref{prop:propagation-general} allows us to prove the following result. 
\begin{thm}[Composing networks]
\label{thm:propagation2}
If $M\ntransSimone[\sigma]$ then 

\[
N \q\, \big| \q\,
\node {\nu_{m}} {m} {\sndv_{p}} \q\, \big| \q\,
\prod_{j\in J} \node {\nu_{n_j}} {n_j} {\rcvtime {x_j}{P_j}{Q_j}}
\simtoll {r}
\: M
\]
with $r={p}{q_1}+ {(1{-}p}){q_2}$, whenever
\begin{itemize}
\item
$N \; | \; 
\node {\nu_{m}} {m} {\nil}  \; | \; 
\prod_{j\in J} \node {\nu_{n_j}} {n_j} {{\subst v {x_j}}P_j}
  \simtoll{q_1} \: M
$
\item
$
N \;  | \; 
 \node {\nu_{m}} {m} {\nil}  \;  | \; 
\prod_{j\in J} \node {\nu_{n_j}} {n_j} {\rcvtime {x_j}{P_j}{Q_j}}   
\simtoll{q_2} \: M 
$
\item  
$\{n_j :  j \in J\} \subseteq \nu_{m} \subseteq\{n_j : j \in J\} \cup \nds{N}$
\item 
nodes in $\nu_{m} \cap \nds{N}$ cannot receive in the current round.\footnote{We could generalise the result to
take into account more senders at the same time. 
This would not add expressiveness, it would just speed up the reduction process.}
\end{itemize}
\end{thm}

Intuitively:
\begin{inparaenum}[(i)]
\item  
in the network $N   |   \node {\nu_{m}} {m} {\sndv_{p}}   |   \prod_{j\in J} \node {\nu_{n_j}} {n_j} {\rcvtime {x_j}{P_j}{Q_j}}$ the sender $m$ has not performed
yet the $\tau$-action that resolves the probabilistic choice between broadcasting the message $v$ or not;
\item  
in the network $\; N \ | \node {\nu_{m}} {m} {\nil}   |  \prod_{j\in J} \node {\nu_{n_j}} {n_j} {{\subst v {x_j}}P_j}$  the sender $m$ has
resolved the probabilistic choice deciding to broadcast the message $v$;
\item 
 finally, 
in the network $ N  |  \node {\nu_{m}} {m} {\nil}  |  \prod_{j\in J} \node {\nu_{n_j}} {n_j} {\rcvtime {x_j}{P_j}{Q_j}}$ the sender $m$
has  resolved the probabilistic choice deciding not to  broadcast $v$.
\end{inparaenum}

Later on we will use  the following instance of Theorem~\ref{thm:propagation2}. 
\begin{cor}
\label{cor:propagation2}
If $M \ntransSimone{\sigma}$ then
\[
N \q \big| \q
\node {\nu_{m}} {m} {\sndv_{p}} \q \big| \q
\prod_{j\in J} \node {\nu_{n_j}} {n_j} {\fwd_{{q}}}
 \simtoll{r}  \: 
M 
\]
with $r=pq_1 + (1{-}p)q_2$, whenever:
\begin{itemize}
\item
$
N \q \big| \q  
\node {\nu_{m}} {m} {\nil}  \q \big| \q
\prod_{j\in J} \node {\nu_{n_j}} {n_j} {\wsndv_{q} } 
\simtoll{q_1} \:  
M $
\item
$
N \q \big| \q
 \node {\nu_{m}} {m} {\nil}  \q \big| \q
\prod_{j\in J} \node {\nu_{n_j}} {n_j} {\fwd_{{q}}} 
\simtoll{q_2} \:  
M$
\item  
$\{n_j \mid j \in J\} \subseteq \nu_{m} \subseteq\{n_j \mid j \in J\} \cup \nds{M}$
\item 
nodes in $\nu_{m} \cap \nds{N}$ cannot receive in the current round.
\end{itemize}
\end{cor}

Now, we have all  algebraic tools to compute an estimation of the tolerance $r$, such that 
$ \mathrm{GSP}_2 \simtoll{r} \mathrm{DONE}_2 $.
Basically, we will compute the tolerances for two partial networks and then will use Theorem~\ref{thm:propagation2} to compose the two tolerances. 

As a first step, we compute an estimation for  the network $\mathrm{GSP}_2$ in which the sender $s_2$ \emph{has already
broadcast} the message $v$ to its neighbours $n_1$ and $n_2$. 
To this end,   we derive the following chain of similarities by applying, in sequence:
\begin{inparaenum}[(i)] 
\item 
Proposition~\ref{thm:timing}(\ref{law:no-trans}) 
to express that if $n_1$ receives the message transmitted by $s_2$ then the message broadcast  by $s_1$ gets lost;
\item
Proposition~\ref{thm:timing}(\ref{timing5})  
to formalise that in the current round  both nodes $n_3$ and $d$ do not receive any message because there are no transmitters in their neighbourhood, as a consequence, they are obliged to timeout;
\item 
Theorem~\ref{thm:propagation} in which nodes $n_1$ and $n_2$ act as senders and $n_3$ as a forwarder, together with Proposition~\ref{thm:timing}(\ref{timing2}) to work underneath $\sigma$ prefixes;
\item
Proposition~\ref{thm:timing}(\ref{timing1}) to remove $\sigma$ prefixes in both nodes $n_1$ and $n_2$;
\item  
Proposition~\ref{thm:timing}(\ref{timing5}) because in the next round there are no nodes that can transmit to  $d$; 
\item 
Theorem~\ref{thm:propagation} in which $n_3$ acts as sender and $d$ as forwarder, together with Proposition~\ref{thm:timing}(\ref{timing2}) to work underneath $\sigma$ prefixes;
\item
Proposition~\ref{thm:timing}(\ref{timing1}) to remove  $\sigma$ prefixes in  $n_3$.
\end{inparaenum}
In all steps, we have reasoned in a compositional manner, up to common parallel components (Theorem~\ref{cor:non_expansiveness}). For convenience, in the chain below we define two auxiliary networks: $\mathrm{NET}_a$ and $\mathrm{NET}_b$.
\[
\begin{array}{rrl} 
&&\mathrm{GSP}_2(s_1[v]\mbox{; }s_2 \mbox{ gossiped})
\\[3pt]
& \deff &
\node {\nu_{s_1}}  {s_1} {\sndv_p} \; \big| \;
\node {\nu_{s_2}}  {s_2} {\nil} \; \big| \;
\prod^{2}_{i=1} \node {\nu_{n_i}}  {n_i} {\wsndv_{p}} \; \big| \;
 \node {\nu_{n_3}}  {n_3} {\fwd_{p}} \; \big| \;
\node {\nu_d} d {\fine}
\\[3pt]
& \simtoll{0} & 
\prod_{i=1}^{2}\node {\nu_{s_i}}  {s_i} {\nil} \; \big| \;
\prod^{2}_{i=1} \node {\nu_{n_i}}  {n_i} {\wsndv_{p}} \; \big| \;
\node {\nu_{n_3}}  {n_3} {\fwd_{p}} \; \big| \;
\node {\nu_d} d {\fine}
\end{array}
\]
\[
\begin{array}{rrl} 
& \simtoll{0} &
\prod_{i=1}^{2}\node {\nu_{s_i}}  {s_i} {\nil} \; \big| \;
\prod^{2}_{i=1} \node {\nu_{n_i}}  {n_i} {\sigma.  \sndv_{p}} \; \big| \;
 \node {\nu_{n_3}}  {n_3} {\sigma.  \fwd_{p}} \; \big| \;
\node {\nu_d} d {\sigma.  \fine}
\\[3pt]
& \deff & \mathrm{NET}_a
\\[3pt]
& \simtoll{(1-p)^2} & 
\prod_{i=1}^{2}\node {\nu_{s_i}}  {s_i} {\nil} \; \big| \;
\prod^{2}_{i=1} \node {\nu_{n_i}}  {n_i} {\sigma.\nil} \; \big| \;
 \node {\nu_{n_3}}  {n_3} {\sigma. \wsndv_{p}} \; \big| \;
\node {\nu_d} d {\sigma .\fine}
\\[3pt]
& \deff & \mathrm{NET}_b 
\\[3pt] 
& \simtoll{0} & 
\prod_{i=1}^{2}\node {\nu_{s_i}}  {s_i} { \nil} \; \big| \;
\prod^{2}_{i=1} \node {\nu_{n_i}}  {n_i} { \nil} \; \big| \;
\node {\nu_{n_3}}  {n_3} { \sigma^2.\, \sndv_{p}} \; \big| \;
\node {\nu_d} d { \sigma\!.\,\fine}
\\[3pt]
& \simtoll{0} & 
\prod_{i=1}^{2}\node {\nu_{s_i}}  {s_i} { \nil} \; \big| \;
\prod^{2}_{i=1} \node {\nu_{n_i}}  {n_i} { \nil} \; \big| \;
\node {\nu_{n_3}}  {n_3} { \sigma^2\!.\, \sndv_{p}} \; \big| \;
\node {\nu_d} d { \sigma^2\!.\,\fine}
\\[3pt]
& \simtoll{1-p} & 
\prod_{i=1}^{2}\node {\nu_{s_i}}  {s_i} {\nil} \; \big| \;
\prod^{2}_{i=1} \node {\nu_{n_i}}  {n_i} {\nil} \; \big| \;
 \node {\nu_{n_3}}  {n_3} {  \sigma^2\!.\, \nil} \; \big| \;
\node {\nu_d} d {\sigma^2\!. \,  \wsndv_1}
\\[3pt]
& \simtoll{0} & 
\prod_{i=1}^{2}\node {\nu_{s_i}}  {s_i} {\nil} \; \big| \;
\prod^{3}_{i=1} \node {\nu_{n_i}}  {n_i} {\nil} \; \big| \;
 \node {\nu_d} d {\sigma^3\!. \, \sndv_1} 
\\[3pt]
 & = & \mathrm{DONE}_2  . 
\end{array}
\]
By transitivity (Proposition~\ref{prop:transitivity}) we can easily derive 
$
\mathrm{NET}_b
\: \simtoll{1-p} \:
 \mathrm{DONE}_2$. 
This allows us  to apply  Proposition~\ref{lem:O_go_to_N_sim_M_conSigma}, in which $M$, $N$, and $O$ are instantiated with $\mathrm{DONE}_2$, $\mathrm{NET}_b$, and $\mathrm{NET}_a$, respectively, to derive the equation 
$
\mathrm{NET}_a 
\, \simtoll{1-2p^2+p^3}  \, 
 \mathrm{DONE}_2 
$
where the tolerance $1-2p^2+p^3$ is obtained by solving the expression $(1-p)(1-(1-p)^2) + (1-p)^2$. 
By transitivity  on the first two steps of the  chain above, we can finally derive the equation
\begin{equation}
\label{eq_GSP1-1}
\mathrm{GSP}_2(s_1[v]\mbox{; }s_2 \mbox{ gossiped})\; 
\simtoll{1-2p^2+p^3} \; \, \mathrm{DONE}_2   
\end{equation}

Similarly, we can compute an estimation of the tolerance  which
allows the network $\mathrm{GSP}_2$ in which the sender $s_2$ \emph{did not 
broadcast} the message $v$ to its neighbours to simulate 
the network $\mathrm{DONE}_2$.  
To this end,  we derive the following chain of similarities by applying, in sequence: 
\begin{inparaenum}[(i)]
\item Theorem~\ref{thm:propagation};
\item Proposition~\ref{thm:timing}(\ref{timing5});
\item Theorem~\ref{thm:propagation} and  Proposition~\ref{thm:timing}(\ref{timing2});
\item  Proposition~\ref{thm:timing}(\ref{timing1}) and 
 Proposition~\ref{thm:timing}(\ref{timing5});
\item Theorem~\ref{thm:propagation} and  Proposition~\ref{thm:timing}(\ref{timing2}); 
\item  Proposition~\ref{thm:timing}(\ref{timing1}) and  Proposition~\ref{thm:timing}(\ref{law:no-trans}). 
\end{inparaenum}
Again, in all steps, we have reasoned up to common parallel components 
(Theorem~\ref{cor:non_expansiveness}).
For convenience, in the chain below we define three auxiliary networks: $\mathrm{NET}_c$,  $\mathrm{NET}_d$ and 
$\mathrm{NET}_e$. 
\[
\begin{array}{rrl}
&&\mathrm{GSP}_2(s_1[v]\mbox{; }s_2 \mbox{ did not gossip})
\\[3pt]
& \deff &
\node {\nu_{s_1}}  {s_1} {\sndv_p} \; \big| \;
\node {\nu_{s_2}}  {s_2} {\nil} \; \big| \;
\prod_{i=1}^{3} \node {\nu_{n_i}}  {n_i} {\fwd_{p}} \; \big| \;
\node {\nu_d} d {\fine}
\\[3pt]
& \simtoll{1-p} &
\prod_{i=1}^2\node {\nu_{s_i}}  {s_i} {\nil} \; \big| \;
 \node {\nu_{n_1}}  {n_1} {\wsndv_{p}} \; \big| \;
\prod^{3}_{i=2} \node {\nu_{n_i}}  {n_i} {\fwd_{p}} \; \big| \;
\node {\nu_d} d {\fine}
\\[3pt]
& \deff & \mathrm{NET}_c
\\[3pt]  
& \simtoll{0} &
\prod_{i=1}^{2}\node {\nu_{s_i}}  {s_i} {\nil} \; \big| \;
 \node {\nu_{n_1}}  {n_1} {\sigma. \sndv_{p}} \; \big| \;
\prod^{3}_{i=2} \node {\nu_{n_i}}  {n_i} {\sigma. \fwd_{p}} \; \big| \;
\node {\nu_d} d { \sigma . \fine}
\\[3pt]
& \deff & \mathrm{NET}_d
\\[3pt] 
& \simtoll{1-p} &
\prod_{i=1}^{2}\node {\nu_{s_i}}  {s_i} {\nil} \; \big| \;
 \node {\nu_{n_1}}  {n_1} {\sigma.\nil} \; \big| \;
 \node {\nu_{n_2}}  {n_2} {\sigma. \fwd_p} \; \big| \;
\node {\nu_{n_3}}  {n_3} {\sigma. \wsndv_{p}} \; \big| \;
\node {\nu_d} d {\sigma. \fine}
\\[3pt]
& \deff & \mathrm{NET}_e
\\[3pt] %%[3pt]
& \simtoll{0} &
\prod_{i=1}^{2}\node {\nu_{s_i}}  {s_i} {\nil} \; \big| \;
 \node {\nu_{n_1}}  {n_1} {\nil} \; \big| \;
 \node {\nu_{n_2}}  {n_2} {\sigma^2\!.\,  \fwd_p} \; \big| \;
\node {\nu_{n_3}}  {n_3} {\sigma^2\!.\,  \sndv_{p}} \; \big| \;
\node {\nu_d} d {\sigma^2\!.\,  \fine}
\\[3pt]
& \simtoll{1-p} &
\prod_{i=1}^{2}\node {\nu_{s_i}}  {s_i} {\nil} \; \big| \;
 \node {\nu_{n_1}}  {n_1} {\nil} \; \big| \;
 \node {\nu_{n_2}}  {n_2} {\sigma^2\!.\,  \wsndv_p} \; \big| \;
\node {\nu_{n_3}}  {n_3} { \sigma^2\!.\nil} \; \big| \;
\node {\nu_d} d {\sigma^2\!.\,  \wsndv_1}
\\[3pt]
& \simtoll{0} &
\prod_{i=1}^{2}\node {\nu_{s_i}}  {s_i} {\nil} \; \big| \;
\prod^{3}_{i=1} \node {\nu_{n_i}}  {n_i} {\nil} \; \big| \;
\node {\nu_d} d { \sigma^3 . \sndv_1} 
\\[3pt]
& = & \mathrm{DONE}_2 \,  . 
\end{array}
\]
Then,   by transitivity (Proposition~\ref{prop:transitivity}) we can derive the equation 
\[
\mathrm{NET}_e
\; \simtoll{1-p} \; \, 
 \mathrm{DONE}_2 . 
\]
This allows us  to apply  Proposition~\ref{lem:O_go_to_N_sim_M_conSigma}, in which $M$, $N$, and $O$ are instantiated with $\mathrm{DONE}_2$, $\mathrm{NET}_e$, and $\mathrm{NET}_d$, respectively, to  derive the equation
\[
\mathrm{NET}_d
\; \simtoll{1- p^2} \; \,
 \mathrm{DONE}_2 
\]
where the tolerance $1- p^2$ is obtained by solving the expression $(1-p)(1-(1-p)) + (1-p)$. 
 By transitivity  we can easily derive:
\[
\mathrm{NET}_c
\; \simtoll{1- p^2} \; \,
 \mathrm{DONE}_2 . 
\] 
This last equation  can be used to apply  Proposition~\ref{lem:O_go_to_N_sim_M},  
in which $M$, $N$, and $O$ are instantiated with $\mathrm{DONE}_2$, $\mathrm{NET}_c$, and $\mathrm{GSP}_2(s_1[v]\mbox{; }s_2 \mbox{ did not gossip}) $, respectively,  to  derive:
\begin{equation}
\label{eq_GSP1-2}
\mathrm{GSP}_2(s_1[v]\mbox{; }s_2 \mbox{ did not gossip}) \; \simtoll{1-p^3} \;  \mathrm{DONE}_2 
\end{equation}
where the tolerance $1- p^3$ is obtained by solving the expression $(1-p)(1-(1-p^2)) + (1-p^2)$.

Finally, 
we can apply  Corollary~\ref{cor:propagation2} to Equations~\ref{eq_GSP1-1} and \ref{eq_GSP1-2} to derive: 
\[ 
\mathrm{GSP}_2 \q
\simtoll{1-(3p^3-2p^4)} \q \; \mathrm{DONE}_2 .
\]
Thus,  the gossip network $\mathrm{GSP}_2$ will succeed in propagating the messages to the destination $d$ with probability at least $3p^3-2p^4$. 
For instance, with a gossip probability $p=0.8$ the destination will receive the message with probability $0.716$, with a margin of $10\%$.
For $p=0.85$ the success probability  increases to $0.798$, with a margin of $6\%$; while for $p=0.9$ the probability to propagate the message    
to the destination rises to  $0.88$, with a margin of only $2\%$. 
As a consequence, $p=0.9$ can be considered the threshold of our small network.\footnote{Had we had
more senders we would have estimated 
a better threshold.}

Note that  in the previous example both messages may reach the 
destination node in exactly three rounds. 
However, more generally,
we could have different message propagation paths in the same
network  which might take a different amount of time to be traversed. 
The algebraic theory we developed up to now does not allow us to deal with 
paths of different lengths. As an example, we would like to estimate the distance between 
the network
\[
\mathrm{GSP}_3 \, \deff \,
 \node {\nu_{s_1}}  {s_1} {\sndv_{p}} \q \big| \q\node {\nu_{s_2}}  {s_2} {\sndv_{p}} \q \big| \q
  \node {\nu_{n }}  {n } {\fwd_{p}} \q \big| \q
\node {\nu_d} d {\fine}
\]
with topology
$\nu_{s_1} = \{  d\}$, $\nu_{s_2} = \{   n  \}$,
$\nu_{n }{=} \{ s_2, d\}$ and
$\nu_{d} = \{s_1,n,\mathit{tester}\}$,
and the network
\[ 
\mathrm{DONE}_3 \,  \deff \, 
\node {\nu_{s_1}}  {s_1} {\nil} \q \big| \q \node {\nu_{s_2}}  {s_2} {\nil} \q \big| \q
  \node {\nu_{n }}  {n } {\nil} \q \big| \q
 \node {\nu_d} d { \tau.(\sigma.  \sndv_1 \oplus_p \sigma^2 . \sndv_1)}
\]
in which the message $v$ propagated  up to the destination node $d$
following two different paths. Thus, $d$  will probabilistically choose between broadcasting $v$  after one or two rounds. 

The following result provides the missing instrument.
\begin{thm}[Composing paths]
\label{thm:propagation2bis}
Let $M$ be a well-formed network.  Then, 
\[
M \q \big| \q
\node {\nu_{m}} {m} {\tau.\bigoplus_{i \in I} p_i {\colon} Q_i} 
\q \simtoll{r}  \q
\prod_{j\in J} \node {\nu_{n_j}} {n_j} {\nil }  \; \big| \; \node {\nu_d} d {\tau.\bigoplus_{i \in I} p_i {\colon} P_i} 
\]
with $r=\sum_{i \in I} p_i s_i$, whenever
$
M | 
\node {\nu_{m}} {m} {Q_i} 
\, \simtoll{s_i}  \;
\prod_{j\in J} \node {\nu_{n_j}} {n_j} {\nil }  |  \node {\nu_d} d {P_i } 
$, for any $i \in I$.
\end{thm}

Let us provide an estimation of the distance between $\mathrm{GSP}_3$ and $\mathrm{DONE}_3$.
As a first step, we compute an estimation of the tolerance  which allows $\mathrm{GSP}_3$ to simulate the first probabilistic behaviour of $\mathrm{DONE}_3$ (broadcast after one round).
To this end, we derive the following chain of similarities by applying, in sequence:   
\begin{inparaenum}[(i)]
\item Theorem~\ref{thm:propagation} to model the transmission of $s_1$;  
\item again Theorem~\ref{thm:propagation} to model the transmission of $s_2$; 
\item Proposition~\ref{thm:timing}(\ref{law:no-trans})  to model that the transmission of $n$ gets lost, together
with Proposition~\ref{thm:timing}(\ref{timing2}) to work underneath $\sigma$ prefixes.
\end{inparaenum}
In all steps, we reason up to parallel components (Theorem~\ref{cor:non_expansiveness}). 
\[
\begin{array}{rrl}
&& \mathrm{GSP}_3 (s_1 \mbox{ gossiped} \mbox{; } s_2[v]) 
\\[2pt]
& \deff & 
\node {\nu_{s_1}}  {s_1} {\sndv_{1}} \; \big| \; 
\node {\nu_{s_2}}  {s_2} {\sndv_{p}} \; \big| \;
\node {\nu_{n }}  {n } {\fwd_{p}} \; \big| \;
\node {\nu_d} d {\fine}  
\\[3pt]
&\simtoll{0} &
\node {\nu_{s_1}}  {s_1} {\nil} \; \big| \;
\node {\nu_{s_2}}  {s_2} {\sndv_{p}} \; \big| \;
\node {\nu_{n }}  {n } {\fwd_{p}} \; \big| \;
\node {\nu_d} d {\wsndv_{1}}  
\\[3pt]
&\simtoll{1-p} &
\node {\nu_{s_1}}  {s_1} {\nil} \; \big| \;
\node {\nu_{s_2}}  {s_2} {\nil} \; \big| \;
\node {\nu_{n }}  {n } {\wsndv_{p}} \; \big| \;
\node {\nu_d} d {\wsndv_{1}}  
\\[3pt]
& =  &
\node {\nu_{s_1}}  {s_1} {\nil} \; \big| \;
\node {\nu_{s_2}}  {s_2} {\nil} \; \big| \;
\node {\nu_{n }}  {n } {\sigma. \sndv_{p}} \; \big| \;
\node {\nu_d} d {\sigma. \sndv_{1}}  
\\[3pt]
& \simtoll{0} &
 \node {\nu_{s_1}}  {s_1} {\nil} \; \big| \;\node {\nu_{s_2}}  {s_2} {\nil} \; \big| \;
  \node {\nu_{n }}  {n } {\nil} \; \big| \;
 \node {\nu_d} d {  \sigma. \sndv_1}  .
\end{array}
\]
By transitivity (Proposition~\ref{prop:transitivity}) we derive: 
\begin{equation}
\label{newEx1}
\mathrm{GSP}_3 (s_1 \mbox{ gossiped} \mbox{; } s_2[v]) \; 
  \simtoll{1-p}  \; \, 
 \node {\nu_{s_1}}  {s_1} {\nil} \; \big| \; \node {\nu_{s_2}}  {s_2} {\nil} \; \big| \; 
  \node {\nu_{n }}  {n } {\nil} \; \big| \;
 \node {\nu_d} d {  \sigma. \sndv_1} 
\end{equation}

Then, we compute an estimation of the tolerance  which allows the network $\mathrm{GSP}_3$ to simulate the second probabilistic behaviour of $\mathrm{DONE}_3$.
To this end, we derive the following  chain of similarities by applying, 
in sequence:
\begin{inparaenum}[(i)]
\item Theorem~\ref{thm:propagation} to model the transmission of $s_2$; 
\item Proposition~\ref{thm:timing}(\ref{timing5}) as without transmitters the node $d$ is obliged to timeout; 
\item   Theorem~\ref{thm:propagation} as node $n$ may transmit to $d$, working underneath a $\sigma$ prefix (Proposition~\ref{thm:timing}(\ref{timing2}));
\item  Proposition~\ref{thm:timing}(\ref{timing1}) to remove a $\sigma$ prefix in node $n$.
\end{inparaenum}
As usual, in all steps, we reason up to parallel components (Theorem~\ref{cor:non_expansiveness}).
For convenience, in the chain below we define two auxiliary networks:   $\mathrm{NET}_f$ and 
$\mathrm{NET}_g$. 
\begin{displaymath}
\begin{array}{rrl}  
& & \mathrm{GSP}_3(s_1 \mbox{ did not gossip} \mbox{; } s_2[v]) 
\\[2pt] 
& \deff & 
\node {\nu_{s_1}}  {s_1} {\nil} \; \big| \;
\node {\nu_{s_2}}  {s_2} {\sndv_{p}} \; \big| \;
\node {\nu_{n }}  {n } {\fwd_{p}} \; \big| \;
\node {\nu_d} d {\fine}
\\[3pt]
& \simtoll{1-p} &
\node {\nu_{s_1}}  {s_1} {\nil} \; \big| \;
\node {\nu_{s_2}}  {s_2} {\nil} \; \big| \;
\node {\nu_{n }}  {n } {\wsndv_{p}} \; \big| \;
\node {\nu_d} d {\fine}
\\
& \deff & \mathrm{NET}_f
\\[3pt]
& \simtoll{0} &
\node {\nu_{s_1}}  {s_1} {\nil} \; \big| \;
\node {\nu_{s_2}}  {s_2} {\nil} \; \big| \;
\node {\nu_{n }}  {n } {\sigma.\sndv_{p}} \; \big| \;
\node {\nu_d} d {\sigma.\fine}
\\[3pt]
& \simtoll{1-p} &
\node {\nu_{s_1}}  {s_1} {\nil} \; \big| \;
\node {\nu_{s_2}}  {s_2} {\nil} \; \big| \;
\node {\nu_{n }}  {n } {\sigma.\nil} \; \big| \;
\node {\nu_d} d {\sigma.\wsndv_1}
\\[3pt]
& \simtoll{0} &
\node {\nu_{s_1}}  {s_1} {\nil} \; \big| \;
\node {\nu_{s_2}}  {s_2} {\nil} \; \big| \;
\node {\nu_{n }}  {n } {\nil} \; \big| \;
\node {\nu_d} d {\sigma^2.  \sndv_1}
\\[2pt]  
& \deff & \mathrm{NET}_g.
\end{array}
\end{displaymath}
Then,   by transitivity (Proposition~\ref{prop:transitivity}) we can derive 
%%\[
$\mathrm{NET}_f
\, \simtoll{1-p} \,
\mathrm{NET}_g $. 
%%\]
This last equation  can be used to apply  Proposition~\ref{lem:O_go_to_N_sim_M},   in which $M$, $N$, and $O$ are instantiated with $\mathrm{NET}_g $, $\mathrm{NET}_f$, and $\mathrm{GSP}_3(s_1 \mbox{ did not gossip} \mbox{; } s_2[v]) $, respectively,  to  derive:
 \begin{equation}
\label{newEx2}
\mathrm{GSP}_3(s_1 \mbox{ did not gossip; } s_2[v]) \q 
 \simtoll{1-p^2} \q \,  
 \mathrm{NET}_g 
\end{equation}
where the tolerance $1- p^2$ is obtained by solving the expression $(1-p)(1-(1-p)) + (1-p)$.

Finally, we can apply Theorem~\ref{thm:propagation2bis} to Equations~\ref{newEx1} and \ref{newEx2} to obtain: 
\[
\mathrm{GSP}_3 \; \simtoll{r} \; \mathrm{DONE}_3 
\]
with $r= p(1-p)+(1-p)(1  -p^2)$. Thus, the network $\mathrm{GSP}_3$ will succeed in transmitting both  messages $v$ to the destination $d$, with probability at least $1- r$. 

We conclude by observing that, in  order to deal with paths of different length, one should apply Theorem~\ref{thm:propagation2bis} to all possible paths.

%%%%%%%%%%%%%%%%%%%%%%%%%%%%%%%%%%%%%

\section{Gossiping with collisions}
\label{sec:gossip-collisions}

An important aspect  of  wireless communication is the presence of collisions:
if two or more stations are transmitting over the same channel during the same round then a \emph{communication collision} occurs at the receiver nodes. 
Collisions in wireless systems cannot be avoided, although there are protocols to reduce their occurrences (see, for instance, the  protocol IEEE 802.11 CSMA/CA). 

In Section~\ref{sec:gossip-nocollisions}, we have reasoned assuming no collisions.
Here, we demonstrate that  the presence of communication collisions deteriorates the performance of gossip protocols.

The calculus \cname{} allows us to easily express communication collisions.
A receiver node experiences a collision if it is exposed to more than one transmission in the same round; as a consequence the reception in that round fails.
We can model this behaviour by redefining the forwarder process 
as follows:
\[
\mathsf{resndc}\langle u \rangle_{p_{\mathrm{g}}} \!\deff \!
 \rcvtime x {\nil} {\mathsf{snd}\langle u \rangle_{p_{\mathrm{g}}}}
\Q\Q\Q\Q
\fwdc_{p_{\mathrm{g}}} \!\deff\! {\rcv x \mathsf{resndc}\langle x \rangle_{p_{\mathrm{g}}}} \, .
\]
Here, the forwarder $\fwdc_{p_{\mathrm{g}}}$ waits for a message in the current 
round. If  a second message is received in the same round then the 
forwarding  is doomed to fail; otherwise if no other messages are received, 
the process moves to next round and broadcasts the received message with 
the given gossip probability. 
Thus, for example, the Equation~\ref{eq:2s2rNoColl} seen in 
 the previous section becomes:  
\begin{equation}
\label{eq:2s2rWithColl}
\begin{array}{c}
\node {\nu} {m_1} {\sndv_{p_1}} \; \big| \;
\node {\nu} {m_2} {\sndv_{p_2}}\; \big| \;
\node {\nu_1} {n_1} {\fwdc_q}\; \big| \;
\node {\nu_2} {n_2} {\fwdc_r} \\[3pt]
\simtoll{r}
\\[3pt]
\Q \node {\nu} {m_1} {\nil} \; \big| \;
\node {\nu} {m_2} {\nil} \; \big| \;
\node {\nu_1} {n_1} {\wsndvc_{q}} \; \big| \;
\node {\nu_2} {n_2} {\wsndvc_{r}} 
\end{array}
\end{equation}
with a tolerance $r=1-(p_1(1-p_2)+ p_2(1-p_1))$ which is definitely greater 
than the tolerance $(1{-} p_1)(1{-} p_2)$ seen in Equation~\ref{eq:2s2rNoColl} for the same network in the 
absence of communication collisions. 
As expected, the difference  between the two tolerances corresponds to the probability $p_1 \cdot p_2$ that both senders $m_1$ and $m_2$ decide to transmit.
In the previous section, we argued that in the absence of collisions the first   broadcast message is received by both  nodes $n_1$ and $n_2$, while the second message %%, broadcast by the other sender, 
 is simply ignored by  $n_1$ and $n_2$. 
However, this situation obviously gives rise to a communication collision. More generally, when collisions are taken into account, 
Theorem~\ref{thm:propagation} should be reformulated as follows:

\begin{thm}[Message propagation with collisions]
\label{thm:propagation3}
Let $I$ and $J$ be disjoint subsets of $\mathbb{N}$. 
Let $M$ be a well-formed network defined as
\[
M \equiv  
N \q \big| \q
\prod_{i\in I} \node {\nu_{m_i}} {m_i} {\sndv_{p_i}} \q \big| \q
\prod_{j \in J } \node {\nu_{n_j}} {n_j} {\fwdc_{q_j}}
\]
such that,  for all $i \in I$:
\begin{inparaenum}[(1)]
\item 
 $\{n_j : j \in J\} \subseteq \nu_{m_i} \subseteq  \nds{M}$, and
\item 
 the nodes in $\nu_{m_i} \cap \nds{N}$ cannot receive in the current round.
\end{inparaenum}
Then,
\[
M 
\;  \simtoll{r} \;  \, 
 N \q \big| \q
\prod_{i \in I} \node {\nu_{m_i}} {m_i} {\nil} \q \big| \q
 \prod_{j \in J} \node {\nu_{n_j}} {n_j} {\wsndvc_{q_j}} 
\]
with $r= 1- \sum_{i\in I}p_i \prod_{j\in I{\setminus}\{i\}}(1-p_j)$. 
\end{thm}
Here, the reader may notice that the tolerance is 
different with respect to that derived in Theorem~\ref{thm:propagation}. This is because, in the presence of collisions,  forwarders successfully receive (and forward) a value only if exactly one sender transmits it; the forwarding fails otherwise.

As an example, consider  a variant 
of the gossiping network $\mathrm{GSP}_1$ of
the previous section where communication collisions are taken into account:
\[
\mathrm{GSP}_4 \; \deff \; 
\prod_{i=1}^{2}\node {\nu_{s_i}}  {s_i} {\sndv_{p}} \q \big| \q 
\node {\nu_d} d {\fwdc_1}   . 
\]
The corresponding network where the message has been 
propagated up to the destination node is not 
affected by the presence of collisions. Just for convenience, we define $\mathrm{DONE}_4 =\mathrm{DONE}_1$.
Since  the node $\mathit{tester} \in \nu_d$
 cannot transmit, by an application of  Theorem~\ref{thm:propagation3} and
Proposition~\ref{thm:timing}(\ref{random3})  it follows that:  
\[
\mathrm{GSP}_4 \;  \simtoll{ 1-2 p(1-p) } \; \,  \mathrm{DONE}_4   .
\]
Thus, $\mathrm{GSP}_4$ succeeds in propagating the message $v$ with (at least) probability  $2 p(1-p)$. In the previous section, we have seen that in 
the absence of collisions, 
the success probability of $\mathrm{GSP_1}$ is (at least)  $2p -p^2$. 
Since  $2 p(1-p) < 2p - p^2$, for any $p> 0$, it follows that the 
 presence of collisions  downgrades the performance of
this small network of $p^2$, that coincides with the probability that both senders transmit.

Notice that,  unlike Theorem~\ref{thm:propagation},  
Theorem~\ref{thm:propagation2} is still valid when dealing with communications collisions due to its general formulation. 
 However, for
commodity, in the following we will use two instances of it. 
\begin{cor}
\label{cor:propagation2_coll}
If $M \ntransSimone[\sigma]$ then
\[
\begin{array}{c}
N \q \big| \q
\node {\nu_{m}} {m} {\sndv_{p}} \q \big| \q
\prod_{j\in J} \node {\nu_{n_j}} {n_j} {\wsndvc_{{q}}}  
\simtoll{r}   \:
M
\end{array}
\]
with $r=ps_1 + (1-p)s_2$, whenever: 
\begin{itemize}
\item
$
N \q \big| \q  
\node {\nu_{m}} {m} {\nil}  \q \big| \q
\prod_{j\in J} \node {\nu_{n_j}} {n_j} { \nil}
\simtoll{s_1} \:
M$
\item
$
N \q \big| \q
 \node {\nu_{m}} {m} {\nil}  \q \big| \q
\prod_{j\in J} \node {\nu_{n_j}} {n_j} {\wsndvc_{q}}
\simtoll{s_2} \:
M$
\item  
$\{n_j : j \in J\} \subseteq \nu_{m} \subseteq\{n_j : j \in J\} \cup \nds{M}$
\item 
nodes in $\nu_{m} \cap \nds{M}$ cannot receive in the current round.
\end{itemize}
\end{cor}

\begin{cor}
\label{cor:propagation2_coll_bis}
If $M \ntransSimone[\sigma]$ then
\[
\begin{array}{c}
N \q \big| \q
\node {\nu_{m}} {m} {\sndv_{p}} \q \big| \q
\prod_{j\in J} \node {\nu_{n_j}} {n_j} {\fwdc_{{q}}} 
\simtoll{r} \:
M
\end{array}
\]
with $r=ps_1 + (1-p)s_2$, whenever:
\begin{itemize}
\item
$
N \q \big| \q  
\node {\nu_{m}} {m} {\nil}  \q \big| \q
\prod_{j\in J} \node {\nu_{n_j}} {n_j} {\wsndvc_{q} } 
\simtoll{s_1} \: 
M
$
\item
$
N \q \big| \q
 \node {\nu_{m}} {m} {\nil}  \q \big| \q
\prod_{j\in J} \node {\nu_{n_j}} {n_j} {\fwdc_{{q}}}
\simtoll{s_2} \: 
M$
\item  
$\{n_j : j \in J\} \subseteq \nu_{m} \subseteq\{n_j : j \in J\} \cup \nds{M}$
\item 
nodes in $\nu_{m} \cap \nds{M}$ cannot receive in the current round.
\end{itemize}
\end{cor}

Let us apply the previous results to 
study how communication collisions degrade the performances 
of gossip protocols. Let us reformulate the network 
$\mathrm{GSP}_2$ of the previous section in a scenario with potential 
communication collisions. Let us define
\[
\mathrm{{GSP}_5} \; \deff \; 
\prod_{i=1}^{2}\node {\nu_{s_i}}  {s_i} {\sndv_{p}} \q \big| \q
\prod_{i=1}^{3} \node {\nu_{n_i}}  {n_i} {\fwdc_{p}} \q \big| \q
\node {\nu_d} d {\fwdc_1}
\]
with the same network topology as $\mathrm{GSP}_2$: $\nu_{s_1} = \{ n_1\}$, $\nu_{s_2} = \{ n_1 , n_2 \}$, $\nu_{n_1}{=} \{ s_1, s_2 , n_3\}$, $\nu_{n_2}{=}\{ s_2, n_3 \}$, $\nu_{n_3} = \{ n_1 , n_2,d \}$ and $\nu_d = \{ n_3 , \mathit{tester} \}$.
Note that, in the same round,  the sender $s_2$ may 
broadcast to its neighbours $n_1$ and $n_2$, whereas  the sender
$s_1$ may broadcast (only) to its neighbour $n_1$, causing a failure at $n_1$.

As in the previous section, we will quantify the tolerance 
between the gossiping network $\mathrm{GSP}_5$ and the network
 $\mathrm{DONE}_5 = \mathrm{DONE}_2$, in which at least one instance of the message
$v$ has been successfully propagated up to the destination node.

As a first step, we compute an estimation of the tolerance which allows $\mathrm{GSP}_5$ to simulate $\mathrm{DONE}_5$, under the hypothesis that the sender
$s_1$  broadcasts to $n_1$ and
the sender $s_2$  
broadcasts to $n_1$ and $n_2$, causing a failure at $n_1$.
To this end,   
we derive the following chain of similarities by applying, 
in sequence:  
\begin{inparaenum}[(i)]
\item Proposition~\ref{thm:timing}(\ref{timing5})  to nodes $n_3$ and $d$, together to  Proposition~\ref{thm:timing}(\ref{timing34}) applied to node $n_2$, to express that these nodes do not receive any message because there are no transmitters in their neighbourhood, as a consequence, they are obliged to timeout;
\item
Theorem~\ref{thm:propagation3} 
in which node  $n_2$ acts as a sender and $n_3$ as a forwarder, together  
 Proposition~\ref{thm:timing}(\ref{timing2}) to work underneath $\sigma$ prefixes;
\item
Proposition~\ref{thm:timing}(\ref{timing1}) to remove a $\sigma$ prefix in $n_2$
together with  Proposition~\ref{thm:timing}(\ref{timing5}) as node $d$ cannot receive message and hence it timeouts; 
\item
Theorem ~\ref{thm:propagation3} in which node  $n_3$ acts as a sender and $d$ as a forwarder, together with Proposition~\ref{thm:timing}(\ref{timing2}) to work underneath $\sigma$ prefixes;
\item
Proposition~\ref{thm:timing}(\ref{timing1}) to remove  $\sigma$ prefixes in $n_3$, together with Proposition~\ref{thm:timing}(\ref{timing34}) as node $d$ is obliged to timeout.
In all steps, we reason up to parallel networks 
(Theorem~\ref{cor:non_expansiveness}).
\end{inparaenum}
For convenience, in the chain below we define two auxiliary networks:  $\mathrm{NET}_h$ and  $\mathrm{NET}_i$. 
\[
\begin{array}{rrl}
& & \mathrm{GSP}_5(s_1 \mbox{ gossiped}; s_2 \mbox{ gossiped}) 
\\[2pt]
& \deff & 
\prod_{i=1}^{2}\node {\nu_{s_i}}  {s_i} {\nil} \; \big| \;
\node {\nu_{n_1}}  {n_1} {\nil} \; \big| \;
\node {\nu_{n_2}}  {n_2} {\sndvc_{p}} \; \big| \;
 \node {\nu_{n_3}}  {n_3} {\fwdc_{p}} \; \big| \;
\node {\nu_d} d {\fwdc_1}
\\[3pt]
& \simtoll{0}& 
\prod_{i=1}^{2}\node {\nu_{s_i}}  {s_i} { \nil} \; \big| \;
\node {\nu_{n_1}}  {n_1} { \nil} \; \big| \;
\node {\nu_{n_2}}  {n_2} {\sigma. \sndv_{p}} \; \big| \;
\node {\nu_{n_3}}  {n_3} { \sigma.\fwdc_{p}} \; \big| \;
\node {\nu_d} d {\sigma. \fwdc_1}
\\[3pt]
& \deff & \mathrm{NET}_h
\\[3pt]
& \simtoll{1-p}& 
\prod_{i=1}^{2}\node {\nu_{s_i}}  {s_i} {\nil} \; \big| \;
\node {\nu_{n_1}}  {n_1} {\nil} \; \big| \;
\node {\nu_{n_2}}  {n_2} {\sigma. \nil} \; \big| \;
\node {\nu_{n_3}}  {n_3} { \sigma.{{\sndvc_{p}}}} \; \big| \;
\node {\nu_d} d { \sigma.\fwdc_1}
\\[3pt]
& \deff & \mathrm{NET}_i
\\[3pt]
& \simtoll{0}& 
\prod_{i=1}^{2}\node {\nu_{s_i}}  {s_i} { \nil} \; \big| \;
\prod^{2}_{i=1} \node {\nu_{n_i}}  {n_i} { \nil} \; \big| \;
 \node {\nu_{n_3}}  {n_3} { \sigma^2 \sndv_{p}} \; \big| \;
\node {\nu_d} d {\sigma^2. \fwdc_1}
\\[3pt]
& \simtoll{1-p}& 
\prod_{i=1}^{2}\node {\nu_{s_i}}  {s_i} {\nil} \; \big| \;
\prod^{2}_{i=1} \node {\nu_{n_i}}  {n_i} {\nil} \; \big| \;
 \node {\nu_{n_3}}  {n_3} {\sigma^2. \nil} \; \big| \;
\node {\nu_d} d {\sigma^2. {{\sndvc_1}}}
\\[2pt]
& \simtoll{0}& 
\prod_{i=1}^{2}\node {\nu_{s_i}}  {s_i} {\nil} \; \big| \;
\prod^{3}_{i=1} \node {\nu_{n_i}}  {n_i} {\nil} \; \big| \;
\node {\nu_d} d { \sigma^3 . \sndv_1} 
\\[3pt]
& =& 
\mathrm{DONE}_5  .
\end{array}
\]
Now, by transitivity (Proposition~\ref{prop:transitivity}) we can derive the following equation: 
\[
\mathrm{NET}_i
%%& 
\; \simtoll{1-p}  \; 
 \mathrm{DONE}_5 . 
\]
Thus, we are in condition  to apply  Proposition~\ref{lem:O_go_to_N_sim_M_conSigma}, in which $M$, $N$, and $O$ are instantiated with $\mathrm{DONE}_5$, $\mathrm{NET}_i$, and $\mathrm{NET}_h$, respectively, to derive the equation:
\[
\mathrm{NET}_h \; \simtoll{1-p^2} \; \, \mathrm{DONE}_5 
\]
where the tolerance $1- p^2$ is obtained by solving the expression $(1-p)(1-(1-p)) + (1-p)$. 
By transitivity  on the first   step  of the  chain above, we can finally derive the equation:
\begin{equation}
\label{eq_GSP3-1}
\mathrm{GSP}_5(s_1 \mbox{ gossiped}; s_2 \mbox{ gossiped}) \; \simtoll{1-p^2} \; \, 
\mathrm{DONE}_5  
\end{equation}
%%% promemoria per noi:
%%% 1-p^2 viene fuori da (1-p)*p + (1-p)

Similarly, we compute an estimation of the tolerance which allows 
  $\mathrm{GSP}_5$ to simulate $\mathrm{DONE}_5$, 
under the hypothesis that the sender $s_2$ broadcasts  to $n_1$ and $n_2$, and
the sender $s_1$ does not transmit. 
To this end, we derive the following chain of similarities by applying, 
in sequence:
\begin{inparaenum}[(i)]
\item Proposition~\ref{thm:timing}(\ref{timing34}) to timeout nodes $n_1$ and $n_2$, together with Proposition~\ref{thm:timing}(\ref{timing5}) to timeout nodes $n_3$ and $d$; 
\item
Theorem~\ref{thm:propagation3} in which $n_1$ and $n_2$ act as senders and 
$n_3$ as forwarder, together with  Proposition~\ref{thm:timing}(\ref{timing2}) 
to work underneath $\sigma$ prefixes; 
\item
Proposition~\ref{thm:timing}(\ref{timing1}) to remove $\sigma$ prefixes from $n_1$ and $n_2$, together with Proposition~\ref{thm:timing}(\ref{timing34}) to timeout $n_3$, together with Proposition~\ref{thm:timing}(\ref{timing5}) to timeout $d$,
together with  Proposition~\ref{thm:timing}(\ref{timing2}) 
to work underneath $\sigma$ prefixes; 
\item
Theorem~\ref{thm:propagation3} where $n_3$ acts as sender and $d$ as forwarder, together with  Proposition~\ref{thm:timing}(\ref{timing2}) to work underneath $\sigma$ prefixes;
\item
Proposition~\ref{thm:timing}(\ref{timing1}) to remove $\sigma$ prefixes in $n_3$, together with  Proposition~\ref{thm:timing}(\ref{timing34}) to timeout node $d$.
\end{inparaenum}
As usual, in all steps, we reason up to parallel composition 
(Theorem~\ref{cor:non_expansiveness}).
For convenience, in the chain below we define two auxiliary networks: $\mathrm{NET}_l$ and $\mathrm{NET}_m$. 
\[
\begin{array}{rrl}
&& \mathrm{GSP}_5(s_1 \mbox{ did not gossip}\mbox{; }s_2 \mbox{ gossiped})
\\[2pt]
& \deff &
\prod_{i=1}^{2}\node {\nu_{s_i}} {s_i} {\nil} \; \big| \;
\prod^{2}_{i=1} \node {\nu_{n_i}} {n_i} {\wsndvc_{p}} \; \big| \;
\node {\nu_{n_3}} {n_3} {\fwdc_{p}} \; \big| \;
\node {\nu_d} d {\fwdc_1}
\\[3pt]
&\simtoll{0}&
\prod_{i=1}^{2}\node {\nu_{s_i}} {s_i} { \nil} \; \big| \;
\prod^{2}_{i=1} \node {\nu_{n_i}} {n_i} {\sigma. \sndv_{p}} \; \big| \;
\node {\nu_{n_3}} {n_3} {\sigma. \fwdc_{p}} \; \big| \;
\node {\nu_d} d {\sigma.\fwdc_1}
\\[3pt]
& \deff & \mathrm{NET}_l
\\[3pt]
&\simtoll{1-q}&
\prod_{i=1}^{2}\node {\nu_{s_i}} {s_i} {\nil} \; \big| \;
\prod^{2}_{i=1} \node {\nu_{n_i}} {n_i} {\sigma. \nil} \; \big| \;
\node {\nu_{n_3}} {n_3} { \sigma.{\sndvc_{p}}} \; \big| \;
\node {\nu_d} d { \sigma.\fwdc_1}
\\[3pt]
& \deff & \mathrm{NET}_m
\\[3pt]
&\simtoll{0}&
\prod_{i=1}^{2}\node {\nu_{s_i}} {s_i} {\nil} \; \big| \;
\prod^{2}_{i=1} \node {\nu_{n_i}} {n_i} {\nil} \; \big| \;
\node {\nu_{n_3}} {n_3} {\sigma^2 .\sndv_{p}} \; \big| \;
\node {\nu_d} d {\sigma^2. \fwdc_1}
\\[3pt]
&\simtoll{1-p}&
\prod_{i=1}^{2}\node {\nu_{s_i}} {s_i} {\nil} \; \big| \;
\prod^{2}_{i=1} \node {\nu_{n_i}} {n_i} {\nil} \; \big| \;
\node {\nu_{n_3}} {n_3} {\sigma^2. \nil} \; \big| \;
\node {\nu_d} d {\sigma^2. {{\sndvc_1}}}
\\[3pt]
&\simtoll{0}&
\prod_{i=1}^{2}\node {\nu_{s_i}} {s_i} {\nil} \; \big| \;
\prod^{3}_{i=1} \node {\nu_{n_i}} {n_i} {\nil} \; \big| \;
\node {\nu_d} d {\sigma^3. \sndv_1}
\\[3pt]
& =&
\mathrm{DONE}_5
\end{array}
\]
with $q=2p(1{-}p)$.
By transitivity (Proposition~\ref{prop:transitivity}) we can derive the equation
\[
\mathrm{NET}_m
\; \simtoll{1-p} \; \, 
\mathrm{DONE}_5 . 
\]
Thus, we are in condition to apply Proposition~\ref{lem:O_go_to_N_sim_M_conSigma}, in which $M$, $N$, and $O$ are instantiated with $\mathrm{DONE}_5$, $\mathrm{NET}_m$, and $\mathrm{NET}_l$, respectively, to derive the equation 
\[
\mathrm{NET}_l
\; \simtoll{1 -2(p^2 -p^3)} \; \, 
\mathrm{DONE}_5
\]
where the tolerance $ 1 -2(p^2 -p^3)$ is obtained by solving the expression $(1-p)(1-(1-q)) + (1-q)$ with $q=2p(1{-}p)$.
By transitivity on the first step of the chain above, we can finally derive the equation 
\begin{equation}
\label{eq_GSP3-2}
\mathrm{GSP}_5(s_1 \mbox{ did not gossip}; s_2 \mbox{ gossiped})
 \q \simtoll{ 1 -2(p^2 -p^3)} \q \, 
\mathrm{DONE}_5  
\end{equation}
%
%%%promemoria per noi:
%%% 1 -2(p^2 -p^3) vien e fuori da 2p(1-p) * (1-p) +[1 - 2p(1-p)]

Now, let us denote with $\mathrm{GSP}_5(s_1[v] \mbox{; } s_2 \mbox{ gossiped})$ the following network: 
\[ 
\node {\nu_{s_1}}  {s_1} {\sndv_{p}} \; \big| \;
\node {\nu_{s_2}}  {s_2} {\nil} \; \big| \;
\prod^{2}_{i=1} \node {\nu_{n_i}}  {n_i} {{\sndvc_{p}}} \; \big| \;
\node {\nu_{n_3}}  {n_3} {\fwdc_{p}} \; \big| \;
\node {\nu_d} d {\fwdc_1}  . 
\]
This is the gossip network $\mathrm{GSP}_5$ in which the sender $s_2$ has already broadcast the message $v$ to its neighbours $n_1$ and $n_2$, and the sender $s_1$ has still  to decide whether to broadcast $v$ (to $n_1$) or not. We can apply  Corollary~\ref{cor:propagation2_coll} to Equation~\ref{eq_GSP3-1} and Equation~\ref{eq_GSP3-2} to derive 
\begin{equation} 
\label{eq_GSP3-3}
\mathrm{GSP}_5(s_1[v]\mbox{; } s_2 \mbox{ gossiped}) \;  \simtoll{r} \; \, 
\mathrm{DONE}_5 
\end{equation}
with tolerance $r=1- 2p^2 +3p^3 -2p^4$, 
obtained by solving $p(1-p^2) + (1-p)(1-2(p^2-p^3))$. 

%%%promemoria per noi:
%%%% 1- 2p^2 +3p^3 -2p^4 viene fuori da p(1-p^2) + (1-p)(1-2p^2+2p^3)
This gives us an estimation of the tolerance which
allows the network $\mathrm{GSP}_5$ to simulate 
the network $\mathrm{DONE}_5$ when 
the sender $s_2$ has already broadcast the message $v$ to its neighbours $n_1$ and $n_2$, whereas the sender $s_1$ has still  to decide whether to broadcast $v$ (to $n_1$) or not.

Similarly, we compute an estimation of the tolerance  which
allows the network $\mathrm{GSP}_5$ to simulate the network $\mathrm{DONE}_5$
when the sender $s_2$ has decided not to broadcast the message $v$,  while $s_1$ has still to decide  whether to broadcast $v$ or not.
To this end, we derive the following chain of similarities by applying, in sequence: 
\begin{inparaenum}[(i)]
\item
Theorem~\ref{thm:propagation3} where $s_1$ broadcasts its message and  $n_1$
receives it; 
\item Proposition~\ref{thm:timing}(\ref{timing34}) to timeout node $n_1$, together with 
Proposition~\ref{thm:timing}(\ref{timing5}) to timeout nodes $n_2$, $n_3$, and $d$; 
\item
Theorem~\ref{thm:propagation3} where $n_1$ acts as transmitter and $n_3$ as receiver, together with Proposition~\ref{thm:timing}(\ref{timing2}) to work underneath $\sigma$ prefixes;
\item
Proposition~\ref{thm:timing}(\ref{timing1}) to remove the $\sigma$ prefix in $n_1$, together with Proposition~\ref{thm:timing}(\ref{timing5}) to timeout nodes $n_2$ and $d$, and Proposition~\ref{thm:timing}(\ref{timing34}) to timeout $n_3$; 
\item
Theorem~\ref{thm:propagation3} in which $n_3$ transmits its value and both 
$n_2$ and $d$ receive it, together with Proposition~\ref{thm:timing}(\ref{timing2})
to work underneath $\sigma$ prefixes;
\item
Proposition~\ref{thm:timing}(\ref{timing1}) to remove $\sigma$ prefixes from $n_3$, 
together with Proposition~\ref{thm:timing}(\ref{timing34}) to timeout $n_2$;  
Proposition~\ref{thm:timing}(\ref{law:no-trans})  to model that the transmission of $n$ gets lost, together
with Proposition~\ref{thm:timing}(\ref{timing2}) to work underneath $\sigma$ prefixes.
\end{inparaenum}
As done before, in all steps we use Theorem~\ref{cor:non_expansiveness} to reason up to parallel components.
For convenience, in the chain below we define three auxiliary networks:  $\mathrm{NET}_n$,  $\mathrm{NET}_o$ and  $\mathrm{NET}_t$. 
\[
\begin{array}{rrl}
&& \mathrm{GSP}_5(s_1[v] \mbox{; } s_2 \mbox{ did not gossip})
\\[2pt]
& \deff &
\node {\nu_{s_1}}  {s_1} {\sndv_p} \; \big| \;
\node {\nu_{s_2}}  {s_2} {\nil} \; \big| \;
\prod_{i=1}^{3} \node {\nu_{n_i}}  {n_i} {\fwdc_{p}} \; \big| \;
\node {\nu_d} d {\fwdc_1} 
\\[3pt]
&\simtoll{1-p}&
\prod_{i=1}^2\node {\nu_{s_i}}  {s_i} {\nil} \; \big| \;
\node {\nu_{n_1}}  {n_1} {\sndvc_{p}} \; \big| \;
\prod^{3}_{j=2} \node {\nu_{n_j}}  {n_j} {\fwdc_{p}} \; \big| \;
\node {\nu_d} d {\fwdc_1}
\\[3pt]
& \deff & \mathrm{NET}_n
\\[3pt]
&\simtoll{0}&
\prod_{i=1}^{2}\node {\nu_{s_i}}  {s_i} { \nil} \; \big| \;
\node {\nu_{n_1}}  {n_1} {  \sigma.\sndv_{p}} \; \big| \;
\prod^{3}_{i=2} \node {\nu_{n_i}}  {n_i} { \sigma. \fwdc_{p}} \; \big| \;
\node {\nu_d} d {  \sigma.\fwdc_1}
\\[3pt]
& \deff & \mathrm{NET}_o
\\[3pt]
&\simtoll{1-p}&
\prod_{i=1}^{2}\node {\nu_{s_i}}  {s_i} {\nil} \; \big| \;
\node {\nu_{n_1}}  {n_1} { \sigma. \nil} \; \big| \;
\node {\nu_{n_2}}  {n_2} { \sigma. \fwdc_p} \; \big| \;
\node {\nu_{n_3}}  {n_3} { \sigma. \sndvc_{p}} \; \big| \;
\node {\nu_d} d { \sigma . \fwdc_1}
\\[3pt]
& \deff & \mathrm{NET}_t
\\[3pt]
&\simtoll{0}&
\prod_{i=1}^{2}\node {\nu_{s_i}} 
 {s_i} {\nil} \; \big| \;
\node {\nu_{n_1}}  {n_1} { \nil} \; \big| \;
\node {\nu_{n_2}}  {n_2} {\sigma^2.\fwdc_p} \; \big| \;
\node {\nu_{n_3}}  {n_3} {\sigma^2.\sndv_{p}} \; \big| \;
\node {\nu_d} d {\sigma^2.\fwdc_1}
\\[3pt]
&\simtoll{1-p}&
\prod_{i=1}^{2}\node {\nu_{s_i}}  {s_i} {\nil} \; \big| \;
 \node {\nu_{n_1}}  {n_1} {\nil} \; \big| \;
 \node {\nu_{n_2}}  {n_2} {\sigma^2 .\sndvc_p} \; \big| \;
\node {\nu_{n_3}}  {n_3} {  \sigma^2.\nil} \; \big| \;
\node {\nu_d} d {\sigma^2. \sndvc_1}
\\[3pt]
&\simtoll{0}&
\prod_{i=1}^{2}\node {\nu_{s_i}}  {s_i} { \nil} \; \big| \;
 \node {\nu_{n_1}}  {n_1} { \nil} \; \big| \;
 \node {\nu_{n_2}}  {n_2} { \sigma^3. \sndv_p} \; \big| \;
\node {\nu_{n_3}}  {n_3} { \nil} \; \big| \;
\node {\nu_d} d { \sigma^3.\sndv_1 }
\\[3pt]
&\simtoll{0}&
\prod_{i=1}^{2}\node {\nu_{s_i}}  {s_i} {\nil} \; \big| \;
\prod^{3}_{i=1} \node {\nu_{n_i}}  {n_i} {\nil} \; \big| \;
 \node {\nu_d} d {\sigma^3. \sndv_1}
\\[3pt]
&=& \mathrm{DONE}_5 \,  .
\end{array}
\]
By transitivity (Proposition~\ref{prop:transitivity}) we can derive the equation 
\[
\mathrm{NET}_t
\; \simtoll{1-p} \; \, 
 \mathrm{DONE}_5 . 
\]
This allows us  to apply  Proposition~\ref{lem:O_go_to_N_sim_M_conSigma}, in which $M$, $N$, and $O$ are instantiated with $\mathrm{DONE}_5$, $\mathrm{NET}_t$, and $\mathrm{NET}_o$, respectively, to  derive the equation 
\[
\mathrm{NET}_o
\; \simtoll{1- p^2} \; \,
 \mathrm{DONE}_5 
\]
where the tolerance $1- p^2$ is obtained by solving the expression $(1-p)(1-(1-p)) + (1-p)$. 
 Then, by transitivity (Proposition~\ref{prop:transitivity}) we can easily derive
\[
\mathrm{NET}_n
\; \simtoll{1- p^2} \; \,
 \mathrm{DONE}_5 . 
\] 
This last equation  can be used to apply  Proposition~\ref{lem:O_go_to_N_sim_M}, 
in which $M$, $N$, and $O$ are instantiated with $\mathrm{DONE}_5$, $\mathrm{NET}_n$, and 
$\mathrm{GSP}_5(s_1[v]\mbox{; } s_2 \mbox{ did not gossip})$, respectively,  to  derive:
\begin{equation}
\label{eq_GSP3-4}
\begin{array}{c}
\mathrm{GSP}_5(s_1[v] \mbox{; } s_2 \mbox{ did not gossip})  \simtoll{1-p^3} \q\,
\mathrm{DONE}_5 
\end{array}
\end{equation}
where the tolerance $1- p^3$ is obtained by solving the expression $(1-p)(1-(1-p^2)) + (1-p^2)$. 

%%% promemoria per noi
%%% 1-p^3 viene fuori come segue. prima calcoliamo 1-p^2 come in un promemoria precedente.
%%% poi: (1-p^2)*p + (1-p)

Now, by applying Corollary~\ref{cor:propagation2_coll_bis} to Equation~\ref{eq_GSP3-3} and Equation~\ref{eq_GSP3-4} we can combine the two probabilistic behaviours of the sender $s_2$, thus obtaining: 
\[
\mathrm{GSP}_5 \; \simtoll{r}\; \,  \mathrm{DONE}_5 
\]
with tolerance $r=1-(3p^3-4p^4+2p^5)$ obtained by solving $p(1-2p^2+3p^3-2p^4) +(1-p)(1-p^3)$.
%%% promemoria per noi
%%%% 1-(3p^3-4p^4+2p^5) viene fuori da p(1-2p^2+3p^3-2p^4) +(1-p)(1-p^3)
This says that the gossip network $\mathrm{GSP}_5$ will succeed in transmitting the message $v$ to the destination $d$ with probability (at least) $3p^3-4p^4+2p^5$. 

In the previous section, we have seen that in 
the absence of collisions, 
the success probability of $\mathrm{GSP}_2$ is (at least)  $3p^3 -2p^4$. 
Since  $3p^3 -4p^4 +2p^5 < 3p^3 - 2p^4$, for any $p> 0$, it follows that the 
 presence of collisions  downgrades the performance of
this small network of $2p^4 - 2p^5$. 
Thus, for instance, for a gossip probability $p=0.8$ the destination node 
in $\mathrm{GSP}_5$ will receive the message with probability (at least) 
$0.55$,  whereas in $\mathrm{GSP}_2$ this probability is $0.716$; similarly for $p=0.9$ the probability of success in $\mathrm{GSP}_5$ is about $0.74$ while in $\mathrm{GSP}_2$ it is $0.88$. In bigger networks, with more collisions, the degradation 
will be much more evident. 

\begin{rem}
Due to its general formulation, Theorem~\ref{thm:propagation2bis}, dealing with message propagation 
in paths of different lengths, is valid also in the presence of communication collisions.
\end{rem}

%%%%%%%%%%%%%%%%%%%%%%%%%%%%%%%%%%%%%%%%%

\section{Random delays to mitigate the effects of collisions}
\label{sec:random}
In this section, we show how \cname{} can be used to  model and reason on
 \emph{randomised gossip protocols\/}, where messages may be broadcast  
 in different time instants,  according to some probability distribution (see, \emph{e.g.}, \cite{ansgar2006}). 
In this manner, if there is a potential  collision at a given round 
there are still chances that the message will reach the final destination 
in the following rounds, improving the probability of success of the gossip protocol. 

Randomisation may be implemented in different ways. 
For instance, one may decide to broadcast messages according
to a \emph{uniform  probability distribution} within a \emph{discrete} time 
interval $1{..}k$, for $k \in \mathbb{N}$. This means that 
the message may be transmitted in each of the $k$ time instants, with a 
probability $\frac{1}{k}$. 
In such a scenario, both sender and forwarder processes should be 
reformulated as follows:
\[
\begin{array}{rcl}
{\sndu v}_{p_{\mathrm{g}},k}
& \deff & \tau.\big(\tau.(\bigoplus_{i=1}^k \frac{1}{k}{\colon}\sigma^{i}. \sndv_{1} )
 \oplus_{p_{\mathrm{g}}} \nil\big)\\[1pt]
\wsndvu_{p_{\mathrm{g}},k}
& \deff & \rcvtime y \nil {{\sndu v}_{p_{\mathrm{g}},k}} \\[1pt]
\fwdu_{p_{\mathrm{g}},k} & \deff &
{\rcv x { \mathsf{resndu}\langle x\rangle_{p_{\mathrm{g}},k}}} .
\end{array}
\]

The main algebraic instrument to reason on the performance of the protocol remains message propagation. 
The following result is a straightforward reformulation of Theorem~\ref{thm:propagation3}. In Corollary~\ref{thm:propagationRandomDelays}, for simplicity, 
we have only one sender that broadcasts with probability $1$. 
This sender  derives by the resolution of the 
probabilistic choice within any term ${\sndu v}_{p_{\mathrm{g}},k}$.
\begin{cor}[Message propagation with collisions and random delays]
\label{thm:propagationRandomDelays}
Let $J$ be a subset of $\mathbb{N}$.
Let $M$ be a well-formed network defined by
\[
M \equiv  
N \q \big| \q
 \node {\nu_{m}} {m} {\sndv_{1}} \q \big| \q
\prod_{j \in J } \node {\nu_{n_j}} {n_j} {\fwdu_{q_j ,k}}
\]
such that:
\begin{inparaenum}[(1)]
\item 
 $\{n_j \mid j \in J\} \subseteq \nu_{m} \subseteq  \nds{M}$, and 
\item 
 the nodes in $\nu_{m} \cap \nds{N}$ cannot receive in the current round.
\end{inparaenum}
Then, 
\[
M
\;
\simtoll{0}
\; 
N \; | \;  \node {\nu_{m}} {m} {\nil} \; | \;
 \prod_{j \in J} \node {\nu_{n_j}} {n_j} {\wsndvu_{q_j,k}} .
\]
\end{cor}

In the rest of this section, we will 
apply Corollary~\ref{thm:propagationRandomDelays} to reason on a variant of the
network $\mathrm{GSP}_1$ of Section~\ref{sec:gossip-nocollisions} where random delays are introduced. 

Consider  a gossiping network $\mathrm{GSP}_6$, with gossip probability $p$, composed by two source nodes $s_1$ and $s_2$ that may 
delay their transmission within a discrete time interval $1{..}2$, and a destination node
 $d$. The network is the following: 
\[
\mathrm{GSP}_6 \q \deff \q 
\prod_{i=1}^{2}\node {\nu_{s_i}}  {s_i} {\sndu v_{p,2}} \q \big| \q
\node {\nu_d} d {\finer}
\]
where 
$\nu_{s_1} = \nu_{s_2} = \{d\}$ and 
$\nu_{d}{=} \{ s_1, s_2 ,\mathit{tester}\}$.

Our goal is to estimate the distance between  $\mathrm{GSP}_6$ and a network  $\mathrm{DONE}_6$  in which the message $v$ propagated  up to the destination node $d$ that will rebroadcast  the message according
to some probabilistic delay depending on the probability $p$. 
\[
\label{eq:GSP4_sure}
\mathrm{DONE}_6  \;  \deff \; 
\prod_{i=1}^{2}\node {\nu_{s_i}}  {s_i} {\nil} \q \big| \q
 \node {\nu_d} d {\tau.(\sigma^3.\sndv_{1} \, \oplus_{(\frac{1}{2} {+} \frac{p^2}{2})}  \,  \sigma^4.\sndv_1)} . 
\]
At the end of this section,  we will show that  
$\mathrm{GSP}_6 \simtoll{r}     \mathrm{DONE}_6 $,  with tolerance $r=1 - (2p - \frac{3}{2} p^2 )$. 
In order to estimate the tolerance $r$ we distinguish three possible 
cases, depending whether the sender nodes $s_1$ and $s_2$ eventually 
transmit or not. 
More precisely, we estimate the tolerance of the simulation in the following cases: 
\begin{inparaenum}[(a)]
\item
\label{1SI2NO}
node $s_1$ transmits while $s_2$ does not,
\item
\label{1NO2SI}
node $s_2$ transmits while $s_1$ does not,
\item
\label{1SI2SI}
both nodes $s_1$ and $s_2$ transmit.
\end{inparaenum}
Notice that we do not analyse the case when neither $s_1$ nor $s_2$ transmits, since such a behaviour does not allow  $\mathrm{GSP}_6$ to simulate $\mathrm{DONE}_6$.
Then, by applying Theorem~\ref{thm:propagation2bis} to compose different paths we combine the tolerances obtained in cases (\ref{1NO2SI}) and (\ref{1SI2SI}) to get the tolerance when $s_2$ transmits independently whether $s_1$ transmits or not.
%%and independently of the probabilistic choice by $s_1$ to transmit or not.
Finally,  we use again Theorem~\ref{thm:propagation2bis} to combine this tolerance  with that obtained from case  (\ref{1SI2NO}) to infer $r$.

Let us analyse these three different cases.
\subsubsection*{(a) Node $s_1$ transmits while $s_2$ does not.}
In this case, there are obviously no collisions and the message will 
eventually reach the destination.
 Now, if $s_1$ decides to transmit with a delay of $h$ time units,  for $h \in \{1,2\}$, then, by applying in sequence: 
\begin{inparaenum}[(i)]
 \item Proposition~\ref{thm:timing}(\ref{timing5}) to timeout node $d$, as for $h$ rounds it will not get any message; 
 \item  Corollary~\ref{thm:propagationRandomDelays} to propagate the message from the sender $s_1$ to the node $d$,  together with  Proposition~\ref{thm:timing}(\ref{timing2}) to work underneath $\sigma$ prefixes; 
 \item   Proposition~\ref{thm:timing}(\ref{timing1}) to remove $\sigma$ prefixes in $s_1$; 
\item  Proposition~\ref{thm:timing}(\ref{timing34}) together with the straightforward equation
$\node {\nu_d} d { {\sndu v}_{1,1} } \simtoll{0} \node {\nu_d} d { {\sigma. {\sndv}_1}  }  $;
\end{inparaenum} it follows that: 
\begin{center}
\begin{math}
\begin{array}{rrl}
& & 
\mathrm{GSP}_6(s_1 \mbox{ sends in } h \mbox{ rounds; }s_2 \mbox{ did not gossip})\\[2pt]
& \deff & 
\node {\nu_{s_1}}  {s_1} { \sigma^h.\sndv_{1}  } \; \big| \;
\node {\nu_{s_2}}  {s_2} {\nil} \; \big| \;
\node {\nu_d} d {\finer} 
\\[3pt]
&\simtoll{0}&
\node {\nu_{s_1}}  {s_1} { \sigma^h.\sndv_{1}  } \; \big| \;
\node {\nu_{s_2}}  {s_2} {\nil} \; \big| \;
\node {\nu_d} d { \sigma^h.\finer} 
\\[3pt]
& \simtoll{0} & 
\node {\nu_{s_1}}  {s_1} { \sigma^h.\nil   } \; \big| \;
\node {\nu_{s_2}}  {s_2} {\nil} \; \big| \;
\node {\nu_d} d { \sigma^h. \wsndvu_{1,1}} 
\\[3pt]
&\simtoll{0}&
\node {\nu_{s_1}}  {s_1} { \nil   } \; \big| \;
\node {\nu_{s_2}}  {s_2} {\nil} \; \big| \;
\node {\nu_d} d { \sigma^h. \wsndvu_{1,1}} 
\\[3pt]
&\simtoll{0}&
\prod_{i=1}^{2}\node {\nu_{s_i}}  {s_i} {\nil} \; \big| \;
 \node {\nu_d} d { \sigma^{h+2}.\sndv_{1} }   .
\end{array}
\end{math}
\end{center}
By transitivity  (Proposition~\ref{prop:transitivity}) we can derive:  
\begin{equation}
\label{eq_GSP4-1}
\mathrm{GSP}_6(s_1 \mbox{ sends in } h \mbox{ rounds; }s_2 \mbox{ did not gossip})
\; 
 \simtoll{0}  \; 
\prod_{i=1}^{2}\node {\nu_{s_i}}  {s_i} {\nil} \q \big| \q 
 \node {\nu_d} d { \sigma^{h+2}.\sndv_{1} } 
\end{equation}

We know that  $s_1$ transmits in $\mathrm{GSP}_6$ according to a uniform distribution probability in the time interval $1{..}2$. Thus, by an application
of Theorem~\ref{thm:propagation2bis}, 
we can combine the two instances of Equation~\ref{eq_GSP4-1}, for $h=1$ and $h=2$, to obtain:
\begin{equation}
\label{eq_GSP4-9}
\begin{array}{rl}
& \mathrm{GSP}_6(s_1 \mbox{ sends in $1$ or $2$ rounds; }s_2 \mbox{ did not gossip})\\[2pt]
\deff & 
\node {\nu_{s_1}}  {s_1} {\tau.(\sigma.\sndv_{1}  \oplus_{\frac{1}{2}}  \sigma^2.\sndv_1)} \; \big| \;
\node {\nu_{s_2}}  {s_2} {\nil} \; \big| \;
\node {\nu_d} d {\finer}  \\[3pt]
 \simtoll{0} &
\prod_{i=1}^{2}\node {\nu_{s_i}}  {s_i} {\nil} \; \big| \;
 \node {\nu_d} d {\tau.(\sigma^3.\sndv_{1}  \oplus_{\frac{1}{2}}  \sigma^4.\sndv_1)} . 
\end{array}
\end{equation} 
The complete behaviour of $s_1$ can be derived by the following chain of similarities,  obtained by an application of Proposition~\ref{prop:first-laws}(\ref{random}) and Equation~\ref{eq_GSP4-9}: 
\[
\begin{array}{rrl}
&& \mathrm{GSP}_6(s_1[v] \mbox{; } s_2 \mbox{ did not gossip})
\\[2pt]
& \deff &\node {\nu_{s_1}}  {s_1} {\sndu v_{p,2} } \; \big| \;
\node {\nu_{s_2}}  {s_2} {\nil} \; \big| \;
\node {\nu_d} d {\finer} 
\\[3pt]
&\simtoll{1-p}&
\node {\nu_{s_1}}  {s_1} {\tau.(\sigma.\sndv_{1}  \oplus_{\frac{1}{2}}  \sigma^2.\sndv_1)} \; \big| \;
\node {\nu_{s_2}}  {s_2} {\nil} \; \big| \;
\node {\nu_d} d {\finer} 
\\[3pt]
&\simtoll{0}&
\prod_{i=1}^{2}\node {\nu_{s_i}}  {s_i} {\nil} \; \big| \;
 \node {\nu_d} d {\tau.(\sigma^3.\sndv_{1}  \oplus_{\frac{1}{2}}  \sigma^4.\sndv_1)} .
\end{array}
\]
Finally, by Proposition~\ref{prop:transitivity} we get:  
\begin{equation}
\label{eq_GSP4-10}
\mathrm{GSP}_6(s_1[v] \mbox{; } s_2 \mbox{ did not gossip})
\:
 \simtoll{1-p} \: 
\prod_{i=1}^{2}\node {\nu_{s_i}}  {s_i} {\nil} \; \big| \;
 \node {\nu_d} d {\tau.(\sigma^3.\sndv_{1}  \oplus_{\frac{1}{2}}  \sigma^4.\sndv_1)} 
\end{equation}

Let us consider now the symmetric case.  
\subsubsection*{(b) Node $s_2$ transmits while $s_1$ does not.}
This case can be arranged  by switching the roles of the nodes $s_1$ and 
$s_2$ in  Equation~\ref{eq_GSP4-9}: 
\begin{equation}
\label{eq_GSP4-14}
\begin{array}{rl}
& \mathrm{GSP}_6(\mbox{$s_1$ did not gossip; $s_2$ sends in $1$ or $2$ rounds})
\\[2pt] 
\deff & \node {\nu_{s_1}}  {s_1} {\nil} \; \big| \;
\node {\nu_{s_2}}  {s_2} {\tau.(\sigma.\sndv_{1}  \oplus_{\frac{1}{2}}  \sigma^2.\sndv_1)} \; \big| \;
\node {\nu_d} d {\finer} \\[3pt]
 \simtoll{0}  &
\prod_{i=1}^{2}\node {\nu_{s_i}}  {s_i} {\nil} \; \big| \;
 \node {\nu_d} d {\tau.(\sigma^3.\sndv_{1}  \oplus_{\frac{1}{2}}  \sigma^4.\sndv_1)} 
\end{array}
\end{equation}

Finally, let us consider the last case. 
\subsubsection*{(c) Both nodes $s_1$ and $s_2$ transmit.}
If both senders decide to transmit in the same time instant then there will be
 a collision and the transmitted message will never reach the destination. 
If this happens then the network $\mathrm{GSP}_6$ cannot simulate a network 
in which the message has reached the destination node $d$. 
On the other hand, if $s_1$ broadcasts with a delay of one time unit, 
 and $s_2$ with a delay of two time units, then there will be  no collisions and the message may eventually reach the destination.

Thus, by an application  of: 
\begin{inparaenum}[(i)]
\item 
Proposition~\ref{thm:timing}(\ref{timing5}) to timeout node $d$ in the absence 
of a sender;  
\item  Corollary~\ref{thm:propagationRandomDelays} to broadcast the message of
$s_1$ to $d$,  together with Proposition~\ref{thm:timing}(\ref{timing2}) to work underneath $\sigma$ prefixes; 
\item Proposition~\ref{thm:timing}(\ref{timing34}) to timeout node $d$, as after one round there are no senders around; 
\item  the straightforward equation
$\node {\nu_d} d { {\sndu v}_{1,1} } \simtoll{0} \node {\nu_d} d { {\sigma. {\sndv}_1}  }  $, together with Proposition~\ref{thm:timing}(\ref{timing2}) to work underneath $\sigma$ prefixes;
\item 
Proposition~\ref{thm:timing}(\ref{law:no-trans}) to model that the message broadcast by $s_2$ gets lost, 
together with Proposition~\ref{thm:timing}(\ref{timing2}) to work underneath $\sigma$ prefixes;
\item   Proposition~\ref{thm:timing}(\ref{timing1}) to remove $\sigma$ prefixes in $s_1$ and $s_2$;
\end{inparaenum} it follows that: 
\[
\begin{array}{rl}
& \mathrm{GSP}_6(\mbox{$s_1$ sends in $1$ round; $s_2$ sends in $2$ rounds})
\\[1pt]
\deff & 
\node {\nu_{s_1}}  {s_1} {\sigma.\sndv_{1}  } \q \big| \q
\node {\nu_{s_2}}  {s_2} {\sigma^2.\sndv_{1}  } \q \big| \q
\node {\nu_d} d {\finer}  
\end{array}
\]
\[
\begin{array}{rl}
 \simtoll{0} &  
\node {\nu_{s_1}}  {s_1} {\sigma.\sndv_{1}  } \q \big| \q
\node {\nu_{s_2}}  {s_2} {\sigma^2.\sndv_{1}  } \q \big| \q
\node {\nu_d} d {\sigma.\finer} 
 \\[3pt]
 \simtoll{0} &  
\node {\nu_{s_1}}  {s_1} {\sigma. \nil  } \q \big| \q
\node {\nu_{s_2}}  {s_2} {\sigma^2.\sndv_{1}  } \q \big| \q
\node {\nu_d} d {\sigma.   \wsndvu  _{1,1}}   
\\[3pt]
 \simtoll{0} &  
\node {\nu_{s_1}}  {s_1} {\sigma. \nil  } \q \big| \q
\node {\nu_{s_2}}  {s_2} {\sigma^2.\sndv_{1}  } \q \big| \q
\node {\nu_d} d {\sigma. \sigma.  {\sndu v}_{1,1}}   \\[3pt]
 \simtoll{0} &  
\node {\nu_{s_1}}  {s_1} {\sigma. \nil  } \q \big| \q
\node {\nu_{s_2}}  {s_2} {\sigma^2.\sndv_{1}  } \q \big| \q
\node {\nu_d} d {\sigma. \sigma. \sigma . {\sndv}_{1}}   \\[3pt]
 \simtoll{0} & 
\node {\nu_{s_1}}  {s_1} {\sigma. \nil  } \q \big| \q
\node {\nu_{s_2}}  {s_2} {\sigma^2.\nil  } \q \big| \q
 \node {\nu_d} d { \sigma^3.\sndv_{1} }   \\[3pt]
  \simtoll{0} & 
\prod_{i=1}^{2}\node {\nu_{s_i}}  {s_i} {\nil} \q \big| \q
 \node {\nu_d} d { \sigma^3.\sndv_{1} }   .
\end{array}
\]
By transitivity  (Proposition~\ref{prop:transitivity}) we derive:
\begin{equation}
\label{eq_GSP4-7}
 \mathrm{GSP}_6(\mbox{$s_1$ sends in $1$ round; $s_2$ sends in $2$ rounds}) \;
 \simtoll{0} \; \, 
\prod_{i=1}^{2}\node {\nu_{s_i}}  {s_i} {\nil} \q \big| \q
 \node {\nu_d} d { \sigma^3.\sndv_{1} }  
\end{equation}  

As nodes $s_1$ and $s_2$ have the same neighbours,  the symmetric case, when $s_1$ broadcasts with a delay of two rounds and $s_2$ with a delay of one round, can be captured by switching the roles of $s_1$ and $s_2$ in Equation~\ref{eq_GSP4-7}: 
\begin{equation}
\label{eq_GSP4-8}
\begin{array}{rl}
& \mathrm{GSP}_6(\mbox{$s_1$ sends in $2$ rounds ; $ s_2$ sends in $1$ round})
\\[2pt]
 \deff &
\node {\nu_{s_1}}  {s_1} {\sigma^2.\sndv_{1} } \q \big| \q
\node {\nu_{s_2}}  {s_2} {\sigma.\sndv_{1}  } \q \big| \q
\node {\nu_d} d {\finer} \\[3pt]
   \simtoll{0} & 
\prod_{i=1}^{2}\node {\nu_{s_i}}  {s_i} {\nil} | 
 \node {\nu_d} d { \sigma^3.\sndv_{1} }  . 
\end{array}
\end{equation}

Next, we compose the previous estimations to evaluate the tolerance when both 
senders transmit with a probabilistic delay. 
Please, notice that
the following  law holds for any network $N$ ($1$ is the maximum distance):
\begin{equation*}
\label{eq_GSP4-5} 
\node {\nu_{s_1}}  {s_1} {\sigma^h.\sndv_{1}  } \q \big| \q
\node {\nu_{s_2}}  {s_2} {\sigma^h.\sndv_{1}  } \q \big| \q
\node {\nu_d} d {\finer} 
\: \simtoll{1 } \: N  
\end{equation*}
In the rest of this section, we will use the following instances of the law above:

\begin{equation}
\label{eq_GSP4-5-instance1}
\begin{array}{c}
\node {\nu_{s_1}}  {s_1} {\sigma.\sndv_{1}  } \: \big| \:
\node {\nu_{s_2}}  {s_2} {\sigma.\sndv_{1}  }\: \big| \:
\node {\nu_d} d {\finer} 
\; \simtoll{1} \;    \prod_{i=1}^{2}\node {\nu_{s_i}}  {s_i} {\nil} \: \big| \:
 \node {\nu_d} d { \sigma^3.\sndv_{1} }  
\end{array}
\end{equation}

\begin{equation}
\label{eq_GSP4-5-instance2}
\begin{array}{c}
\node {\nu_{s_1}}  {s_1} {\sigma^2.\sndv_{1}  }  \big| \,
\node {\nu_{s_2}}  {s_2} {\sigma^2.\sndv_{1}  } \big| \,
\node {\nu_d} d {\finer} 
\, \simtoll{1} \,   \prod_{i=1}^{2}\node {\nu_{s_i}}  {s_i} {\nil} \, \big| \,
 \node {\nu_d} d { \sigma^3.\sndv_{1} } 
\end{array}
\end{equation}

Thus,  in the  case when $s_1$ broadcasts with a delay of one round and $s_2$ 
with a probabilistic delay of one or two rounds, by applying, in sequence: 
\begin{inparaenum}[(i)]
\item Theorem~\ref{thm:propagation2bis} in which the network $M$ is 
structural congruent to $\node {\nu_{s_1}}  {s_1} {\sigma.\sndv_{1}  } | \node {\nu_d} d {\finer}$, and the node $m$ is instantiated as $s_2$, 
to compose the paths deriving from Equation~\ref{eq_GSP4-7} and Equation~\ref{eq_GSP4-5-instance1};  
\item
the  straightforward law
$ \node{\nu_d}{d} {\tau.(\sigma^3.\sndv_{1}  \oplus_{\frac{1}{2}}  \sigma^3.\sndv_1)}  \simtoll{0}  \node{\nu_d}{d}{\sigma^3.\sndv_{1}}  $;
\end{inparaenum}
we derive: 
\[
\begin{array}{rl}
& \mathrm{GSP}_6(\mbox{$s_1$ sends in $1$ round; $s_2$ sends in $1$ or $2$ rounds}) 
\\[2pt]
\deff & 
\node {\nu_{s_1}}  {s_1} { \sigma.\sndv_{1}  } \q \big| \q
\node {\nu_{s_2}}  {s_2} {\tau.(\sigma.\sndv_{1}  \oplus_{\frac{1}{2}}  \sigma^2.\sndv_1)} \q \big| \q
\node {\nu_d} d {\finer} \\[3pt]
\simtoll{\frac{1}{2}} & 
\prod_{i=1}^{2}\node {\nu_{s_i}}  {s_i} {\nil} \q \big| \q
 \node {\nu_d} d {\tau.(\sigma^3.\sndv_{1}  \oplus_{\frac{1}{2}}  \sigma^3.\sndv_1)}   \\[3pt] 
\simtoll{0} & 
\prod_{i=1}^{2}\node {\nu_{s_i}}  {s_i} {\nil} \q \big| \q
\node {\nu_d} d {\sigma^3.\sndv_{1} }  .
\end{array}
\]
By transitivity (Proposition~\ref{prop:transitivity})
we obtain 
\begin{equation}
\label{eq_GSP4-11}
\begin{array}{c}
\mathrm{GSP}_6(\mbox{\small $s_1$ sends in $1$ rnd; $s_2$ sends in $1$ or $2$ rnds}) 
\: \simtoll{\frac{1}{2}} \: 
\prod_{i=1}^{2}\node {\nu_{s_i}}  {s_i} {\nil} | 
\node {\nu_d} d {\sigma^3.\sndv_{1} }  
\end{array}
\end{equation}

Analogously, in  the case when $s_1$ broadcasts with a delay of two rounds and $s_2$ with a probabilistic delay of one or two rounds, we can apply
in sequence:
\begin{inparaenum}[(i)]
\item Theorem~\ref{thm:propagation2bis} in which the network $M$ is 
structural congruent to $\node {\nu_{s_1}}  {s_1} {\sigma^2.\sndv_{1}  } | \node {\nu_d} d {\finer}$ and the node $m$ is instantiated as $s_2$, 
to compose the paths deriving from Equation~\ref{eq_GSP4-8} and Equation~\ref{eq_GSP4-5-instance2}; 
\item
the  straightforward law
$ \node{\nu_d}{d} {\tau.(\sigma^3.\sndv_{1}  \oplus_{\frac{1}{2}}  \sigma^3.\sndv_1)}  \simtoll{0}  \node{\nu_d}{d}{\sigma^3.\sndv_{1}}  $;
\end{inparaenum}
we derive: 
\[
\begin{array}{rl}
& \mathrm{GSP}_6(\mbox{$s_1$ sends in $2$ rounds; $s_2$ sends in $1$ or $2$ rounds}) \\
\deff &
\node {\nu_{s_1}}  {s_1} { \sigma^2.\sndv_{1}  } \; \big| \;
\node {\nu_{s_2}}  {s_2} {\tau.(\sigma.\sndv_{1}  \oplus_{\frac{1}{2}}  \sigma^2.\sndv_1)} \; \big| \;
\node {\nu_d} d {\finer} \\[3pt]
 \simtoll{\frac{1}{2}} &
 \prod_{i=1}^{2}\node {\nu_{s_i}}  {s_i} {\nil} \q \big| \q
 \node {\nu_d} d {\tau.(\sigma^3.\sndv_{1}  \oplus_{\frac{1}{2}}  \sigma^3.\sndv_1)}   \\[3pt] 
\simtoll{0} & 
\prod_{i=1}^{2}\node {\nu_{s_i}}  {s_i} {\nil} \; \big| \;
 \node {\nu_d} d { \sigma^3.\sndv_{1} }  .
\end{array}
\]
By transitivity (Proposition~\ref{prop:transitivity})
we obtain
\begin{equation}
\label{eq_GSP4-12}
\mathrm{GSP}_6(\mbox{\small $s_1$ sends in $2$ rnds; $s_2$ sends in $1$ or $2$ rnds})  
 \: \simtoll{\frac{1}{2}} \;  \prod_{i=1}^{2}\node {\nu_{s_i}}  {s_i} {\nil} \; \big| \;
 \node {\nu_d} d { \sigma^3.\sndv_{1} }  
\end{equation}

Thus, to estimate the tolerance 
 when both $s_1$ and $s_2$ broadcast with a probabilistic delay of one or two steps,  
we can apply  in sequence: 
\begin{inparaenum}[(i)]
\item Theorem~\ref{thm:propagation2bis} in which the network $M$ is 
structural congruent to $\node {\nu_{s_2}}  {s_2} {\tau.(\sigma.\sndv_{1}  \oplus_{\frac{1}{2}}  \sigma^2.\sndv_1)}   | 
\node {\nu_d} d {\finer} $ and the node $m$ is instantiated as $s_1$, 
to compose the paths deriving from Equation~\ref{eq_GSP4-11} and Equation~\ref{eq_GSP4-12};
\item
the  straightforward law
$ \node{\nu_d}{d} {\tau.(\sigma^3.\sndv_{1}  \oplus_{\frac{1}{2}}  \sigma^3.\sndv_1)}  \simtoll{0}  \node{\nu_d}{d}{\sigma^3.\sndv_{1}}  $;
\end{inparaenum}
 to obtain: 
\[
\begin{array}{rl}
& \mathrm{GSP}_6(\mbox{both $s_1$ and $s_2$ send in $1$ or $2$ rounds})
\\[2pt]
\deff &
\prod_{i=1}^{2}\node {\nu_{s_i}}  {s_i} {\tau.(\sigma.\sndv_{1}  \oplus_{\frac{1}{2}}  \sigma^2.\sndv_1)} \q \big| \q
\node {\nu_d} d {\finer}  \\[3pt]
 \simtoll{\frac{1}{2}} &
 \prod_{i=1}^{2}\node {\nu_{s_i}}  {s_i} {\nil} \q \big| \q
 \node {\nu_d} d {\tau.(\sigma^3.\sndv_{1}  \oplus_{\frac{1}{2}}  \sigma^3.\sndv_1)}   \\[3pt] 
\simtoll{0} & 
\prod_{i=1}^{2}\node {\nu_{s_i}}  {s_i} {\nil} \q \big| \q
 \node {\nu_d} d {\sigma^3.\sndv_{1}  }  .
\end{array}
\]
By transitivity (Proposition~\ref{prop:transitivity})
we obtain 
\begin{equation}
\label{eq_GSP4-13}
\mathrm{GSP}_6(\mbox{\small both $s_1$ and $s_2$ send in $1$ or $2$ rounds})
\: \simtoll{\frac{1}{2}}  \; 
\prod_{i=1}^{2}\node {\nu_{s_i}}  {s_i} {\nil}  \; \big| \; 
 \node {\nu_d} d {\sigma^3.\sndv_{1}  }  
\end{equation}

By exploiting the three cases \emph{(a), (b), (c)} analysed above, we can calculate now the tolerance $r$ of $\mathrm{GSP}_6 \, \simtoll{r} \, \mathrm{DONE}_6$. 
First, to estimate the tolerance 
when $s_2$ transmits (with a probabilistic delay of one or two steps) independently whether $s_1$ transmits or not,  
we can apply: 
\begin{inparaenum}[(i)]
\item Theorem~\ref{thm:propagation2bis} in which the network $M$ is 
structural congruent to $\node {\nu_{s_2}}  {s_2} {\tau.(\sigma.\sndv_{1}  \oplus_{\frac{1}{2}}  \sigma^2.\sndv_1)}   | 
\node {\nu_d} d {\finer}  $ and the node $m$ is instantiated as $s_1$, 
to compose the paths deriving from  Equation~\ref{eq_GSP4-14} and Equation~\ref{eq_GSP4-13};
\item 
 Proposition~\ref{prop:first-laws}(\ref{random3}) 
\item both additivity and commutativity properties of  $\oplus$;
\end{inparaenum}
 to obtain:     
\[
\begin{array}{rl}
& \mathrm{GSP}_6(\mbox{$s_1[v]$; $s_2$ sends in $1$ or $2$ rounds})
\\[2pt]
\deff & \node {\nu_{s_1}}  {s_1} {\sndu v_{p,2} } \q \big| \q
\node {\nu_{s_2}}  {s_2} {\tau.(\sigma.\sndv_{1}  \oplus_{\frac{1}{2}}  \sigma^2.\sndv_1)} \q \big| \q
\node {\nu_d} d {\finer} \\[3pt]
 \simtoll{ \frac{p}{2}} & 
  \prod_{i=1}^{2}\node {\nu_{s_i}}  {s_i} {\nil} \q \big| \q
 \node {\nu_d} d {\tau.\big(\sigma^3.\sndv_{1}  \oplus_{p} \tau.(\sigma^3.\sndv_{1}  \oplus_{\frac{1}{2}}  \sigma^4.\sndv_1)\big)}   \\[3pt] 
\simtoll{0} & 
  \prod_{i=1}^{2}\node {\nu_{s_i}}  {s_i} {\nil} \q \big| \q
 \node {\nu_d} d {\tau.\big(p{\colon}\sigma^3.\sndv_{1}  \: \oplus \:  {\frac{1}{2}}(1-p){\colon}\sigma^3.\sndv_{1}  \: \oplus \:  
 {\frac{1}{2}} (1-p){\colon} \sigma^4.\sndv_1 \big)}   \\[3pt] 
\simtoll{0} & 
\prod_{i=1}^{2}\node {\nu_{s_i}}  {s_i} {\nil} \q \big| \q
 \node {\nu_d} d {\tau.(\sigma^3.\sndv_{1} \oplus_{\left( \frac{1}{2} + \frac{p}{2}\right)}  \sigma^4.\sndv_1)}   .
\end{array}
\]
By transitivity (Proposition~\ref{prop:transitivity})
we obtain 
\begin{equation}
\label{eq_GSP4-15}
 \mathrm{GSP}_6(\mbox{\small $s_1[v]$; $s_2$ sends in $1$ or $2$ rnds}) \: \simtoll{\frac{p}{2}}  \: 
\prod_{i=1}^{2}\node {\nu_{s_i}}  {s_i} {\nil}  | 
 \node {\nu_d} d {\tau.(\sigma^3.\sndv_{1} \oplus_{q}  \sigma^4.\sndv_1)}  
\end{equation}
with $q = \left( \frac{1}{2} + \frac{p}{2}\right)$.
Finally, in order to estimate the behaviour of the whole network $\mathrm{GSP}_6$ we can apply in sequence:
\begin{inparaenum}[(i)]
\item Theorem~\ref{thm:propagation2bis}
in which the network $M$ is 
structural congruent to $\node {\nu_{s_1}}  {s_1} {\sndu v_{p,2}}   | 
\node {\nu_d} d {\finer}  $ and the node $m$ is instantiated as $s_2$, 
to compose the paths deriving from Equation~\ref{eq_GSP4-10} and  Equation~\ref{eq_GSP4-15};
\item 
 Proposition~\ref{prop:first-laws}(\ref{random3}) together with the additivity and commutativity properties of  $\oplus$;
\end{inparaenum}
 to obtain:     
\[
\begin{array}{rl}
& \mathrm{GSP}_6 \\[2pt]
\deff &\prod_{i=1}^{2}\node {\nu_{s_i}}  {s_i} {\sndu v_{p,2}} \q \big| \q
\node {\nu_d} d {\finer}\\[3pt]
 \simtoll{ r} & 
  \prod_{i=1}^{2}\node {\nu_{s_i}}  {s_i} {\nil} \q \big| \q
 \node {\nu_d} d {\tau.\big(
\tau.(\sigma^3.\sndv_{1}  \oplus_{q} \sigma^4.\sndv_{1}) 
\oplus_{p} 
\tau.(\sigma^3.\sndv_{1} \oplus_{\frac{1}{2} }  \sigma^4.\sndv_1)\big)}   \\[3pt] 
\simtoll{0} & \prod_{i=1}^{2}\node {\nu_{s_i}}  {s_i} {\nil} \q \big| \q
 \node {\nu_d} d {\tau.(\sigma^3.\sndv_{1} \, 
\oplus_{q'} 
 \,  \sigma^4.\sndv_1)} \\[2pt]
\deff & 
\mathrm{DONE}_6   . 
\end{array}
\]
with   
$q' = p \cdot q+ (1-p) \cdot \frac{1}{2}  = \frac{1}{2} {+} \frac{p^2}{2}$ and
$r=p \cdot \frac{p}{2} +  (1-p) (1-p)=1 -(2p - \frac{3}{2} p^2 )$.
%\remarkR{l'oplus ha 1-p  quindi secondo me e' $(1-p) \cdot \frac{1}{2} +  p  \cdot r'$}
By transitivity (Proposition~\ref{prop:transitivity})
we can finally derive
\[
\mathrm{GSP}_6 \; \simtoll{r} \; \,   \mathrm{DONE}_6  .
\]
Hence,  the network $\mathrm{GSP_6}$ succeeds in  
delivering the message to the destination node with (at least) probability  
$2p -  \frac{3}{2} p^2$.

Summarising, in the previous sections we have seen that in the absence of collisions 
the success probability of $\mathrm{GSP}_1$ is (at least) $2p-p^2$. 
Whereas when collisions are taken into account, the network $\mathrm{GSP}_4$
succeeds in propagating the message with (at least) probability $2p-2p^2$,
where $2p -2p^2 < 2p - p^2$, for any $p>0$. 
Since $2p-2p^2 < 2p - \frac{3}{2}p^2 < 2p - p^2$, for any $p>0$, 
this small example shows that the introduction of random delays 
may mitigate the effect of collisions.

%%%%%%%%%%%%%%%%%%%%%%%%%%%%%%%%%%%%%%%%%
%%%%%%%%%%%%%%%%%%%%%%%%%%%%%%%%%%%%%%%%%
%%%%                                                                                                        %%%%%
%%%%                              C O N C L U S I O N I                                       %%%%%
%%%%                                                                                                        %%%%%       
%%%%%%%%%%%%%%%%%%%%%%%%%%%%%%%%%%%%%%%%%
%%%%%%%%%%%%%%%%%%%%%%%%%%%%%%%%%%%%%%%%%

\section{Conclusions, related and future work}
\label{sec:ConRelWor}

We have  proposed a compositional analysis technique to formally study probabilistic gossip protocols 
expressed in a simple probabilistic timed process calculus for wireless
 networks. 
The calculus is equipped with a notion of \emph{weak simulation quasimetric} 
which is used to define a \emph{weak simulation with tolerance}, \emph{i.e.}, 
a \emph{compositional} simulation theory   to express that a probabilistic 
 system may be  simulated by another one with a given tolerance  measuring  the distance between the two systems. Basically, weak simulation quasimetric
is the asymmetric counterpart of  \emph{weak bisimulation metric}~\cite{DJGP02}, and the quantitative analogous of  \emph{weak simulation preorder}~\cite{BKHH02,BHK04}.
Based on our simulation quasimetric,  we have developed an \emph{algebraic theory} to estimate the performance of \emph{gossip wireless networks} in terms of the probability to successfully propagate messages up to the desired destination. 
In our study we have considered gossip networks, with and 
without \emph{communication collisions\/}, and \emph{randomised gossip networks\/} 
adopting a uniform probability distribution to decide whether broadcasting a message.

A nice survey of formal verification techniques for the analysis of gossip
 protocols appears in Bakhashi et al.'s paper~\cite{Fokkink2007}. 
Probabilistic model-checking has been used by Fehnker and Gao~\cite{ansgar2006} to study the influence of different modelling choices on message propagation in flooding and gossip protocols, and
by Kwiatkowska et al.~\cite{Marta2008} to investigate the expected rounds of gossiping required to form a connected network and how the expected path length between nodes evolves over the execution of the protocol.
However, the analysis of gossip protocols in large-scale networks remains beyond the capabilities of current probabilistic model-checking tools.  
For this reason, Bakhashi et al.~\cite{WanFokkink2009b} have suggested to apply mean-field analysis for a formal evaluation of gossip protocols.
Intuitively, the stochastic process representing the modelled system converges to a deterministic process if the number of nodes goes to infinity, providing an approximation for large numbers of nodes. 
Finally, Bakhashi et al.~\cite{WanFokkink2009} have developed and validated an analytical model, based on epidemic techniques, for a shuffle protocol, a protocol to disseminate data items to a collection of wireless devices, in a decentralised fashion.

A preliminary version of the current paper has appeared in~\cite{LMT17}. 
However, in that paper neither  communication collisions nor randomised gossip protocols are taken into account. %%Moreover, no proofs are provided. 
The current paper is the ideal continuation of Lanotte and Merro's work~\cite{LaMe11}. 
In that paper,  the authors developed a notion of \emph{simulation up to probability} to measure the closeness rather than the distance between two
probabilistic systems. Then, as in here, simulation up to probability  has been  used to provide 
an algebraic theory to evaluate the performance of gossip networks. 
Despite the similarity of the two simulation theories, 
the simulation up to probability has a number 
of limitations that have motivated the current work: 
\begin{inparaenum}[(i)]
\item 
  the simulation up to probability is not transitive, while 
simulation quasimetrics are transitive by definition;
 \item 
  in order  to work with a transitive relation, 
 paper~\cite{LaMe11} introduces  
an auxiliary rooted simulation which is much stronger than the main
definition;
 \item
  that rooted simulation (and hence the simulation up to probability)
 is not suitable to compose
estimates originating from paths with different lengths (as we
do here by means of Theorem~\ref{thm:propagation2bis}), and, more
generally, to deal  with more transmissions; 
\item
  paper~\cite{LaMe11}
does not consider randomised gossip protocols. 
\end{inparaenum}

Our simulation quasimetric has been inspired by~\cite{DGJP04,DJGP02,BW05,DCPP06}, where the notion of behavioural distance between two probabilistic systems  is formalised in terms of the notion of bisimulation metric.
Bisimulation metric works fine for systems being approximately equivalent.
Recent applications of weak bisimulation metrics can be found in \cite{LMT18b,iFM2018}.
However, when the simulation game works only in one direction, as in the gossip protocols analysed in the current paper, an asymmetric notion of simulation pseudometric is required.

The compositionality criteria  proposed in the literature for bisimulation metrics require that operators satisfy different forms of \emph{uniform continuity} property~\cite{GT13,GT14,GLT15,GT15,LMT17b}.
We proved that simulation quasimetric matches one of the most restrictive, namely \emph{non-expansiveness}~\cite{DJGP02,DGJP04}  (also known as 1-non-extensiveness~\cite{BBLM13b}).

Finally, several process calculi for wireless systems have
been proposed in the last years \cite{MeSa06,NaHa06,Godskesen07,SRS06,GhasWanFok09, FGHMPT12,BHJRVPP15}.  
Among these, Merro et al.~\cite{MBS11,CHM15} have proposed two different 
calculi modelling time-consuming communication to formally represent and study communication collisions
in a wireless setting. 
Song and Godskesen~\cite{SoGo10} have proposed the first probabilistic untimed calculus
for wireless systems, where connections are established with a given 
probability.  
\cname{} is a probabilistic variant of the calculus proposed by Macedonio and Merro~\cite{MaMe14} that takes inspiration from Deng et al.'s probabilistic CSP~\cite{Dengetal2008}. 

As future work, we intend to study gossip protocols in the presence 
of lossy channels. 
We then plan to apply our metric-based simulation theory in the context of other probabilistic protocols, such as those for \emph{probabilistic 
anonymity}~\cite{Cha88}  and \emph{probabilistic non-repudiation}~\cite{MR99}.

\section*{Acknowledgments}
The authors would like to thank the anonymous reviewers for their insightful reviews.

%%%%%%%%%%%%%%%%%%%%%%%%%%%%%%%%%%%%%%%%%%%%%%%%%%%%%%%%%%%%%%%%%%%%%%%

\bibliography{main}
\bibliographystyle{alpha}

%%%%%%%%%%%%%%%%%%%%%%%%%%%%%%%%%%%%%%%%%%%%%%%%%%%%%%%%%%%%%%

\appendix

\section{Proofs}

\subsection{Proofs of results in Section~\ref{sec:simulation}}
\
\vspace{0.2 cm}

\noindent
\emph{Proof of Proposition~\ref{prop_simulazione}}.
Let ${\mathcal R}$ 
denote the relation ${\mathcal R} = \{ (M, N) \, : \, d(M,N) =0 \}$.
We have to prove that, given any pair of related networks
$M \, {\mathcal R} \, N$,
whenever
$M \transSim{\alpha} \Delta$ there are a transition $N \TransSim{\hat{\alpha}} \Theta$ and a matching $\omega \in \Omega(\Delta, \Theta)$ with $\omega(M',N')>0$ only if
$M' \, {\mathcal R} \, N'$.
By definition of ${\mathcal R}$, relation $M \, {\mathcal R} \, N$ implies $d(M,N)=0$.
Being $d$ a weak simulation quasimetric, $d(M,N)=0$ implies that whenever
$M \transSim{\alpha} \Delta$ there are a transition $N \TransSim{\hat{\alpha}} \Theta$
with $\Kantorovich(d)(\Delta, \Theta + (1-\size{\Theta}) \overline{\dummyN}) \le d(M,N) = 0$.
From $\Kantorovich(d)(\Delta, \Theta + (1-\size{\Theta}) \overline{\dummyN}) \le 0$ we infer that $\size{\Theta} = 1$, thus giving $\Kantorovich(d)(\Delta, \Theta) = 0$.
Let $\omega \in \Omega(\Delta, \Theta)$ be one of the optimal matchings realising $\Kantorovich(d)(\Delta, \Theta)$. 
From $\Kantorovich(d)(\Delta, \Theta) = 0$ we infer that $\omega(M',N') > 0$ only if $d(M',N') = 0$.
Namely, $\omega(M',N') > 0$ only if $M' \, {\mathcal R} \, N'$, which complete the proof.
\qed

%%%
%%%
%%%

To prove Theorem~\ref{thm:exists_metric} we need two preliminary results.
First we show that the pseudoquasimetric property is preserved by function $\Kantorovich$, namely $\Kantorovich(d)$ is a pseudoquasimetric whenever $d$ is a pseudoquasimetric.

\begin{prop}
\label{prop_kant_quasimetric}
If $d \colon \cnamed{} \times \cnamed{} \to [0,1]$ is a $1$-bounded pseudoquasimetric, then $\Kantorovich(d)$ is a $1$-bounded pseudoquasimetric.
\end{prop}
\begin{proof} 
To show $\Kantorovich(d)(\Delta,\Delta) = 0$ it is enough to take the matching $\omega \in \Omega(\Delta,\Delta)$ defined by $\omega(M,M) = \Delta(M)$, for all $M \in \cnamed{}$, and $\omega(M,N) = 0$, for all $M,N \in \cnamed{}$ with $M \neq N$. 
In fact, we obtain $\Kantorovich(d)(\Delta,\Delta) =0 $ by $\Kantorovich(d)(\Delta,\Delta) \le
\sum_{M,N \in \cnamed}\omega(M,N) \cdot d(M,N) = \sum_{M \in \cnamed} \Delta(M) \cdot d(M,M) =0$, %% with the last equality from the property 
because $d(M,M) =0$. %% of the pseudoquasimetric $d$.

To show the triangular property $\Kantorovich(d)(\Delta_1,\Delta_2) \le \Kantorovich(d)(\Delta_1,\Delta_3) + \Kantorovich(d)(\Delta_3,\Delta_2)$, we take
the function $\omega \colon \cnamed{} \times \cnamed{} \to [0,1]$ defined for all networks $M_1,M_2 \in \cnamed{}$ as $\omega(M_1,M_2) = \sum_{M_3 \in \cnamed{} \mid \Delta_3(M_3) \neq 0} \frac{\omega_1(M_1,M_3) \cdot \omega_2(M_3,M_2)}{\Delta_3(M_3)}$, with $\omega_1 \in \Omega(\Delta_1,\Delta_3)$ one of the optimal matchings realising $\Kantorovich(d)(\Delta_1,\Delta_3)$, and $\omega_2 \in \Omega(\Delta_3,\Delta_2)$ one of the optimal matchings realising $\Kantorovich(d)(\Delta_3,\Delta_2)$.
Then, we prove that:
\begin{inparaenum}[(i)]
\item \label{Kant_triang_uno}
$\omega$ is a matching in $\Omega(\Delta_1,\Delta_2)$, and
\item \label{Kant_triang_due}
$\sum_{M_1,M_2 \in \cnamed{}} \omega(M_1,M_2) \cdot d(M_1,M_2) \le \Kantorovich(d)(\Delta_1,\Delta_3)  + \Kantorovich(d)(\Delta_3,\Delta_2)$, which immediately implies $\Kantorovich(d)(\Delta_1,\Delta_2) \le \Kantorovich(d)(\Delta_1,\Delta_3)  + \Kantorovich(d)(\Delta_3,\Delta_2)$.
\end{inparaenum}
To show (\ref{Kant_triang_uno}) we prove that the left marginal of $\omega$ is $\Delta_1$ by\\[3pt]
\begin{math}
\begin{array}{rl}
& 
\sum_{M_2 \in \cnamed{}} \omega(M_1,M_2)
\\[3pt]
= \quad & 
\sum_{M_2 \in \cnamed{}}  \sum_{M_3 \in \cnamed{} \mid \Delta_3(M_3) \neq 0} \frac{\omega_1(M_1,M_3) \cdot \omega_2(M_3,M_2)}{\Delta_3(M_3)}
\\[3pt]
= \quad & 
\sum_{M_3 \in \cnamed{} \mid \Delta_3(M_3) \neq 0} \frac{\omega_1(M_1,M_3) \cdot \Delta_3(M_3)}{\Delta_3(M_3)}
%\text{(by $\omega_2 \in \Omega(\Delta_3,\Delta_2)$)}
\\[3pt]
= \quad & 
\sum_{M_3 \in \cnamed{} \mid \Delta_3(M_3) \neq 0} \omega_1(M_1,M_3) 
\\[3pt]
= \quad & 
\Delta_1(M_1) 
% \text{(by $\omega_1 \in \Omega(\Delta_1,\Delta_3)$)}
\end{array}
\end{math}\\

with the second step by $\omega_2 \in \Omega(\Delta_3,\Delta_2)$ and the last step by $\omega_1 \in \Omega(\Delta_1,\Delta_3)$.
The proof that the right marginal of $\omega$ is $\Delta_2$ is analogous.
Then, we show (\ref{Kant_triang_due}) by\\[3pt]
\begin{math}
\begin{array}{rlr}
&  \sum_{M_1,M_2 \in \cnamed{}} \omega(M_1,M_2) \cdot d(M_1,M_2)
\\[3pt]
=  \quad &\sum_{M_1,M_2 \in \cnamed{}} \sum_{M_3 \in \cnamed{} \mid \Delta_3(M_3) \neq 0} \frac{\omega_1(M_1,M_3) \cdot \omega_2(M_3,M_2)}{\Delta_3(M_3)} \cdot d(M_1,M_2)
\\[3pt]
\le   \quad &  \sum_{M_1,M_2 \in \cnamed{},M_3 \in \cnamed{} \mid \Delta_3(M_3) \neq 0} \frac{\omega_1(M_1,M_3) \cdot \omega_2(M_3,M_2)}{\Delta_3(M_3)} \cdot d(M_1,M_3)  \; + 
\\[3pt]
 \quad  & 
 \sum_{M_1,M_2 \in \cnamed{}, M_3 \in \cnamed{} \mid \Delta_3(M_3) \neq 0} \frac{\omega_1(M_1,M_3) \cdot \omega_2(M_3,M_2)}{\Delta_3(M_3)} \cdot d(M_3,M_2) 
\\[4 pt]
=  \quad &  \sum_{M_1,M_3 \in \cnamed{}} \frac{\omega_1(M_1,M_3) \cdot \Delta_3(M_3)}{\Delta_3(M_3)} \cdot  d(M_1,M_3)  \; + 
\\[3 pt]
\quad & 
   \sum_{M_2,M_3 \in \cnamed{}} \frac{\Delta_3(M_3) \cdot \omega_2(M_3,M_2)}{\Delta_3(M_3)}  \cdot d(M_3,M_2)
\\[3pt]
=  \quad &  \sum_{M_1,M_3 \in \cnamed{}} \omega_1(M_1,M_3) \cdot d(M_1,M_3) \; +
 \sum_{M_2,M_3 \in \cnamed{}} \omega_2(M_3,M_2)  \cdot d(M_3,M_2)
\\[3pt]
=  \quad &  \Kantorovich(d)(\Delta_1,\Delta_3)  + \Kantorovich(d)(\Delta_3,\Delta_2) 
\end{array}
\end{math}\\[2pt]
where the inequality follows from the triangular property of $d$ and the third last equality follows by $\omega_2 \in \Omega(\Delta_3,\Delta_2)$ and $\omega_1 \in \Omega(\Delta_1,\Delta_2)$.
\end{proof}

%%%
%%%
%%%
%%%
%%%

Now we prove that if $d(M,N) < 1$ for any weak simulation quasimetric $d$, then $N$ is able to simulate transitions of the form $M \TransSim{\hat \alpha} \Delta$, besides those of the form $M \transSim{\alpha} \Delta$.
\begin{lem} 
\label{lemma_sim_weak_transitions}
Assume any weak simulation quasimetric $d$ and two networks $M,N \in \cnamed{}$ with $d(M,N) <1$.
If $M \TransSim{\hat \alpha} \Delta$ then there is a weak transition $N \TransSim{\hat \alpha} \Theta$ such that $\Kantorovich(d)(\Delta + (1-\size{\Delta}) \overline{\dummyN},\Theta + (1-\size{\Theta}) \overline{\dummyN}) \le d(M,N)$.
\end{lem}
\begin{proof}
We proceed by induction on the length $n$ of $M \TransSim{\hat \alpha} \Delta$.

\emph{\underline{Base case $n=1$}}. We have two cases: 
The first is $\alpha = \tau$ and $\Delta = \overline M$, the second is $M \transSim{\alpha} \Delta$.  
In the first case, by definition of $\Trans {\widehat \tau}$ we have $N \TransSim{\widehat \tau} \overline N$ and
the thesis holds for the distribution $\Theta = \overline{N}$ by  $\Kantorovich(d)(\overline{M} + (1-\size{\overline M})\overline{\dummyN}),\overline{N} + (1-\size{\overline N})\overline{\dummyN}) = \Kantorovich(d)(\overline{M},\overline{N}) = d(M,N)$.
In the second case, the thesis follows directly by the definition of weak simulation quasimetric.

\emph{\underline{Inductive step $n>1$}}.
The derivation $M \TransSim{\hat \alpha} \Delta$ is obtained by $M \TransSim{\hat \beta_1} \Delta'$ and $\Delta' \transSim{\hat{\beta}_2} \Delta$, for some distribution $\Delta' \in {\mathcal D}(\cnamed{})$.
The length of the derivation $M \TransSim{\hat \beta_1} \Delta'$ is $n-1$ and hence, by the inductive hypothesis, there is a weak transition $N \TransSim{\hat \beta_1} \Theta'$ such that $\Kantorovich(d)(\Delta' + (1-\size{\Delta'}) \overline{\dummyN},\Theta' + (1-\size{\Theta'}) \overline{\dummyN}) \le d(M,N)$.
The sub-distributions $\Delta'$ and $\Theta'$ are of the form $\Delta' = \sum_{i \in I}p_i M_i$ and $\Theta' = \sum_{j \in J}q_j N_j$, for suitable networks $M_i$ and $N_j$.
We have two sub-cases: The first is $\beta_1=\tau$ and $\beta_2=\alpha$, the other $\beta_1=\alpha$ and $\beta_2=\tau$.

We consider the case $\beta_1=\tau$ and $\beta_2=\alpha$, the other is analogous.
In this case we have $\size{\Delta'} = \size{\Theta'} = 1$ and $\Kantorovich(d)(\Delta' ,\Theta') \le d(M,N)$.
The transition $\Delta' \transSim{\hat{\beta}_2} \Delta$ is derived from a $\beta_2$-transition by some of the networks $M_i$, namely $I$ is partitioned into sets $I_1$ and $I_2$ such that:
\begin{inparaenum}[(i)]
\item  for all $i \in I_1$ we have $M_i \transSim{\beta_2} \Delta_i$ for suitable distributions $\Delta_i$, \item for each $i \in I_2$ we have $M_i \not\!\!\transSim{\beta_2}$,
\item $\Delta = \sum_{i \in I_1} p_i \Delta_i$,
\item $1-|\Delta| = \sum_{i \in I_2} p_i$.
\end{inparaenum}
Analogously, $J$ is partitioned into sets $J_1$ and $J_2$ such that for all $j \in J_1$ we have $N_j \TransSim{\hat{\beta}_2} \Theta_j$ for suitable distributions $\Theta_j$ and for each $j \in J_2$ we have $N_j\not\!\!\TransSim{\hat{\beta}_2}$. This gives $\Theta' \TransSim{\hat{\beta}_2} \Theta$ with $\Theta = \sum_{j \in J_1} q_j \Theta_j$.
Since we had $N \TransSim{\hat{\beta}_1} \Theta'$, we can conclude $N \TransSim{\hat{\alpha}} \Theta$.
In the following we will prove that the transitions $N_j \TransSim{\hat{\beta}_2} \Theta_j$ can be chosen so that $\Kantorovich(d)(\Delta + (1-\size{\Delta}) \overline{\dummyN},\Theta + (1-\size{\Theta}) \overline{\dummyN}) \le d(M,N)$, which will conclude the proof.

Let $\omega$ be one of the optimal matchings realising $\Kantorovich(d)(\Delta' ,\Theta')$.
Since $p_i = \sum_{j \in J} \omega(M_i,N_j)$ for all $i \in I$ and, then, $q_j = \sum_{i \in I} \omega(M_i,N_j)$ for all $j \in J$, we can rewrite the distributions $\Delta' = \sum_{i \in I}p_i  M_i$ and $\Theta' = \sum_{j \in J}q_j N_j$ as 
$\Delta' = \sum_{i \in I, j \in J} \omega(M_i,N_j)  M_i$ and 
$\Theta' = \sum_{i \in I, j \in J} \omega(M_i,N_j)  N_j$.
For all $i \in I_1$ and $j \in J$, define $\Delta_{i,j} = \Delta_i$.
Since $\Delta = \sum_{i \in I_1} p_i \Delta_i$, we can rewrite $\Delta$ as $\Delta = \sum_{i \in I_1,j\in J} \omega(M_i,N_j) \Delta_{i,j}$.
Analogously, for each $j \in J_1$ and $i \in I$ we note that the transition $q_j \overline{N_j} \TransSim{\hat{\beta}_2} \Theta_{j}$ can always be split into $\sum_{i \in I} \omega(M_i,N_j)\overline{N_j} \TransSim{\hat{\beta}_2} \sum_{i \in I} \omega(M_i,N_j)\Theta_{i,j}$ for suitable distributions $\Theta_{i,j}$
so that for all $j \in J_1$ 
we can rewrite $\Theta_j$ as $\Theta_j = \sum_{i \in I}\omega(M_i,N_j) \Theta_{i,j}$, and we can rewrite $\Theta$ as
$\Theta = \sum_{i\in I, j \in J_1} \omega(M_i,N_j) \Theta_{i,j}$. 
Then we note that for all $i \in I_1$ and $j \in J_1$ with $d(M_i,N_j) < 1$, the transition $N_j \TransSim{\hat{\beta}_2} \Theta_{i,j}$ can be chosen so that  
$\Kantorovich(d)(\Delta_{i,j},\Theta_{i,j} + (1-\size{\Theta_{i,j}})\overline{\dummyN}) \le d(M_i,N_j)$.
For all $i \in I_1$ and $j \in J_1$ with $d(M_i,N_j) <1$, let $\omega_{i,j}$ be one of the optimal matchings realising $\Kantorovich(d)(\Delta_{i,j}, \Theta_{i,j} + (1-\size{\Theta_{i,j}}) \overline{\dummyN})$.

Define $\omega' \colon \cnamed{} \times \cnamed{} \to [0,1]$ as the function such that
$\omega'(M',N')$ is the summation of the following values:
\begin{enumerate}
\item
$\sum_{i \in I_1, j \in J_1} \omega(M_i,N_j) \cdot \omega_{i,j}(M',N') $
\item
$\sum_{i \in I_1, j \in J_2} \omega(M_i,N_j) \cdot \Delta_{i,j}(M') \cdot\overline{\dummyN}(N')$
\item
$\sum_{i \in I_2, j \in J_1} \omega(M_i,N_j) \cdot 
\left( \overline{\dummyN}(M') \cdot\Theta_{i,j}(N') + (1-\size{\Theta_{i,j}}) \cdot \overline{\dummyN}(M') \cdot \overline{\dummyN}(N')\right) $
\item
$\sum_{i \in I_2, j \in J_2} \omega(M_i,N_j) \cdot \overline{\dummyN}(M')\cdot \overline{\dummyN}(N')$.
\end{enumerate}
To infer the proof obligation $\Kantorovich(d)(\Delta + (1-\size{\Delta}) \overline{\dummyN},\Theta + (1-\size{\Theta}) \overline{\dummyN}) \le d(M,N)$ it is now enough to show that:
\begin{compactenum}[(a)]
\item \label{matching} the function $\omega'$ is a matching in $\Omega(\Delta + (1-\size{\Delta}) \overline{\dummyN},\Theta + (1-\size{\Theta}) \overline{\dummyN})$, and
\item \label{metric_condition} $\sum_{M',N' \in \cnamed{}} \omega'(M',N') \cdot d(M',N') \le d(M,N)$.
\end{compactenum}

To show (\ref{matching}) we prove that the left marginal of $\omega'$ is $\Delta + (1-\size{\Delta}) \overline{\dummyN}$.
The proof that the right marginal is $\Theta + (1-\size{\Theta}) \overline{\dummyN})$ is analogous.
For any network $M' \in \cname{}$ we have that 
$\sum_{N' \in \cname}\omega'(M',N')$ is the summation of the following values:
\begin{enumerate}
\item
$\sum_{N' \in \cname}\sum_{i \in I_1, j \in J_1} \omega(M_i,N_j) \cdot \omega_{i,j}(M',N') 
=
 \sum_{i \in I_1, j \in J_1} \omega(M_i,N_j) \cdot \Delta_{i,j}(M')$
\item
$\sum_{N' \in \cname}\sum_{i \in I_1, j \in J_2} \omega(M_i,N_j) \cdot \Delta_{i,j}(M') \cdot\overline{\dummyN}(N')
=
\sum_{i \in I_1, j \in J_2} \omega(M_i,N_j) \cdot \Delta_{i,j}(M')$
\item
$\sum_{N' \in \cname}\sum_{i \in I_2, j \in J_1} \omega(M_i,N_j)
\left( \overline{\dummyN}(M') \Theta_{i,j}(N') + (1-\size{\Theta_{i,j}}) \overline{\dummyN}(M')  \overline{\dummyN}(N')\right)
\\ =  \sum_{i \in I_2, j \in J_1} \omega(M_i,N_j) \cdot  \overline{\dummyN}(M') $
\item
$\sum_{N' \in \cname}\sum_{i \in I_2, j \in J_2} \omega(M_i,N_j)  \overline{\dummyN}(M') \overline{\dummyN}(N') =
\sum_{i \in I_2, j \in J_2} \omega(M_i,N_j)  \overline{\dummyN}(M')
$
\end{enumerate}
with the first equality by $\omega_{i,j} \in \Omega(\Delta_{i,j}, \Theta_{i,j} + (1-\size{\Theta_{i,j}})$, the second by $\sum_{N' \in \cname{}}\overline{\dummyN}(N') = 1$,
the third by $\sum_{N' \in \cname{}} \Theta_{i,j}(N') = \size{\Theta_{i,j}}$ and $\sum_{N' \in \cname{}}\overline{\dummyN}(N') = 1$,
and the last by $\sum_{N' \in \cname{}}\overline{\dummyN}(N') = 1$.

The summation of the first two items gives
$\sum_{i \in I_1, j \in J} \omega(M_i,N_j) \cdot \Delta_{i,j}(M')$, which is $\Delta(M')$ by the definition of $\Delta$, and the summation of the last two items gives
$\sum_{i \in I_2, j \in J} \omega(M_i,N_j) \cdot  \overline{\dummyN}(M')$, which is
$(1-\size{\Delta})\overline{\dummyN}(M')$.
Overall, the summation of all items gives $\Delta(M') + (1-\size{\Delta})\overline{\dummyN}(M')$, namely
$(\Delta + (1-\size{\Delta}) \overline{\dummyN})(M')$.

To prove (\ref{metric_condition}), by the definition of $\omega'$ we have that $\sum_{M',N' \in \cnamed} \omega'(M',N') \cdot d(M',N')$ is the summation of the following values:
\begin{enumerate}
\item
$\sum_{M',N' \in \cnamed} \sum_{i \in I_1, j \in J_1} \omega(M_i,N_j) \cdot \omega_{i,j}(M',N') \cdot d(M',N')$
\item
$\sum_{M',N' \in \cnamed}  \sum_{i \in I_1, j \in J_2} \omega(M_i,N_j) \cdot \Delta_{i,j}(M') \cdot\overline{\dummyN}(N')\cdot d(M',N')
\\= 
\sum_{M'\in \cnamed}  \sum_{i \in I_1, j \in J_2} \omega(M_i,N_j) \cdot \Delta_{i,j}(M') \cdot d(M',\dummyN)
\\
\le
\sum_{M'\in \cnamed}  \sum_{i \in I_1, j \in J_2} \omega(M_i,N_j) \cdot \Delta_{i,j}(M')
$
\item
$\sum_{M',N' \in \cnamed} \sum_{i \in I_2, j \in J_1} \omega(M_i,N_j) \cdot
( \overline{\dummyN}(M') \cdot \Theta_{i,j}(N') 
+ (1-\size{\Theta_{i,j}}) \cdot \overline{\dummyN}(M')\cdot  \overline{\dummyN}(N') )  \cdot d(M',N') =0$
\item
$\sum_{M',N' \in \cnamed} \sum_{i \in I_2, j \in J_2} \omega(M_i,N_j) \cdot \overline{\dummyN}(M')\cdot \overline{\dummyN}(N')\cdot d(M',N') = 0$.
\end{enumerate}
where in the second item the equality follows from $\overline{\dummyN}(\dummyN)=1$ and $\overline{\dummyN}(N')=0$ for all $N' \neq \dummyN$ and the inequality from
$d(M',\dummyN) \le 1$, and in the third and fourth item the equality follows by the fact the only network with $\overline{\dummyN}(M') >0$ is $\dummyN$, for which we have $d(\dummyN',N') = 0$ for all $N' \in \cname{}$.

By the definition of $\omega_{i,j}$ the first item is $\sum_{i \in I_1, j \in J_1} \omega(M_i,N_j) \cdot \Kantorovich(d)(\Delta_{i,j},\Theta_{i,j})$.
If $d(M_i,N_j) < 1$, we argued above 
that the distribution $\Theta_{i,j}$ satisfies $\Kantorovich(d)(\Delta_{i,j},\Theta_{i,j}) \le d(M_i,N_j)$.
If $d(M_i,N_j) = 1$, then $\Kantorovich(d)(\Delta_{i,j},\Theta_{i,j}) \le d(M_i,N_j)$ is immediate.
Henceforth we are sure that in all cases the first item is less or equal $\sum_{i \in I_1, j \in J_1} \omega(M_i,N_j) \cdot d(M_i,N_j)$.
The second item is clearly less or equal than $\sum_{i \in I_1, j \in J_2}\omega(M_i,N_j)$.
Summarising, we have
\[
\sum_{M',N' \in \cnamed} \omega'(M',N') \cdot d(M',N') \le \sum_{i \in I_1, j \in J_1} \omega(M_i,N_j) \cdot d(M_i,N_j) + \sum_{i \in I_1, j \in J_2}\omega(M_i,N_j) .
\]
Then, since $\Kantorovich(d)(\Delta' ,\Theta' )$ is the summation of the following values:
\begin{itemize}
\item
$\sum_{i \in I_1,j\in J_1} \omega(M_i,N_j) \cdot d(M_i,N_j)$ 
\item
$\sum_{i \in I_1,j\in J_2} \omega(M_i,N_j) \cdot d(M_i,N_j) = \sum_{i \in I_1,j\in J_2} \omega(M_i,N_j)$ (since $M_i \transSim{\beta_2}$ and $N_j \nTransSimone[\hat{\beta_2}]$ give $d(M_i,N_j) =1$) 
\item
$\sum_{i \in I_2,j\in J_1} \omega(M_i,N_j) \cdot d(M_i,N_j)$ 
\item
$\sum_{i \in I_2,j\in J_2} \omega(M_i,N_j) \cdot d(M_i,N_j)$.
\end{itemize}
it is immediate that $\sum_{i \in I_1, j \in J_1} \omega(M_i,N_j) \cdot d(M_i,N_j) + \sum_{i \in I_1, j \in J_2}\omega(M_i,N_j) \le \Kantorovich(d)(\Delta' ,\Theta' )$.
Since we had $\Kantorovich(d)(\Delta' ,\Theta') \le d(M,N)$ we can conclude that $\sum_{M',N' \in \cnamed} \omega'(M',N') \cdot d(M',N') \le d(M,N)$, as required.
\end{proof}

We are now ready to prove Theorem~\ref{thm:exists_metric}.

\noindent
\emph{Proof of Theorem~\ref{thm:exists_metric}.}
Let $D$ be the set of all weak simulation quasimetrics. 
Define the function $\metric \colon \cnamed{} \times \cnamed{} \to [0,1]$ by 
\[
\metric(M,N) = \inf\left\{\sum_{i=1}^{n} \inf_{d \in D}d(O_i,O_{i+1})  \mid n \in \bbbn^{+} , O_1 =M , O_{n+1}=N, O_i \in \cnamed{}\right\}
\]
for all $M,N \in \cnamed{}$. 
It is immediate that for all $d \in D$ we have $\metric(M,N) \le d(M,N)$ for all $M,N \in \cnamed{}$.
Therefore, to derive the thesis it is enough to show to $\metric$ is a weak simulation quasimetric.
We start with showing that $\metric$ is a pseudoquasimetric.
The property $\metric(M,M)=0$ follows directly by $\metric(M,M) \le d(M,M)$ and $d(M,M) =0$ for all $d \in D$.
Then we show the triangular inequality $\metric(M,N) \le \metric(M,O) + \metric(O,N)$ by\\[3pt]
\begin{math}
\begin{array}{rl}
& \metric(M,N) \\[3pt]
= & \inf\{\sum_{i=1}^{n} \inf_{d \in D}d(O_i,O_{i+1})  \mid n \in \bbbn^{+} ,  O_1 =M , O_{n+1}=N , O_i \in \cnamed{} \} \\[3pt]
\le & \inf\{\sum_{i=1}^{n} \inf_{d \in D}d(O_i,O_{i+1})  \mid n \in \bbbn^{+} ,  O_1 =M , O_{n+1}=N , O_i \in \cnamed{} \} 
\end{array}
\end{math}\\
\begin{math}
\begin{array}{rl}
= & 
\inf\{\sum_{i=1}^{n} \inf_{d \in D}d(O_i,O_{i+1})  \mid n \in \bbbn^{+} , O_1 =M , O_{n+1}=O , O {\in} \{O_2,..,O_n\},  O_i \in \cnamed{} \} \\[3pt]
&  + \\[3pt]
& \inf\{\sum_{i=1}^{n} \inf_{d \in D}d(O_i,O_{i+1})  \mid n \in \bbbn^{+} , O_1 =O , O {\in} \{O_2,..,O_n\},   O_{n+1}=N , O_i \in \cnamed{} \}  \\[3pt]
= & \metric(M,O) + \metric(O,N). 
\end{array}
\end{math}\\[2pt]
We conclude that the pseudometric property of $\metric$ is given.
It remains to prove that whenever $\metric(M,N) < 1$ and $M \transSim{\alpha} \Delta$ there is a transition $N \TransSim{\hat{\alpha}} \Theta$ with $\Kantorovich(\metric)(\Delta,\Theta + (1-\size{\Theta}) \overline{\dummyN}) \le \metric(M,N)$.
By the finite branching property, it is enough to prove that, fixed $O_1=M$ and $O_{n+1}=N$ and given arbitrary networks $O_2,\ldots,O_n \in \cnamed{}$ and arbitrary weak  simulation quasimetrics $d_1,\ldots , d_n \in D$ such that $\sum_{i=1}^{n}d_i(O_i,O_{i+1}) \le 1$, then $M \transSim{\alpha} \Delta$ implies that there is a transition $N \TransSim{\hat{\alpha}} \Theta$ with $\Kantorovich(\metric)(\Delta,\Theta + (1-\size{\Theta}) \overline{\dummyN}) \le  \sum_{i=1}^{n} d_i(O_i,O_{i+1})$.
Assume such a $M \transSim{\alpha} \Delta$, and let $\Delta_1 = \Delta$.
For $i=1,\ldots,n$ we get $O_i \TransSim{\hat{\alpha}} \Delta_i$ with $\Kantorovich(d_i)(\Delta_i + (1-\size{\Delta_{i}}) \overline{\dummyN} ,\Delta_{i+1} + (1-\size{\Delta_{i+1}})\overline{\dummyN}) \le d_i(O_i,O_{i+1})$, by exploiting  
Lemma~\ref{lemma_sim_weak_transitions} in cases $i=2,\ldots,n$.
This allows us to infer\\[3pt]
\begin{math}
\begin{array}{rlr}
& \Kantorovich(\metric)(\Delta_1 , \Delta_{n+1} + (1-\size{\Delta_{n+1}}) \overline{\dummyN})  \\[3pt]
\le & \sum_{i=1}^{n} \Kantorovich(\metric)(\Delta_i + (1-\size{\Delta_{i}}) \overline{\dummyN}, \Delta_{i+1} + (1-\size{\Delta_{i+1}}) \overline{\dummyN})  \\[3pt]
\le & \sum_{i=1}^{n} \Kantorovich(d_i)(\Delta_i + (1-\size{\Delta_{i}}) \overline{\dummyN}, \Delta_{i+1} + (1-\size{\Delta_{i+1}}) \overline{\dummyN}) \\[3pt]
\le & \sum_{i=1}^{n} d_i(O_i,O_{i+1})
\end{array}
\end{math}\\[2pt]
thus confirming that $\Delta_{n+1}$ is the distribution $\Theta$ we were looking for,
where the first step follows by the triangular inequality for $\Kantorovich(\metric)$, which is a pseudoquasimetric by 
Proposition~\ref{prop_kant_quasimetric} and the fact that we have already proved that $\metric$ is a pseudoquasimetric, and the second step follows by the monotonicity of $\Kantorovich$ and the fact that we have proved that $\metric(M',N') \le d_i(M',N')$ for all $M',N' \in \cnamed{}$.
\qed

%%%
%%%
%%%
%%%
%%%
%%%

\noindent
\emph{Proof of Proposition~\ref{M_go_to_N}}.
We have to prove that any transition $N \transSim{\alpha} \Theta$ is simulated by a suitable weak transition $M \TransSim{\hat{\alpha}} \Theta'$ with  $\Kantorovich(\metric)(\Theta,\Theta' + (1-\size{\Theta'})\overline{\dummyN}) \le q$.
Assume an arbitrary transition $N \transSim{\alpha} \Theta$. 
From $N \transSim{\alpha} \Theta$ and the hypothesis $M \TransSim{ \hat {\tau}} (1-q) \overline{N}+ q \Delta$, we get $M \TransSim{\hat{\alpha}} (1-q)  \Theta + q \Delta'$, for a suitable sub-distribution $\Delta'$ such that $\Delta \TransSim{\hat{\alpha}} \Delta'$.
To conclude we have to prove that $\Kantorovich(\metric)(\Theta,(1-q)  \Theta + q\Delta' + q(1-\size{\Delta'} )\overline{\dummyN}) \le q$, namely $(1-q)  \Theta + q\Delta' $ is the sub-distribution $\Theta'$ we were looking for.
Let $\omega \colon \cnamed{} \times \cnamed{} \to [0,1]$ be the function defined by 
\begin{itemize}
\item
$\omega(M',M') = (1-q) \cdot \Theta(M') + q  \cdot \Theta(M')  \cdot \Delta'(M') + q \cdot (1-\size{\Delta'})  \cdot \Theta(M') \cdot \overline{\dummyN}(M')$, and 
\item 
$\omega(M',N') =  q  \cdot \Theta(M')  \cdot \Delta'(N') + q \cdot (1-\size{\Delta'})  \cdot \Theta(M') \cdot \overline{\dummyN}(N')$, for $M' \neq N'$.
\end{itemize}
We have to prove that:
\begin{enumerate}
\item \label{lem:N_go_to_M_uno} $\omega \in \Omega(\Theta,(1-q)  \Theta + q \Delta' + q(1-\size{\Delta'})\overline{\dummyN}))$, and 
\item \label{lem:N_go_to_M_due} $\sum_{M',N' \in \cnamed} \omega(M',N') \cdot \metric(M',N') \le q$.
\end{enumerate}
To prove (\ref{lem:N_go_to_M_uno}) we show that the left marginal of $\omega$ is $\Theta$ by\\ 
\begin{math}
\begin{array}{rl}
& \sum_{N' \in \cnamed} \omega(M',N') 
\\[2pt]
 = & 
\omega(M',M') +  \sum_{N' \neq M'} \omega(M',N')
\\[2pt]
= &
(1-q) \cdot \Theta(M') + q   \cdot \Theta(M')  \cdot \Delta'(M') + q \cdot (1-\size{\Delta'})\cdot \Theta(M') \cdot  \overline{\dummyN}(M') \: + \\
 &
 \sum_{N' \neq M'} q \cdot \Theta(M') \cdot \Delta'(N')  + 
 q \cdot (1-\size{\Delta'}) \cdot \Theta(M') \cdot  \overline{\dummyN}(N')
\\[2pt]
= &  
 (1-q) \cdot \Theta(M') + \sum_{N' \in \cname{}}  q \cdot \Theta(M') \cdot \Delta'(N') \: + \\
&
\sum_{N' \in \cname{}} q \cdot (1-\size{\Delta'}) \cdot \Theta(M') \cdot  \overline{\dummyN}(N')
\end{array}
\end{math}\\
\begin{math}
\begin{array}{rl} 
= &  
 (1-q) \cdot \Theta(M') +  q \cdot \size{\Delta'} \cdot \Theta(M') + q \cdot (1-\size{\Delta'}) \cdot \Theta(M') 
\\[2 pt]
= &  
 \Theta(M')
\end{array}
\end{math}\\
the proof that the right marginal is $(1-q)  \Theta + q \Delta' + q(1-\size{\Delta'} )\overline{\dummyN}$ being analogous.
We get (\ref{lem:N_go_to_M_due}) by\\
\begin{math}
\begin{array}{rl}
& \sum_{M',N'} \omega(M',N') \cdot \metric(M',N') 
\\[2 pt]
= & 
\sum_{M',M'} \omega(M',M') \cdot \metric(M',M') +  \sum_{M' \neq N'} \omega(M',N') \cdot \metric(M',N')
\\[2 pt]
= & 
\sum_{M' \neq N'} \omega(M',N') \cdot \metric(M',N')
\\[2 pt]
= & 
\sum_{M' \neq N'}
 (q  \cdot \Theta(M')  \cdot \Delta'(N') + q \cdot (1-\size{\Delta'})  \cdot \Theta(M') \cdot \overline{\dummyN}(N'))\cdot \metric(M',N')
\\[2 pt]
\le &
 q \cdot \size{\Delta'} + q\cdot  (1-\size{\Delta'})
\\[2 pt]
 \le & q
\end{array}
\end{math}\\
\enlargethispage{.65\baselineskip}
with the second equality by $\metric(M',M') = 0$ and the first inequality by $\metric(M',N') \le 1$, the fact that $\sum_{M' \in \cname{}} \Theta(M') = 1$ and $\sum_{N' \in \cname{}}\overline{\dummyN}(N') = 1$.
\qed

%%%
%%%
%%%
%%%
%%%
%%%

\noindent
\emph{Proof of  Proposition~\ref{lem:O_go_to_N_sim_M}.} 
Assume a transition $M \transSim{\alpha} \Theta$. 
By $M \simtol{p} N$ there is a transition $N \TransSim{\hat{\alpha}} \Theta'$ with $\Kantorovich(\metric)(\Theta,\Theta' + (1-\size{\Theta'}) \overline{\dummyN}) \le p$.
By the hypothesis $O \TransSim{\hat{\tau}} (1-q) \overline{N} + q \Delta$ we get  $O \TransSim{\hat{\alpha}} (1-q) \Theta' + q \Delta'$, for a suitable distribution $\Delta'$ such that $\Delta \TransSim{\hat{\alpha}} \Delta'$.
To conclude the proof we should show that 
\[ 
\Kantorovich(\metric)(\Theta,(1-q) \Theta' + q \Delta' + ((1-q)(1-\size{\Theta'}) + q (1-\size{\Delta'}))\overline{\dummyN}) \le p(1-q) + q.
\]
Let us consider $\omega \in \Omega(\Theta,\Theta' + (1-\size{\Theta'}) \overline{\dummyN})$ as one of the optimal matching realising $\Kantorovich(\metric)(\Theta,\Theta' + (1-\size{\Theta'}) \overline{\dummyN}) $.
Define the function $\omega' \colon \cnamed{} \times \cnamed{} \to [0,1]$ by 
\[
\omega'(M',N') = (1-q)\cdot  \omega(M',N') + q \cdot \Theta(M') \cdot \Delta'(N') + q \cdot (1-\size{\Delta'}) \cdot \Theta(M') \cdot \overline{\dummyN}(N')
\]
 for all $M',N' \in \cname{}$.
We have to prove that
\begin{enumerate}
\item \label{lem:N_go_to_M_unobis} $\omega' \in \Omega(\Theta,(1-q) \cdot \Theta' + q \Delta' + ((1-q)(1-\size{\Theta'}) + q (1-\size{\Delta'}))\overline{\dummyN})$, and
\item \label{lem:N_go_to_M_duebis} $\sum_{M',N' \in \cnamed} \omega'(M',N') \cdot \metric(M',N') \le p(1-q) + q$.
\end{enumerate}
To prove (\ref{lem:N_go_to_M_unobis}) we show that the left marginal of $\omega'$ is $\Theta$ by
\\
\begin{math}
\begin{array}{rl}
& \sum_{N' \in \cnamed} \omega'(M',N') 
\\[2 pt]
 = & 
(1-q) \cdot \sum_{N' \in \cnamed} \omega(M',N') +q \cdot \sum_{N' \in \cnamed} \Theta(M') \cdot \Delta'(N') + \\
& q \cdot \sum_{N' \in \cnamed{}}  (1-\size{\Delta'}) \cdot \Theta(M') \cdot \overline{\dummyN}(N')
\\[2 pt]
= & 
(1-q)\cdot \Theta(M') +q \cdot \Theta(M') \cdot |\Delta'| +q \cdot \Theta(M') \cdot (1-\size{\Delta'})
\\[2 pt]
= & 
\Theta(M')
\end{array}
\end{math}\\
with the second equality by the fact that $\omega$ is a matching in $\Omega(\Theta,\Theta' + (1-\size{\Theta'}) \overline{\dummyN}))$ and by the equalities 
$\sum_{N' \in \cname{}}\Delta'(N') = \size{\Delta'}$ and
$\sum_{N' \in \cname{}}\overline{\dummyN}(N') = 1$.
The proof that the right marginal is $(1-q) \Theta' + q \Delta'+ ((1-q)(1-\size{\Theta'}) + q (1-\size{\Delta'}))\overline{\dummyN}$ is analogous.
We get (\ref{lem:N_go_to_M_duebis}) by the following chain of inequalities: 
\begin{displaymath}
\begin{array}{rl}
& \sum_{M',N'} \omega'(M',N') \cdot \metric(M',N') 
\\[2pt]
= &  (1{-}q) \sum_{M',N' \in \cnamed} \omega(M',N') \cdot \metric(M',N') + q \sum_{M',N' \in \cnamed} \Theta(M') \cdot \Delta'(N') \cdot \metric(M',N') 
\\
+ & q \cdot \sum_{M',N' \in \cnamed{}}  (1-\size{\Delta'}) \cdot \Theta(M') \cdot \overline{\dummyN}(N') \cdot \metric(M',N') 
\\[2pt]
\le &
 (1-q) \cdot \Kantorovich(\metric)(\Theta,\Theta' + (1-\size{\Theta'})\overline{\dummyN}) + q \size{\Delta'} + q(1-\size{\Delta'})
\\[2pt]
\le & p(1-q) + q
\end{array}
\end{displaymath}
\noindent 
with the first inequality deriving by the fact that 
 $\omega$ is a matching realising
\begin{align*}
  \Kantorovich(\metric)(\Theta,\Theta' + (1-\size{\Theta'}) \overline{\dummyN}),\metric(M',N') \le 1
\end{align*}
and the equalities
\begin{align*}
\sum_{M' \in \cname{}}\Theta(M') = 1, \qquad
\sum_{N' \in \cname{}}\Delta'(N') = \size{\Delta'}, \quad \text{ and }
\sum_{N' \in \cname{}}\overline{\dummyN}(N') = 1.
\tag*{\qed}
\end{align*}
% \qed

\noindent
\emph{Proof of Proposition~\ref{lem:O_go_to_N_sim_M_conSigma}.}
Assume a transition $M \transSim{\alpha} \Theta$. 
By $M \simtol{p} N$ there is a transition $N \TransSim{\hat{\alpha}} \Theta'$ with $\Kantorovich(\metric)(\Theta,\Theta' + (1-\size{\Theta'}) \overline{\dummyN}) \le p$.
By the hypothesis that $N$ can perform only the transition $N \transSim{\sigma} \overline{N'}$ for a network $N'$ such that $N' \ntransSimone[\tau]$, we are sure that $\alpha = \sigma$ and $\Theta' = \overline{N'}$.
By the hypothesis $O \transSim{\sigma} O'  \TransSim{\hat{\tau}} (1-q) \overline{N'} + q \Delta$ we get  $O \TransSim{\hat{\alpha}} (1-q) \Theta' + q \Delta'$, for a suitable distribution $\Delta'$ such that $\Delta \TransSim{\hat{\alpha}} \Delta'$.
To conclude, it remains to prove $\Kantorovich(\metric)(\Theta,(1-q) \Theta' + q \Delta' + ((1-q)(1-\size{\Theta'}) + q (1-|\Delta'|))\overline{\dummyN}) \le p(1-q) + q$, which follows as in the proof of  Lemma~\ref{lem:O_go_to_N_sim_M}.
\qed

%%%
%%%
%%%

To prove  Theorem~\ref{cor:non_expansiveness} we give the following preliminary result.

\begin{lem}
$M \simtol{p} N$  entails $M | O \simtol{p} N | O$ and $O | M \simtol{p} O | N$.
\label{lem:non_expansiveness_support}
\end{lem}
\begin{proof}
We show that $M \simtol{p} N$ entails $M | O \simtol{p} N | O$, the property $O | M \simtol{p} O | N$ can be derived analogously.
It is enough to prove that $\metric(M,N) = p$ entails $\metric(M | O, N | O) \le p$.
Let us define the function $d' \colon \cnamed{} \times \cnamed{} \to [0,1]$ by $d'(M \mid O , N \mid O) = \metric(M,N)$ for all $M,N,O \in \cnamed{}$.
To prove the thesis it is enough to show that $d'$ is a weak simulation quasimetric.
In fact, since $\metric$ is the minimal weak simulation quasimetric, this implies $\metric \sqsubseteq d'$, thus giving $\metric(M \mid O , N \mid O) \le d'(M \mid O , N \mid O) = \metric(M,N)$.
Hence we need to show that $d'(M \mid O , N \mid O) < 1$ and $M \mid O \transSim{\alpha} \Delta$ implies that there is a transition $N \mid O \TransSim{\hat{\alpha}} \Theta$ with $\Kantorovich(d')(\Delta,\Theta + (1 -\size{\Theta})\overline{\dummyN}) \le d'(M \mid O , N \mid O)$.
We distinguish three cases.
\begin{itemize}
\item
$M \mid O \transSim{\alpha} \Delta$ follows by $M \transSim{\alpha} \Delta'$ and $\Delta = \Delta' \mid \overline{O}$ by applying rule \rulename{TauPar}
\item
$M \mid O \transSim{\alpha} \Delta$ follows by $O \transSim{\alpha} \Delta'$ and $\Delta = \overline{M} \mid \Delta'$
by applying rule \rulename{TauPar} 
\item
$M \mid O \transSim{\alpha} \Delta$ follows by $M \transSim{\alpha_1} \Delta_1$ and $O \transSim{\alpha_2} \Delta_2$ and $\Delta = \Delta_1 \mid \Delta_2$ by applying one of the rules 
\rulename{RcvPar}, \rulename{Bcast}, \rulename{$\sigma$-Par}.
\end{itemize}
We consider the first case, the others are analogous. 
Since $M \mid O \transSim{\alpha} \Delta$ is obtained by applying rule \rulename{TauPar}, we have $\alpha = \tau$.
By $d'(M \mid O,N\mid O) < 1$ we get $\metric(M,N) <1$, which, together with $M \transSim{\alpha} \Delta'$ gives $N \TransSim{\hat{\alpha}} \Theta'$, with $\size{\Theta'} = 1$ following by $\alpha = \tau$, and $\Kantorovich(\metric)(\Delta',\Theta') \le \metric(M,N)$.
By applying the rule \rulename{TauPar} possibly several times we can derive $N \mid O \TransSim{\hat{\alpha}} \Theta' \mid \overline{O}$.
We have now to show that $\Kantorovich(d')(\Delta' \mid \overline{O},\Theta' \mid \overline{O}) \le d'(M \mid O , N \mid O)$, namely that $\Theta' \mid \overline{O}$ is the distribution $\Theta$ we were looking for.
To this end, we take one of the optimal matchings $\omega \in \Omega(\Delta',\Theta')$ realising $\Kantorovich(\metric)(\Delta',\Theta')$ and
we define $\omega' \colon \cnamed{} \times \cnamed{} \to [0,1]$ by $\omega'(M' \mid O,N' \mid O) = \omega(M',N')$ and
$\omega'(M'',N'') = 0$ whenever $M''$ is not of the form $M' \mid O$ or $N''$ is not of the form $N' \mid O$.
Thus, we prove:
\begin{enumerate}
\item \label{muovesxi} $\omega'$ is a matching in $\Omega(\Delta' \mid \overline{O},\Theta' \mid \overline{O})$, and 
\item \label{muovesxii} $\sum_{M'',N'' \in \cnamed{}} d'(M'',N'') \cdot \omega'(M'',N'') \le d'(M' \mid O ,N' \mid O)$.
\end{enumerate}
To show (\ref{muovesxi}) we prove that the left marginal of $\omega'$ is $\Delta' \mid \overline{O}$ by \\
\begin{math}
\begin{array}{rl}
& 
\sum_{N'' \in \cnamed{}} \omega'(M' \mid O,N'') 
\\[2 pt]
= &
\sum_{N' \in \cnamed{}} \omega'(M' \mid O , N' \mid O) 
\\[2 pt]
= &
\sum_{N' \in \cnamed{}} \omega(M' , N') 
\\[2 pt]
= &
\Delta'(M') 
%& \text{(by } \omega \in \Omega(\Delta',\Theta')\text{)}
\\[2 pt]
= &
(\Delta' \mid \overline{O}) (M' \mid O)
\end{array}\\
\end{math}
with the second last equality by $\omega \in \Omega(\Delta',\Theta')$.
The proof that  the right marginal of  $\omega'$ is $\Theta' \mid \overline{O'}$ can be proved analogously.
Then, (\ref{muovesxii}) follows by\\[1pt]
\begin{math}
\begin{array}{rl}
& 
\sum_{M'',N'' \in \cnamed{}} d'(M'',N'') \cdot \omega'(M'',N'') 
\\[2 pt]
= & \sum_{M',N' \in \cnamed{}} d'(M' \mid O,N' \mid O) \cdot \omega(M',N')
\end{array}
\end{math}\\
\begin{math}
\begin{array}{rl}
= & \sum_{M',N' \in \cnamed{}} \metric(M',N') \cdot \omega(M',N')
\\[2 pt]
= &\Kantorovich(\metric)(\Delta',\Theta')  
\\[2 pt]
\le &
\metric(M',N') 
\\[2 pt]
= & d'(M' \mid O', N' \mid O').
\end{array}\\
\end{math}
with the second last equality by the fact that $\omega $ is an optimal matching in $\Omega(\Delta',\Theta')$.
\end{proof}

\noindent
\emph{Proof of Theorem~\ref{cor:non_expansiveness}.}
By  Lemma~\ref{lem:non_expansiveness_support} and transitivity (Proposition~\ref{prop:transitivity}).
\qed

%%%
%%%
%%%

\noindent
\emph{Proof of Proposition~\ref{prop:propagation-general}.}
Assume that $M \transSim{\alpha} \Delta$.
By the hypothesis $M \sqsubseteq_{s_i} \node \nu n {  P_i }  | N$ we have $\node \nu n {  P_i }  | N \TransSim{\hat{\alpha}} \Delta_i$ with $\Kantorovich(\metric)(\Delta,\Delta_i + (1-\size{\Delta_i}) \overline{\dummyN}) \le s_i$.
Hence we have $\node \nu n  {\tau.\bigoplus_{i \in I} p_i {:} P_i } \TransSim{\hat{\alpha}} \sum_{i \in I}p_i \Delta_i$.
It remains to show the inequality $\Kantorovich(\metric)(\Delta,\sum_{i \in I}p_i \Delta_i + (1-\size{\sum_{i \in I}p_i \Delta_i}) \overline{\dummyN}) \le \sum_{i \in I}p_i s_i$.
Assume that $\omega_i$ is one of the optimal matchings realising $\Kantorovich(\metric)(\Delta ,\Delta_i + (1-\size{\Delta_i}) \overline{\dummyN})$.
Define the function $\omega \colon \cnamed{} \times \cnamed \to [0,1]$ by $\omega(M',N') = \sum_{i \in I} p_i \omega_i(M',N')$ for all networks $M',N' \in \cnamed{}$.
To conclude that $\Kantorovich(\metric)(\Delta ,\sum_{i \in I}p_i \Delta_i + (1-\size{\sum_{i \in I}p_i \Delta_i}) \overline{\dummyN}) \le \sum_{i \in I}p_is_i$, we show that: 
 \begin{enumerate}
\item 
\label{thm:propagation-general_item1} $\omega$ is a matching in $\Omega(\Delta,\sum_{i \in I}p_i \Delta_i + (1-\size{\sum_{i \in I}p_i \Delta_i}) \overline{\dummyN})$ and
\item 
\label{thm:propagation-general_item2} $\sum_{M',N' \in \cnamed{}} \omega(M',N') \cdot d(M',N') \le \sum_{i \in I}p_i s_i$.
\end{enumerate}
Let us start with (\ref{thm:propagation-general_item1}).
The proof that the left marginal of $\omega$ is $\Delta$ is immediate.
We show
that the right marginal of $\omega$ is the distribution $\sum_{i \in I}p_i \Delta_i + (1-\size{\sum_{i \in I}p_i \Delta_i})\overline{\dummyN}$ by\\[2pt]
\begin{math}
\begin{array}{rl}
& \sum_{M' \in \cnamed{}} \omega(M',N') \\[2 pt]
 = \quad &  \sum_{M' \in \cnamed{}} \sum_{i \in I} p_i \cdot \omega_i(M',N') \\[3pt]
= \quad &  \sum_{i \in I} p_i  \sum_{M' \in \cnamed{}} \omega_i(M',N')  \\[3pt]
= \quad &  \sum_{i \in I} p_i  \cdot  (\Delta_i + (1-\size{\Delta_i}) \overline{\dummyN})(N') 
\\[3 pt]
= \quad & (\sum_{i \in I}p_i \Delta_i + (1-\size{\sum_{i \in I}p_i \Delta_i})\overline{\dummyN})(N')
\end{array}
\end{math}\\[2pt]
for each $N' \in \cnamed{}$,
with the second last equality by $\omega_i \in \Omega(\Delta,\Delta_i + (1-\size{\Delta_i}) \overline{\dummyN})$.
Consider now (\ref{thm:propagation-general_item2}).
We have\\[2pt]
\begin{math}
\begin{array}{rl}
& \sum_{M',N' \in \cnamed{}} \omega(M',N') \cdot d(M',N') \\[4 pt]
= \quad & \sum_{M',N' \in \cnamed{}} \sum_{i \in I} p_i \cdot \omega_i(M',N')  \cdot d(M',N') \\[4 pt]
= \quad & \sum_{i \in I} p_i \sum_{M',N' \in \cnamed{}}  \omega_i(M',N')  \cdot d(M',N') \\
= \quad & \sum_{i \in I} p_i \cdot \Kantorovich(\metric)(\Delta ,\Delta_i+ (1-\size{\Delta_i})\overline{\dummyN} )  \\[4 pt]
%& \text{(by definition of $\omega_i$)}\\[4 pt]
\le \quad & \sum_{i \in I} p_i \cdot s_i 
\end{array}
\end{math}\\[2pt]
with the last equality by the definition of $\omega_i$.
This concludes the proof.
\qed

%%%
%%%
%%%

\noindent
\emph{Proof of Proposition~\ref{thm:timing}.}
\begin{enumerate}
\item    
We proceed by induction over $k$. The base case $k=0$ is immediate.
Consider the inductive step $k+1$.
We prove that any transition by $\node \mu n {\sigma^{k+1}.\nil}$ can be simulated by $\node \mu n {\nil}$,
the symmetric case is analogous.
The only transition by $\node \mu n {\sigma^{k+1}.\nil}$ is the transition $\node \mu n {\sigma^{k+1}.\nil} \transSim{\sigma} \overline{\node \mu n {\sigma^{k}.\nil}}$ derived by rule \rulename{Sleep}.
By an application of rule \rulename{$\sigma$-nil} we get $\node \mu n {\nil} \transSim{\sigma} \overline{\node \mu n {\nil}}$. 
We have $\Kantorovich(\metric)(\overline{\node \mu n {\sigma^{k}.\nil}},\overline{\node \mu n {\nil}}) = \metric(\node \mu n {\sigma^{k}.\nil},\node \mu n {\nil}) = 0$, where the last equality follows by the inductive hypothesis.
\item 

By the fact that the only transitions that can be made by the two networks $\prod_{i \in I}\node {\nu_{m_i}} {m_i} {\sigma.P_i}$ and $\prod_{j \in J}\node {\mu_{n_j}} {n_j} {\sigma.Q_j}$ are 
$\prod_{i \in I}\node {\nu_{m_i}} {m_i} {\sigma.P_i} \transSim{\sigma} \overline{ \prod_{i \in I} \node {\nu_{m_i}} {m_i} {P_i}}$ and
$\prod_{j \in J}\node {\mu_{n_j}} {n_j} {\sigma.Q_j} \transSim{\sigma} \overline{\prod_{j \in J} \node {\mu_{n_j}} {n_j} {Q_j}}$ (see rules \rulename{Sleep} and \rulename{$\sigma$-Par}).

\item
The hypothesis that   nodes in $\mu$ do not send in the current time interval imply that the only transition by $\node \mu n {\rcvtime x C D}$  is $\node \mu n {\rcvtime x C D} \transSim{\sigma}  \sem{ \node \mu n {D}}$.
On the other hand, the only transition by $\node \mu n {\sigma.D}$ is $\node \mu n {\sigma.D} \transSim{\sigma}  \sem{\node{\nu}{n}{D}}$.
The thesis immediately follows.

\item
This is an immediate corollary of item \ref{timing34}.

\item
By  Proposition~\ref{M_go_to_N},
it is enough to prove that
\begin{equation}
\label{tthm5.6}
\node {\nu} m {\tau.(\bcastzero{v} \oplus_p \nil)} \q \big| \q 
\prod_{i \in I}
\node {\nu_{i}} {n_i} {P_i} \q \TransSim{\tau} \q
\overline{ \node {\nu} m {\nil} \q \big| \q
\prod_{i \in I}
\node {\nu_{i}} {n_i} {P_i} }
\end{equation}
By rule \rulename{Tau} we get
\[
\node {\nu} m {\tau.(\bcastzero{v} \oplus_p \nil)} \q \transSim{\tau} \q p \overline{ \node {\nu} m {\bcastzero{v}} }  + (1-p) \overline{ \node {\nu} m {\nil } }.
\]
Hence, to prove Equation~\ref{tthm5.6} it suffices to show that
\begin{equation}
\label{tthm5.6bis}
\node {\nu} m {\bcastzero{v}} \q big| \q \prod_{i \in I}
\node {\nu_{i}} {n_i} {P_i} \q
\TransSim{\tau} \q 
\overline{ \node {\nu} m {\nil} \q \big| \q
\prod_{i \in I}
\node {\nu_{i}} {n_i} {P_i} }
\end{equation}
To this end, first of all we note that, since for all $i \in I$    $P_i\neq \rcvtime x C D$, by rule \rulename{RcvEnb} we have that
\[
\node {\nu_{i}} {n_i} {P_i} \q \transSim{ \rcva {m} v}  \q  \overline{\node {\nu_{i}} {n_i} {P_i}} 
\]
and, then, by rule \rulename{RcvPar} we get 
$\prod_{i \in I}\node {\nu_{i}} {n_i} {P_i} \transSim{ \rcva {m} v} \prod_{i \in I}  \overline{\node {\nu_{i}} {n_i} {P_i}}$. 
By an application of rule \rulename{Snd} we have 
$\node {\nu} m {\bcastzero{v}} \transSim{\sndto {m} v {\nu}} \overline{ \node {\nu} m {\nil } }$.
Finally, we can apply rule \rulename{Bcast} to the last two transitions, 
and, since $\nu = \bigcup_{i \in I}n_i$, by applying rule \rulename{Shh} 
we get the proof obligation Equation~\ref{tthm5.6bis}. \qed
\end{enumerate}

%%%
%%%
%%%
%%%
%%%
%%%

\subsection{Proofs of results in Section~\ref{sec:gossip-nocollisions}}
\label{section:proofNoCollision}

\noindent
First we describe how we prove Theorem~\ref{thm:propagation}.
Let $O$ denote the network 
$O  \deff  N  |
\prod_{i  \in I} \node {\nu_{m_i}} {m_i} {\nil} |
 \prod_{j \in J} \node {\nu_{n_j}} {n_j}
{\wsndv_{q_j}}$.
We will prove that for some distribution $\Delta$ there is a transition $M \TransSimone[\hat{\tau}](1-\prod_{i \in I}(1-p_i))\cdot \overline{O} + \prod_{i \in I}(1-p_i) \cdot \Delta$, then the thesis follows by  Proposition~\ref{M_go_to_N}.
The proof that $M \TransSimone[\hat{\tau}](1-\prod_{i \in I}(1-p_i))\cdot\overline{O} + \prod_{i \in I}(1-p_i)\cdot \Delta$ is divided in two parts.
First we prove that with probability $(1-\prod_{i \in I}(1-p_i))$ the network $M$ can reach through a sequence of $\tau$-transitions a network 
$O_{I_1,I_2,I_3}$ of the form 
$O_{I_1,I_2,I_3} = N \big| 
\prod_{i \in I_1}  \node {\nu_{m_i}} {m_i} { \bcastzero v }   \big| 
\prod_{i \in I_2} \node {\nu_{m_i}} {m_i} { \nil }  \big| $
$ \prod_{i \in I_3} \node {\nu_{m_i}} {m_i} {\sndv_{p_i}} \big| 
\prod_{j\in J} \node {\nu_{n_j}} {n_j} {\wsndv_{q_j}}$,
for $I_1 \cup I_2 \cup I_3$ some partition of $I$.
In particular, one of the $\tau$-transitions in this sequence is derived via rule \rulename{ShhSnd} as a consequence of the broadcast of the message by one of the senders.
Then we prove that $O_{I_1,I_2,I_3} \TransSim{\tau} \overline{O}$.
The first part of the proof is given directly in the proof of Theorem~\ref{thm:propagation}, whereas 
$O_{I_1,I_2,I_3} \TransSim{\tau} \overline{O}$ is proved in the following Lemma. 
\begin{lem}\label{lem:propagation_support} 
Assume the hypothesis of  Theorem~\ref{thm:propagation}.
Let $O$ denote the network 
\[
\begin{array}{c}
O  \deff  N  \q \big| \q
\prod_{i  \in I} \node {\nu_{m_i}} {m_i} {\nil} \q \big| \q
 \prod_{j \in J} \node {\nu_{n_j}} {n_j}
{\wsndv_{q_j}}
\end{array}
\]
and $O_{I_1,I_2,I_3}$ denote the network
\[
\begin{array}{c}
N  \;   \big|  \;  
\prod_{i \in I_1}  \node {\nu_{m_i}} {m_i} { \bcastzero v }  \;  \big|  \;  
\prod_{i \in I_2} \node {\nu_{m_i}} {m_i} { \nil }   \;   \big| \;
\prod_{i \in I_3} \node {\nu_{m_i}} {m_i} {\sndv_{p_i}}  \;    \big|  \;  
\prod_{j\in J} \node {\nu_{n_j}} {n_j} {\wsndv_{q_j}} 
\end{array}
\]
for $I_1 \cup I_2 \cup I_3$ any partition of $I$.
Then, we have
$
O_{I_1,I_2,I_3} \TransSim{\tau} \overline{O}$. 
\end{lem}
\begin{proof}
We reason by induction over $\size{I_1 \cup I_3}$.
The base case $\size{I_1 \cup I_3} = 0$ is immediate, since $I_2 = I$ and $O_{I_1,I_2,I_3}$ is the network $O$.
Consider the inductive step $\size{I_1 \cup I_3} = n+1$. 
We distinguish two sub-cases, depending whether $I_3$ is empty or not.\\[0.1 ex]

\noindent
\underline{Case $I_3 \neq \emptyset$}.
In this case the network $O_{I_1,I_2,I_3}$ can perform a $\tau$-step originated by a sender $m_k$ with $k \in I_3$.
Formally, by rule \rulename{Tau}  we get
\[
\begin{array}{c}
\node {\nu_{m_k}} {m_k} {\sndv_{p_k}} \transSim{\tau} p_k \cdot  \overline{\node {\nu_{m_k}} {m_k} { \bcastzero v } } + (1-p_k)\cdot \overline{\node {\nu_{m_k}} {m_k} { \nil }}
\end{array}
\]
and, then, by an application of rule \rulename{TauPar} we get
\begin{equation*}
\begin{array}{rcl}
O_{I_1,I_2,I_3} &\transSim{\tau} &
\overline{N} \q \big| \q
\prod_{i \in I_1}  \overline{\node {\nu_{m_i}} {m_i} { \bcastzero v } }  \q \big| \q
\prod_{i \in I_2} \overline{\node {\nu_{m_i}} {m_i} { \nil } } \q \big| \;q
\prod_{i \in I_3 \setminus{\{k\}} } \overline{\node {\nu_{m_i}} {m_i} {\sndv_{p_i}} }    \q \big|  \\[4pt]
&& (p_k \cdot \overline{\node {\nu_{m_k}} {m_k} { \bcastzero v }  } + (1-p_k)\cdot \overline{\node {\nu_{m_k}} {m_k} { \nil }} ) \q \big| \q 
\prod_{j\in J} \overline{\node {\nu_{n_j}} {n_j} {\wsndv_{q_j}}}
\end{array}
\end{equation*}
namely 
\begin{equation}
\begin{array}{c}
\label{prop_supp_eq1}
O_{I_1,I_2,I_3} \q \transSim{\tau} \q p_k  \cdot \overline{O_{I_1\cup\{k\},I_2,I_3\setminus\{k\}}} + (1-p_k)  \cdot \overline{O_{I_1,I_2\cup\{k\},I_3\setminus\{k\}}} 
\end{array}
\end{equation}

Consider first the network $O_{I_1\cup\{k\},I_2,I_3\setminus\{k\}}$ in Equation~\ref{prop_supp_eq1}.
A $\tau$-transition by network $O_{I_1\cup\{k\},I_2,I_3\setminus\{k\}}$ models the broadcasting of $v$ by $m_k$ that is not received by any other node. In detail,
by rule \rulename{Snd} we get
\[
\begin{array}{c}
\node {\nu_{m_k}} {m_k} { \bcastzero v }\transSim{\sndto {m_k} v {\nu_{m_k}}}
\overline{\node {\nu_{m_k}} {m_k} {\nil}} \, .
\end{array}
\]
Then, by an application of rule \rulename{RcvEnb} we infer
\[
\begin{array}{c}
N   \transSim{\rcva {m_k} v} \overline{N} 
\end{array}
\]
which can be applied since by the hypothesis we know that no node in $N$ is able to receive, by the same  \rulename{RcvEnb} we get
\[
\node {\nu_{m_i}} {m_i} { \bcastzero v }   \transSim{\rcva {m_k} v} \overline{\node {\nu_{m_i}} {m_i} { \bcastzero v }} 
\]
for all $i \in I_1$,  
\[
\node {\nu_{m_i}} {m_i} {\nil}  \transSim{\rcva {m_k} v} \overline{\node {\nu_{m_i}} {m_i} {\nil}}
\]
for all $i \in I_2$, and
\[
\node {\nu_{m_i}} {m_i} {\sndv_{p_i}}   \transSim{\rcva {m_k} v} \overline{\node {\nu_{m_i}} {m_i} {\sndv_{p_i}}}
\]
for all $i \in I_3 \setminus \{k\}$, then again by rule \rulename{RcvEnb} we infer
\[
\node {\nu_{n_j}} {n_j} {\wsndv_{q_j}} \transSim{\rcva {m_k} v} \overline{ \node {\nu_{n_j}} {n_j} {\wsndv_{q_j}}}
\]
for all $j \in J$, thus allowing us to apply rules \rulename{RcvPar} and \rulename{Bcast} and obtain
\[
O_{I_1\cup\{k\},I_2,I_3\setminus\{k\}} \transSim{\sndto {m_k} v {\nu_{m_k} \setminus \nds M}} \overline{O_{I_1,I_2\cup\{k\},I_3 \setminus\{k\}}}
\]
from which we get 
\begin{equation}
\label{prop_supp_eq14giu}
\begin{array}{c}O_{I_1\cup\{k\},I_2,I_3\setminus\{k\}} \transSim{\tau} \overline{O_{I_1,I_2\cup\{k\},I_3\setminus\{k\}}}
\end{array}
\end{equation}
by rule \rulename{ShhSnd}, which can be applied since by the hypothesis we know that $\nu_{m_k} \subseteq \nds M$.
Now, since $\size{I_1 \cup I_3 \setminus\{k\}}=n$ we can apply the inductive hypothesis thus obtaining $O_{I_1,I_2\cup\{k\},I_3\setminus\{k\}} \TransSim{\tau}  \overline{O}$.
Finally, by Equation~\ref{prop_supp_eq14giu}
and $O_{I_1,I_2\cup\{k\},I_3\setminus\{k\}} \TransSim{\tau}  \overline{O}$ we get
\begin{equation}
\begin{array}{c}
\label{prop_supp_eq2}
O_{I_1\cup\{k\},I_2,I_3\setminus\{k\}} \TransSim{\tau} \overline{O} 
\end{array}
\end{equation}

Consider now the network 
$O_{I_1,I_2\cup\{k\},I_3\setminus\{k\}}$ in Equation~\ref{prop_supp_eq1}.
Since $\size{I_1 \cup I_3\setminus\{k\}}=n$ we can apply the inductive hypothesis thus giving $O_{I_1,I_2\cup\{k\},I_3\setminus\{k\}} \TransSim{\tau} \overline{O}$, which, together with Equation~\ref{prop_supp_eq1} and Equation~\ref{prop_supp_eq2}, gives the thesis. \\[0.1 ex]

\noindent
\underline{Case $I_3 = \emptyset$}. Since $\size{I_1 \cup I_3} = n+1$, we are sure that $I_1 \neq \emptyset$.
For some $k \in I_1$, we can reason as above to infer that from a broadcasting transition by some $m_k$ with $k \in I_1$ we get
$O_{I_1,I_2,I_3} \transSim{\tau} \overline{O_{I_1\setminus\{k\},I_2\cup\{k\},I_3}}$, then by the inductive hypothesis we derive 
$O_{I_1\setminus\{k\},I_2\cup\{k\},I_3} \TransSim{\tau} \overline{O}$, and, finally, from these two transitions we conclude $O_{I_1,I_2,I_3} \TransSim{\tau} \overline{O}$.
\end{proof}

%%%
%%%
%%%

We are now ready to prove  Theorem~\ref{thm:propagation}.

\noindent
\emph{Proof of Theorem~\ref{thm:propagation}.}
As in the proof of  Lemma~\ref{lem:propagation_support}, we use $O$ to denote the network 
\[
\begin{array}{c}
O  \deff  N  \q \big| \q
\prod_{i  \in I} \node {\nu_{m_i}} {m_i} {\nil} \q \big| \q
 \prod_{j \in J} \node {\nu_{n_j}} {n_j}
{\wsndv_{q_j}} \enspace .
\end{array}
\]

We prove that $M \TransSimone[\hat{\tau}](1-\prod_{i \in I}(1-p_i))\overline{O} + \prod_{i \in I}(1-p_i) \Delta$, for some distribution $\Delta$, then the thesis follows by  Proposition~\ref{M_go_to_N}.

By rule \rulename{Tau} we derive the following transition step for any sender $m_i$ with $i \in I$:
\[
\node {\nu_{m_i}} {m_i} {\sndv_{p_i}} \transSim{\tau} p_i  \cdot \overline{\node {\nu_{m_i}} {m_i} { \bcastzero v } } + (1-p_i) \cdot \overline{\node {\nu_{m_i}} {m_i} { \nil }} .
\]
The network $M$ can start by performing $k \le \size{I}$ transitions labelled $\tau$, where each of these transitions comes from rule \rulename{Tau} applied as above (namely is the initial $\tau$-step by some sender) and \rulename{TauPar}.
Therefore, we can partition $I$ as $I=I' \cup I''$ such that 
\[
\begin{array}{c}
M \transSim{\tau} \ldots \transSim{\tau} \Theta
\end{array} 
\]
with 
\begin{equation}\label{thm:propagation_step1}
\Theta  \! = \!
\overline{N}  
| 
\prod_{i \in I'}  (p_i  \overline{\node {\nu_{m_i}} {m_i} { \bcastzero v }  } + (1{-}p_i) \overline{\node {\nu_{m_i}} {m_i} { \nil }} ) 
| 
\prod_{i \in I''} \overline{\node {\nu_{m_i}} {m_i} {\sndv_{p_i}}}  
| 
\prod_{j\in J} \overline{\node {\nu_{n_j}} {n_j} {\fwd_{q_j}}} 
\end{equation}
and now we distinguish two cases.

The first case is that all senders $m_i$ with $i \in I$ reach the state $\node {\nu_{m_i}} {m_i} { \nil }$ through their initial $\tau$ transition, namely in the Equation~\ref{thm:propagation_step1} we have that $I'= I$, $I'' = \emptyset$ and the network in the support of $\Theta$ that is reached by $M$ is $N  |   \prod_{i \in I} \node {\nu_{m_i}} {m_i} { \nil }   |  \prod_{j\in J} \node {\nu_{n_j}} {n_j} {\fwd_{q_j}}$.
Formally, 
\begin{equation}
\label{eq:cannotSimulate}
\begin{array}{c}
M \TransSim{\hat{\tau}}  q \cdot (\overline{N} 
\q  \big|  \q
\prod_{i \in I}  \overline{\node {\nu_{m_i}} {m_i} { \nil } } 
 \q \big| \q
\prod_{j\in J} \overline{\node {\nu_{n_j}} {n_j} {\fwd_{q_j}}})
\end{array}
\end{equation}
 for  $q=\prod_{i \in I}(1-p_i)$.  
The network $N  |  \prod_{i \in I} \node {\nu_{m_i}} {m_i} { \nil }  |  \prod_{j\in J} \node {\nu_{n_j}} {n_j} {\fwd_{q_j}}$ is then unable to simulate $O$.

The second case is that at least one of the senders $m_i$ decides to transmit and reaches the state $\node {\nu_{m_i}} {m_i} { \bcastzero v }$ through the $\tau$ transition, namely there is a partition of $I' = I_1' \cup I_2'$ such that the state in the support of distribution $\Theta$ that is reached by $M$ has the form
$N \big| 
\prod_{i \in I_1'}  \node {\nu_{m_i}} {m_i} { \bcastzero v }   \big| 
\prod_{i \in I_2'}  \node {\nu_{m_i}} {m_i} { \nil }   \big|  
\prod_{i \in I''} \node {\nu_{m_i}} {m_i} {\sndv_{p_i}}  \big|   
\prod_{j\in J} \node {\nu_{n_j}} {n_j} {\fwd_{q_j}}$
with $I'_1 \neq \emptyset$.
Formally, 
\begin{equation}\label{thm:propagation_step1_detailed}
M \TransSim{\hat{\tau}}  q'  \cdot  \Big(\overline{N} 
  | \!
\prod_{i \in I_1'}  \overline{\node {\nu_{m_i}} {m_i} { \bcastzero v } }  
 | \!
\prod_{i \in I_2'}  \overline{\node {\nu_{m_i}} {m_i} { \nil } }  
 | \!
 \prod_{i \in I''} \overline{\node {\nu_{m_i}} {m_i} {\sndv_{p_i}}} 
 |  \!
\prod_{j\in J} \overline{\node {\nu_{n_j}} {n_j} {\fwd_{q_j}}}\Big)
\end{equation}
with $q' = \prod_{i \in I'_1}p_i \cdot \prod_{i \in I'_2}(1-p_i)$. 
In this case, each of the nodes $m_i$ with $i \in I'_1$ can then broadcasts $v$. Indeed, by applying the rule \rulename{Snd} we get
%%\[
$\node {\nu_{m_h}} {m_h} { \bcastzero v }\transSim{\sndto {m_h} v {\nu_{m_h}}}
\overline{\node {\nu_{m_h}} {m_h} {\nil}} $,  
%%\]
for some $h \in I'_1$.
Since by the hypothesis we know that all nodes in $\nu_{m_h} \cap \nds{N}$ cannot receive in the current round, and, analogously, all nodes $n_i$ with $i \in I \setminus \{h\} = (I'_1 \cup I'_2\cup I'')\setminus \{h\}$ cannot receive, by applying rule \rulename{RcvEnb} we infer:
%%\begin{inparaenum}[(a)]
%%\item 
(a) $N \transSim{\rcva {m_h} v} \overline{N}$,
%%\item 
(b)
$\node {\nu_{m_i}} {m_i} {\sndv_{p_i}}   \transSim{\rcva {m_h} v}   \overline{\node {\nu_{m_i}} {m_i} {\sndv_{p_i}}}$,
%%\item 
(c) $\node {\nu_{m_i}} {m_i} {\nil}   \transSim{\rcva {m_h} v}  \overline{ \node {\nu_{m_i}} {m_i} {\nil}}$. 
%%\end{inparaenum}.
Moreover, for all $j\in J$, by rule \rulename{Rcv} we infer
\[
\node {\nu_{n_j}} {n_j} {\fwd_{q_j}} \transSim{\rcva {m_h} v} \overline{\node {\nu_{n_j}} {n_j} {\wsndv_{q_j}}}
\]
thus implying that we can apply rules \rulename{RecPar} and \rulename{Bcast} and obtain 
\begin{equation*}
{\scriptsize
\begin{array}{c}
\displaystyle
N \q 
\big|  \q
\prod_{i \in I_1'}  \node {\nu_{m_i}} {m_i} { \bcastzero v }  \q 
\big|  \q
\prod_{i \in I_2'} \node {\nu_{m_i}} {m_i} { \nil }  \q 
\big| \q 
\prod_{i \in I''} \node {\nu_{m_i}} {m_i} {\sndv_{p_i}} \q 
\big| \q 
\prod_{j\in J} \node {\nu_{n_j}} {n_j} {\fwd_{q_j}}  \\[8pt]
\transSimone[\sndto {m_h} v {\nu_{m_h} \setminus \nds M}] \\[8pt]
\displaystyle \overline{N} 
 \q \big| \q
 \overline{\node{\nu_{m_h} } {m_h} { \nil } }  
 \q \big|
\prod_{i \in I_1'\setminus\{h\}}  \overline{\node {\nu_{m_i}} {m_i} { \bcastzero v } } 
\q \big|
\prod_{i \in I_2'} \overline{\node {\nu_{m_i}} {m_i} { \nil } }  
\q \big|
\prod_{i \in I''} \overline{\node {\nu_{m_i}} {m_i} {\sndv_{p_i}} }  
\q \big|
\prod_{j\in J} \overline{\node {\nu_{n_j}} {n_j} {\wsndv_{q_j}}} \, . 
\end{array}
}
\end{equation*}
Since by the hypothesis we have $\nu_{m_h} \subseteq \nds{M}$, we can apply rule \rulename{ShhSnd} to infer 
\begin{equation}
\label{thm:propagation_step2}
{\scriptsize 
\begin{array}{c}
\displaystyle
 N \q \big| \q 
\prod_{i \in I_1'}  \node {\nu_{m_i}} {m_i} { \bcastzero v }   
\q \big| \q  
\prod_{i \in I_2'} \node {\nu_{m_i}} {m_i} { \nil }
\q \big| \q
\prod_{i \in I''} \node {\nu_{m_i}} {m_i} {\sndv_{p_i}} 
\q\big| \q
\prod_{j\in J} \node {\nu_{n_j}} {n_j} {\fwd_{q_j}}  
\\[1pt]
\transSim{\tau}  \qquad \qquad 
\\[4pt]
\displaystyle
\overline{N} 
\; \big| \;
 \overline{\node{\nu_{m_h} } {m_h} { \nil } } 
\;\big| 
\prod_{i \in I_1'\setminus\{h\}}  \overline{\node {\nu_{m_i}} {m_i} { \bcastzero v } }  \;
\;\big| 
\prod_{i \in I_2'} \overline{\node {\nu_{m_i}} {m_i} { \nil } }  \;
\;\big| 
\prod_{i \in I''} \overline{\node {\nu_{m_i}} {m_i} {\sndv_{p_i}} } \;
 \;\big| 
\prod_{j\in J} \overline{\node {\nu_{n_j}} {n_j} {\wsndv_{q_j}}}
\end{array}
}
\end{equation}
Let $O'$ denote the only network in the support of the distribution in Equation~\ref{thm:propagation_step2}, i.e.
$N \big| 
\prod_{i \in I_1'\setminus\{h\}}  \node {\nu_{m_i}} {m_i} { \bcastzero v }  
\big| 
\prod_{i \in I_2' \cup \{h\}} \node {\nu_{m_i}} {m_i} { \nil }
\big| 
 \prod_{i \in I''} \node {\nu_{m_i}} {m_i} {\sndv_{p_i}} 
\big| 
\prod_{j\in J} \node {\nu_{n_j}} {n_j} {\wsndv_{q_j}}$.
We observe that by adopting the notation in  Lemma~\ref{lem:propagation_support}, $O'$ is in the form $O_{I_1,I_2,I_3}$ with $I_1=I'_1\setminus\{h\}$, $I_2=I'_2\cup \{h\}$, $I_3=I''$.
By  Lemma~\ref{lem:propagation_support} we have $O' \TransSim{\tau} \overline{O}$. 
This transition together with Equation~\ref{thm:propagation_step1_detailed} and Equation~\ref{thm:propagation_step2} gives $M \TransSim{{\tau}} q' \overline{O}$, where we remind that $q' = \prod_{i \in I'_1}p_i \cdot \prod_{i \in I'_2}(1-p_i)$. 
This transition is derived for any partition $I'_1 \cup I'_2 \cup I''$ of $I$ with $I'_1 \neq \emptyset$.
Hence $M \TransSim{{\tau}} (1-q) \overline{O}$, with $q$ the probability value in Equation~\ref{eq:cannotSimulate}, which gives the thesis.
\qed

%%%%
%%%%
%%%%

\noindent
\emph{Proof of Theorem~\ref{thm:propagation2}}.
Let us denote with $M_0$, $M_1$ and $M_2$ the following networks:
\begin{itemize}
\item 
$
M_0 \: = \:  N \q \big| \q
\node {\nu_{m}} {m} {\bcastzero v} \q\big| \q
\prod_{j\in J} \node {\nu_{n_j}} {n_j} {\rcvtime {x_j}{P_j}{Q_j}}$
\item 
$M_1 \: = \: N\q \big| \q \node {\nu_{m}} {m} {\nil} \q \big| \q 
\prod_{j\in J} \node {\nu_{n_j}} {n_j} {{\subst v {x_j}}P_j}$
\item 
$M_2 \: = \:  N \q \big| \q  \node {\nu_{m}} {m} {\nil} \q \big| \q
\prod_{j\in J} \node {\nu_{n_j}} {n_j} {\rcvtime {x_j}{P_j}{Q_j}}$. 
\end{itemize}
We recall that $\sndv_p = \tau.(\bcastzero v \oplus_{p} \nil)$. 
Therefore we infer the thesis from  Proposition~\ref{prop:propagation-general} if we show that $M  \simtol{s_1} M_0$ and $M  \simtol{s_2} M_2$.
The relation $M \simtol{s_2} M_2$ holds by the hypothesis, hence we have to prove $M  \simtol{s_1} M_0$.
To this purpose, it suffices to prove $M_1 \simtol{0} M_0$. 
In fact, $M_1  \simtol{0} M_0$ and the hypothesis $M  \simtol{s_1} M_1$ give, by transitivity (Proposition~\ref{prop:transitivity}),  $M  \simtol{s_1} M_0$.
Finally, to prove $M_1  \simtol{0} M_0$, by  Proposition~\ref{M_go_to_N} it is enough to prove that 
$
M_0 \transSim{\tau} \overline{M_1}
$. 

We show now that this transition  can be derived.
By rule \rulename{Snd} we get
$
\node {\nu_{m}} {m} { \bcastzero v }\transSim{\sndto {m} v {\nu_{m}}}
\overline{ \node {\nu_{m}} {m} {\nil}}$. 
By an application of rule \rulename{Rcv} we get
\[
\node {\nu_{n_j}} {n_j} {\rcvtime {x_j}{P_j}{Q_j}}
  \q \transSim{\rcva {m} v} \q 
 \overline{ \node {\nu_{n_j}} {n_j} {{\subst v {x_j}}P_j}}
\]
for all $j\in J$, where \rulename{Rcv} can be applied since we have $m \in \nu_{n_j}$ for all $j \in J$, which follows by the hypothesis  $\{n_j \mid j \in J\} \subseteq \nu_{m}$ and well-formedness. 
By rule \rulename{RcvEnb} we get
$
N \transSim{\rcva {m} v} \overline{N}
$, 
where  \rulename{RcvEnb} can be applied since by the hypothesis we know that all nodes in $\nu_{m} \cap \nds{N}$ are not able to receive in the current instant of time.
Summarising, we can apply rule \rulename{Bcast} to obtain
$
M_0 \transSim{\sndto {m } v {\nu_{m } \setminus \nds M}} \overline{M_1}$. 
Then, since by the hypothesis we have $\nu_{m} \subseteq    \nds{M}$, to conclude the proof, we can apply rule \rulename{ShhSnd} to derive the  
transition $M_0 \transSim{\tau} \overline{M_1} $.
\qed

\noindent
\emph{Proof of Theorem~\ref{thm:propagation2bis}.}
Let us denote with $M_0$ and  $M_i$, for $i\in I$, the following networks:
\begin{itemize}
\item 
$M_0  \: = \:  \prod_{j\in J} \node {\nu_{n_j}} {n_j} {\nil }  \q \big| \q \node {\nu_d} d {\tau.\bigoplus_{i \in I} p_i {:} P_i}$
\item 
$M_i \: = \: \prod_{j\in J} \node {\nu_{n_j}} {n_j} {\nil }  \q \big| \q \node {\nu_d} d { P_i }$. 
\end{itemize}
The only admissible transition for the network $M_0$ is $M_0 \transSim{\tau} \Delta =  \sum_{i \in I} p_i   \cdot \overline{M_i}$, derived by rules \rulename{Tau} and \rulename{TauPar}.
By rules \rulename{Tau} and \rulename{TauPar}, we can derive the transition $M \mid \node {\nu_{m}} {m} {\tau.\bigoplus_{i \in I} p_i {:} Q_i} \transSim{\tau} \Theta= \sum_{i \in I} p_i  \cdot \overline{M  \mid \node {\nu_{m}} {m} {Q_i}}$.
Since $\size{\Theta}=1$, it is enough to prove that 
$\Kantorovich(\metric)(\Delta,\Theta ) \le  \sum_{i \in I} p_i s_i$.
Define $\omega \colon \cnamed{} \times \cnamed \to [0,1]$ by 
$\omega(M_i, M \; \big| \;\node {\nu_{m}} {m} {Q_i } ) = p_i$, for any $i \in I$, 
 and $\omega(M',N') = 0$ for all other pairs $(M',N')$.
To conclude $\Kantorovich(\metric)(\Delta ,\Theta) \le \sum_{i \in I} p_i s_i$, it is enough to show that: 
\begin{enumerate}
\item \label{thm:propagation-general_item1c} the function $\omega$ is a matching in $\Omega(\Delta,\Theta)$ and
\item \label{thm:propagation-general_item2c} $\sum_{M',N' \in \cnamed{}} \omega(M',N') \cdot \metric(M',N') \le
\sum_{i \in I} p_i s_i$.
\end{enumerate}
The case (\ref{thm:propagation-general_item1c}) is immediate.
Consider now (\ref{thm:propagation-general_item2c}).
We have that
$ \sum_{M',N' \in \cnamed{}} \omega(M',N') \cdot \metric(M',N') 
= \sum_{i \in I} p_i   \metric(M_i, M \; \big| \;\node {\nu_{m}} {m} {Q_{i} } )    
\le \sum_{i \in I} p_i s_i$ with the inequality by the hypothesis
$M_i \simtol{s_i} M \mid \node {\nu_{n}} {n} {Q_{i}}$.
\qed

%%%
%%%
%%%
%%%
%%%

\subsection{Proofs of results in Section~\ref{sec:gossip-collisions}}
\
\label{section:proofCollisioni}

\noindent
\emph{Proof of Theorem~\ref{thm:propagation3}.}
Let $O \, = \,   N  \q \big| \q
\prod_{i  \in I} \node {\nu_{m_i}} {m_i} {\nil} \q \big| \q
 \prod_{j \in J} \node {\nu_{n_j}} {n_j}
{\wsndvc_{q_j}}$. 
We prove that 
\begin{equation}\label{thm:propagation_coll_proof_obligation}
\begin{array}{c}
M \TransSim{ \hat{\tau}} ( \sum_{i\in I}p_i \prod_{j\in I{\setminus}\{i\}}(1-p_j)) \overline{O}+ ( 1- \sum_{i\in I}p_i \prod_{j\in I{\setminus}\{i\}}(1-p_j) ) \Delta 
\end{array}
\end{equation}
for some distribution $\Delta$,
then the proof follows by  Proposition~\ref{M_go_to_N}.
By rule \rulename{Tau} we get the following transition step for any sender $m_i$ with $i \in I$:
\[
\node {\nu_{m_i}} {m_i} {\sndv_{p_i}} \transSim{\tau} p_i  \cdot \overline{\node {\nu_{m_i}} {m_i} { \bcastzero v } } + (1-p_i) \cdot \overline{\node {\nu_{m_i}} {m_i} { \nil }} . \]

The network $M$ can start by performing $|\!I\!|$ transitions labelled $\tau$, where each of these transitions comes from rule \rulename{Tau} applied as above (namely is the initial $\tau$-step by some sender) and \rulename{TauPar}.
Therefore, we have
\[
M \transSim{\tau} \ldots \transSim{\tau} \Theta
\]
with
$
\Theta = \overline{N} \; \big| \;
\prod_{i \in I}  \big(p_i  \cdot \overline{\node {\nu_{m_i}} {m_i} { \bcastzero v }  } + (1-p_i) \cdot  \overline{\node {\nu_{m_i}} {m_i} { \nil }} \big) \; \big| \;
\prod_{j\in J} \overline{\node {\nu_{n_j}} {n_j} {\fwd_{q_j}}}
$.  
In detail, all elements in the support of $\Theta$ have the form
\[
\begin{array}{c}
M_{I',I''} = N \q \big| \q 
\prod_{i \in I'}  \node {\nu_{m_i}} {m_i} { \bcastzero v }  \q \big| \q 
\prod_{i \in I''} \node {\nu_{m_i}} {m_i} { \nil }  \q \big| \q 
\prod_{j\in J} \node {\nu_{n_j}} {n_j} {\fwd_{q_j}}
\end{array}
\]
for a suitable partition $I' \cup I'' = I$, and we have
\[
M \TransSim{\hat{\tau}}  \big(\prod_{i \in I'} p_i \prod_{j \in I''}(1-p_j)\big) \cdot  M_{I',I''} .
\]
Consider any $M_{I',I''}$.
If $I' = \emptyset$ then $M_{I',I''}$ cannot simulate $N$ since no forwarder process in node $n_j$ can receive any message.
If $I' \neq \emptyset$, each of the $m_i$ with $i \in I'$ can then broadcasts $v$. Indeed, by applying the rule \rulename{Snd} we get
\[
\node {\nu_{m_h}} {m_h} { \bcastzero v } \q \transSim{\sndto {m_h} v {\nu_{m_h}}}
\q \overline{\node {\nu_{m_h}} {m_h} {\nil}} 
\]
for some $h \in I'$.
Since by the hypothesis we know that all nodes in $\nu_{m_h} \cap \nds{M'}$ cannot receive in the current round, and, analogously, all nodes $n_i$ with $i \in I \setminus \{h\} = (I' \cup I'')\setminus \{h\}$ cannot receive, by applying rule \rulename{RcvEnb} we infer:
\begin{inparaenum}[(i)]
\item 
$N  \transSim{\rcva {m_h} v} \overline{N}$;
\item
$
\node {\nu_{m_i}} {m_i} {\sndv_{p_i}}   \transSim{\rcva {m_h} v}   \overline{\node {\nu_{m_i}} {m_i} {\sndv_{p_i}}}$;
\item 
$\node {\nu_{m_i}} {m_i} {\nil}   \transSim{\rcva {m_h} v}  \overline{ \node {\nu_{m_i}} {m_i} {\nil}}$.
\end{inparaenum}
Moreover, for all $j\in J$, by rule \rulename{Rcv} we infer: 
\[
\node {\nu_{n_j}} {n_j} {\fwd_{q_j}} \q \transSim{\rcva {m_h} v} \q \overline{\node {\nu_{n_j}} {n_j} {\wsndvc_{q_j}}}
\]
thus implying that we can apply rules \rulename{RcvPar} and \rulename{Bcast} and obtain 
\begin{equation*}
\begin{array}{c}
M_{I',I''} \transSimone[\alpha]
\overline{N} \q \big| \q
\prod_{i \in I' \setminus \{h\}}  \overline{\node {\nu_{m_i}} {m_i} { \bcastzero v } } \q \big| \q
\prod_{i \in I'' \cup \{h\}}  \overline{\node {\nu_{m_i}} {m_i} { \nil } } \q \big| \q
\prod_{j\in J} \overline{\node {\nu_{n_j}} {n_j} {\wsndvc_{q_j}}}
\end{array}
\end{equation*}
with $\alpha = \sndto {m_h} v {\nu_{m_h} \setminus \nds M}$.
Since by the hypothesis we have $\nu_{m_h} \subseteq \nds{M}$, we can apply rule \rulename{ShhSnd} to infer
\begin{equation*}
\begin{array}{c}
M_{I',I''} \transSim{\tau} 
\overline{N} \q \big| \q
\prod_{i \in I' \setminus \{h\}}  \overline{\node {\nu_{m_i}} {m_i} { \bcastzero v } } \q \big| \q
\prod_{i \in I'' \cup \{h\}}  \overline{\node {\nu_{m_i}} {m_i} { \nil } }\q  \big| \q
\prod_{j\in J} \overline{\node {\nu_{n_j}} {n_j} {\wsndvc_{q_j}}}
 \enspace .
\end{array}
\end{equation*}
Now we note that if $I' \setminus\{h\} = \emptyset$ then the target of this transition is the network
$O$ and we can derive 
\begin{equation}
\begin{array}{c}
\label{thm:propagation_collisioni_proof_obligation_singola}
M \q \TransSim{ {\tau}}  \q p_i \prod_{j \in I \setminus \{i\}}(1-p_j)  \overline{O} 
\end{array}
\end{equation}
Otherwise, if $I' \setminus\{h\} \neq \emptyset$ then the network $M_{I',I''}$ cannot simulate $O$ since the nodes $\wsndvc_{q_j}$ fail in the present round since they receive from some sender $m_i$ with $i \in I' \setminus \{h\}$.
Overall, we have Equation~\ref{thm:propagation_collisioni_proof_obligation_singola} for all $i \in I$, from which the proof obligation in Equation~\ref{thm:propagation_coll_proof_obligation} follows.
\qed

\end{document}